\documentclass[a4paper, 11pt, oneside,english]{article}        

\usepackage[T1]{fontenc}
\usepackage[utf8]{inputenc}
\usepackage[english]{babel}
\usepackage{textcomp}
\usepackage{tocbibind}
\usepackage{graphicx}
\usepackage{amsmath}
\usepackage{amsthm}
\usepackage[LGR,LGR,T1]{fontenc}
\usepackage{amssymb}
\usepackage{textcomp}
\usepackage{amstext}
\usepackage{graphicx}
\usepackage{setspace}
\usepackage{esint}
\usepackage{caption}
\usepackage{braket}
\usepackage[]{latexsym}
\usepackage{epsfig}
\usepackage{makeidx}
\usepackage{longtable} 
\usepackage{eucal}
\usepackage{subfigure}
\usepackage{latexsym}
\usepackage{caption}
\usepackage{amsfonts}
\usepackage{floatflt}
\usepackage[titletoc]{appendix}
\usepackage{hyperref}
\usepackage{dsfont}
\usepackage{csquotes}
\usepackage[titletoc]{appendix}
\usepackage{stmaryrd}
\usepackage{blindtext}
\usepackage{tikz-cd}
\usepackage{mathrsfs}
\usepackage{authblk} 
\usepackage{sectsty}
\usepackage{etoolbox}
\usepackage{cite}

\numberwithin{equation}{section}

\newcommand{\be}{\begin{equation}}
\newcommand{\ee}{\end{equation}}
\def\beqa{\begin{eqnarray}}
\def\eeqa{\end{eqnarray}}
\def\bean{\begin{eqnarray*}}
\def\eean{\end{eqnarray*}}

\newcommand{\de}{\mathrm{d}}
\newcommand{\eqn}[1]{(\ref{#1})}

\newcommand{\IZ}{\mathbb{Z}}
\newcommand{\IC}{\mathbb{C}}

\newcommand{\IR}{\mathbb{R}}

\newcommand{\frg}{\mathfrak{g}}
\newcommand{\frh}{\mathfrak{h}}

\newcommand{\frd}{\mathfrak{d}}

\def\e{{\,\rm e}\,}

\newcommand{\cM}{{\mathcal M}}

\newcommand{\cH}{{\mathcal H}}

\newcommand{\cE}{{\mathcal E}}

\newcommand{\cQ}{{\mathcal Q}}

\newcommand{\cL}{{\mathcal L}}
\newcommand{\cF}{{\mathcal F}}

\newcommand{\cT}{{\mathcal T}}

\newcommand{\sfD}{{\mathsf{D}}}

\newcommand{\sfH}{{\mathsf{H}}}
\newcommand{\sfG}{{\mathsf{G}}}

\newcommand{\unit}{\mathds{1}}   			

\theoremstyle{definition}
\newtheorem{definition}[equation]{Definition}
\newtheorem{example}[equation]{Example}
\newtheorem{proposition}[equation]{Proposition}
\newtheorem{theorem}[equation]{Theorem}
\newtheorem{remark}[equation]{Remark}
\newtheorem{interlude}[equation]{Interlude}

\title{\textbf{Para-Hermitian Geometry, Dualities \\ and Generalized Flux Backgrounds
\\}
\vspace{0.5cm}}
\date{}

\author[ ]{\large Vincenzo E. Marotta}
\author[ ]{\large ~Richard J. Szabo}

\affil[ ]{}
\affil[ ]{\textit{\normalsize Department of Mathematics \protect\\ Heriot-Watt University \protect\\ Colin Maclaurin Building, Riccarton, Edinburgh EH14 4AS, UK  \protect\\ Maxwell Institute for Mathematical Sciences, Edinburgh, UK  \protect\\  The Higgs Centre for Theoretical Physics, Edinburgh, UK}}

\affil[ ]{}
\affil[ ]{\normalsize e-mail: \texttt{vm34@hw.ac.uk and R.J.Szabo@hw.ac.uk}}

\begin{document}

\maketitle

\vspace{20mm}

\begin{abstract}
\noindent
We survey physical models which capture the main
concepts of double field theory on
para-Hermitian manifolds. We show that the geometric theory of
Lagrangian and Hamiltonian dynamical systems is an instance of
para-K\"ahler geometry which extends to a natural example of a Born
geometry. The corresponding phase space geometry belongs to the family
of natural almost para-K\"ahler structures which we
construct explicitly as deformations of the canonical para-K\"ahler
structure by non-linear connections. We extend this framework to a
class of non-Lagrangian dynamical systems which naturally encodes the notion
of fluxes in para-Hermitian geometry. In this case we describe the
emergence of fluxes in terms of weak integrability
defined by the D-bracket, and we extend the construction to arbitrary
cotangent bundles where we reproduce the standard generalized fluxes of
double field theory. We also describe the para-Hermitian geometry of
Drinfel'd doubles, which gives an explicit illustration of the
interplay between fluxes, D-brackets and different polarizations. The
left-invariant para-Hermitian structure on a Drinfel'd double in a
Manin triple polarization
descends to a doubled twisted torus, which we use to illustrate how
changes of polarizations give rise to different fluxes and string backgrounds
in para-Hermitian geometry.
\end{abstract}

\vfill

\begin{flushright}
\small
{\sf EMPG--18--20}
\normalsize
\end{flushright}

\newpage

\makeatletter
\@starttoc{toc}
\makeatother

\newpage

\section{Introduction and Overview}

Para-Hermitian geometry has acquired renewed interest in recent years
because of its relevance to flux compactifications of string theory, which is inspired by its
connections to generalized geometry and double field theory. Before
describing the new impetus that the present paper provides to this
endeavour, let us set the stage by briefly recalling the
connections between generalized geometry, double field theory and
para-Hermitian geometry.

\paragraph{Generalized Geometry.}
Generalized geometry~\cite{gualtieri:tesi,hitch} is a powerful
mathematical framework in which a unified description of vector fields
and 1-forms is achieved. It is a framework in which dualities, which
typically emerge in physical theories with extra dimensions such as
string theory, can be naturally studied. It is particularly relevant for the
description of T-duality of a string background $(g,B)$ on a
$d$-dimensional target space, where $g$ is the spacetime metric (which
we take to be of Euclidean signature for definiteness) and $B$ is the
Kalb-Ramond field. In the low-energy limit,
string theory is described by supergravity and the bosonic part of the
effective action is given by
\be\label{eq:stringaction}
S_{\rm SUGRA} = \int\,\de^dx\ \sqrt{g} \ \e^{-2\phi}\, \Big({\rm Ric}(g) +
4\,\partial_i\phi\,\partial^i\phi-\frac1{12}\, H_{ijk}\,H^{ijk}\Big) \ ,
\ee
where $H=\de B$ is the NS--NS $H$-flux, $\phi(x)$ is the string
dilaton field and ${\rm Ric}(g)$ is the Ricci scalar of the metric
$g$. The field equations resulting from this action impose vanishing
$\beta$-functions of the background $(g,B)$ and the dilaton $\phi$,
ensuring conformal invariance of the string theory at 1-loop order in
worldsheet perturbation theory.

The action \eqref{eq:stringaction} has a non-manifest ${\sf
  O}(d,d)$-symmetry. In generalized geometry this appears in a generalized tangent bundle over the target
space which has the structure of a Courant algebroid with fiber metric
$\eta$ of signature $(d,d)$ and the Courant bracket of its sections
which are generalized vector fields. The 
generalized metric $\cH$ on a generalized tangent bundle encodes all
information about a given background.
For compactifications on torus fibrations, the ${\sf O}(d,d)$-symmetry includes T-duality which relates the backgrounds of two
non-linear worldsheet sigma-models with different target spaces. In
this case T-duality transformations of the generalized metric appear
as isomorphisms between the generalized tangent bundles of principal
torus bundles endowed with a closed 3-form~\cite{cavalcanti}, viewed
as (twisted) Courant algebroids. In this way the symmetries of string
theory have a natural interpretation in terms of structures emerging
from generalized geometry.

\paragraph{Double Field Theory.}
Dualities in theories with extra dimensions typically indicate the
presence of hidden symmetries of the theory, and it is natural to
search for an extended theory where these become manifest symmetries.
T-duality may be a key to understanding the structures
characterizing the spacetime of a theory with extra dimensions, and in
this case the extended theory should be a manifestly ${\sf
  O}(d,d)$-invariant theory. The extended theory is then formulated on
a doubled target space with generalized coordinates $x^I=(x^i,\tilde
x_i)$, where the spacetime coordinates $x^i$ and their duals $\tilde
x_i$ naturally emerge from string theory where they are fields
appearing in dual actions related by ${\sf
  O}(d,d)$-transformations. In this way T-dualities may play an even
bigger role as part of the diffeomorphisms of such a
space~\cite{geonongeo,tfold}, and this is the starting point in
formulating a theory which is manifestly invariant under T-duality
transformations.

It is clear that generalized geometry is not the framework in which a
doubled target space can be implemented, as it only doubles the fibres
of the tangent bundle over the original spacetime. Only a theory
defined on a doubled geometry, as proposed in~\cite{geonongeo,tfold},
is possible and a T-duality invariant theory called double field
theory has been suggested for this
purpose~\cite{Siegel1993a,Siegel1993b,Hull2009}. It is related to the
low-energy limit of bosonic string theory, since it gives an effective
action which is T-duality invariant only if it is formulated on a
doubled space~\cite{hohmhz,hullzw} and is constructed in the
following way. The doubled space has two metrics, the constant metric
$\eta$ of signature $(d,d)$ which raises and lowers indices in the
following, and the Riemannian metric $\cH$ which
incorporates the dynamical fields $(g,B)$. Starting from the
low-energy effective action \eqref{eq:stringaction}, the action of
double field theory can be written using the ${\sf O}(d,d)$-tensor
$\cH^{-1}$ as
\begin{align}\label{eq:DFTaction}
S_{\rm DFT} = \int\,\de^dx\ \de^d\tilde x \ \e^{-2\phi}\, \Big( &
  \frac18\,\cH^{IJ}\,\partial_I\cH^{KL}\,\partial_J\cH_{KL}-\frac12\,
  \cH^{IJ}\,\partial_J\cH^{KL}\,\partial_L\cH_{IK} \nonumber\\ & -\,2\,\partial_I\phi\,\partial_J\cH^{IJ}+
  4\,\cH^{IJ}\,\partial_I\phi\,\partial_J\phi\Big) \ ,
\end{align}
which is manifestly ${\sf O}(d,d)$-invariant~\cite{hullzw}.

The action \eqref{eq:DFTaction} reduces to \eqref{eq:stringaction}
upon imposing a constraint which halves the number of coordinates;
this is called the \emph{section condition}. In particular, we can consider
the level matching condition $L_0-\bar L_0=0$ of the worldsheet theory, which in the doubled
formalism becomes $\partial_M\,\partial^M=0$ on any field and any
parameter; this is called the strong constraint and one of its
solutions is obtained by taking all fields to be independent of the
dual coordinates $\tilde x_i$, in which case we recover
\eqref{eq:stringaction}.

The action \eqref{eq:DFTaction} is also invariant under
diffeomorphisms which are generated by doubled vector fields and a
suitable notion of generalized Lie derivative. The generalized Lie
derivative of the metric $\eta$ vanishes, which implies that the
constraint for $\cH$ to be an ${\sf O}(d,d)$-tensor is compatible with
diffeomorphisms, while the generalized Lie derivative of the identity
$\cH\,\eta\,\cH=\eta^{-1}$ gives the compatibility between the ${\sf
  O}(d,d)$ and gauge symmetries. The ${\sf O}(d,d)$-covariant
extension of the Courant bracket for doubled vector fields is called
the \emph{C-bracket} and is given by
\be\label{eq:Cbracketflat}
\big(\llbracket \xi_1,\xi_2 \rrbracket ^{\tt C}\big)^I =
\xi_{[1}^J\,\partial_J\xi_{2]}^I - \frac12\,
\xi_{[1}^J\,\partial^I\xi_{2]\,J} \ .
\ee
The C-bracket governs the algebra of generalized Lie
derivatives. Although this theory still lacks a global formulation, it
suggests a close relationship with generalized geometry, at least
locally, since the Courant algebroid structure of the generalized
tangent bundle is recovered upon imposing the section condition.

Having the algebra of gauge transformations at hand, one can now write the
gauge-invariant action in terms of $\cH$. This action is manifestly
${\sf O}(d,d)$-invariant and is expressed in an Einstein-Hilbert type
form~\cite{hohmhz,hullzw}, where the generalized scalar curvature is
defined as the field equation for the dilaton following from the
action \eqref{eq:DFTaction}. The equations of motion combine the field
equations for the metric $g$ and the $B$-field in an ${\sf
  O}(d,d)$-covariant form, which has the form of the Einstein equation
extended to the doubled formalism. In this sense double field theory
can be regarded as a low-energy effective field theory of string generalized geometry. For more detailed reviews of double
field theory, see~\cite{Aldazabal2013,Berman2013,Hohm2013}. We stress that this is only a
local formulation of the theory, and the problem is to find a suitable geometry in
which all these structures can be defined globally.

\paragraph{Para-Hermitian Geometry.}
The discussion thus far stresses the need of a mathematical framework
in which a global formulation of doubled geometry is possible, while
at the same recovering a precise relation with generalized
geometry. A promising approach is achieved by \emph{para-Hermitian geometry}, which was first proposed in the context of a global
formulation of double field
theory in~\cite{Vai2,Vai1}, and developed further
by~\cite{freidel,Borngeometry} to provide a global formulation of its
kinematics. A unified approach was presented
in~\cite{membrane,svoboda} where a further generalization of the notion
of Courant algebroid is given in order to encode the desired features
of a doubled target space. However, the first appearance of
para-Hermitian geometry in the description of T-duality can be traced
back to the work of~\cite{geonongeo}, where a T-fold is
described in terms of doubled torus bundles in which para-Hermitian structures
are defined on the fibers. We review this approach in some precise detail in
Section~\ref{sec:ParaHermitianGeometry} below, and present here a
brief overview of the relevant aspects while glossing over many
technicalities. For a more complete
survey of the mathematical aspects of para-Hermitian geometry,
see~\cite{Cruceanu1996}.\footnote{An alternative mathematical approach to double field
theory based on graded symplectic geometry is discussed in~\cite{Deser2014,Deser2016,Heller2016}.} 

Para-Hermitian geometry is, roughly speaking, formulated as a real version of the more familiar concepts from
complex, Hermitian and K\"ahler geometry. 
A \emph{para-complex structure} on a vector bundle $E\to M$ of even
rank $2d$ over a manifold $M$ is a bundle endomorphism
$K\in{\rm End}(E)$ such that $K^2=\mathds{1}$ and the
$\pm\,1$-eigenbundles of $K$ have the same rank $d$. A symmetric
non-degenerate pairing $\eta$ of sections of $E$ is called
{para-Hermitian} if $\eta(K(X),K(Y))=-\eta(X,Y)$ for all
$X,Y\in\Gamma(E)$, and the pair $(K,\eta)$ is then called a
\emph{para-Hermitian structure} on $E$. A \emph{para-Hermitian
  manifold} is a manifold $M$ whose tangent bundle carries a
para-Hermitian structure; the compatible (almost) para-complex
structure $K$ and Lorentzian metric $\eta$ naturally give rise to a
fundamental 2-form $\omega$ on $M$, and if $\omega$ is closed then $(M,K,\eta)$
is called a \emph{para-K\"ahler manifold}. 

In the framework for double field theory on para-Hermitian manifolds, the original
spacetime $F$ is regarded as a submanifold of the doubled space $M$ when one of the eigenbundles of the para-complex structure is integrable.
The generalized geometry perspective is recovered as the tangent
bundle $TM$, whose metric $\eta$ defines a bundle isomorphism from $TM$ to the generalized tangent bundle $TF\oplus T^*F$
on the integral foliation $F$ of the integrable distribution. 
{Para-Hermitian connections} are defined as connections on $TM$
which preserve both $K$ and $\eta$, and in particular the parallel
transport of sections of the eigenbundles; the Levi-Civita connection of the ${\sf O}(d,d)$-metric $\eta$ is a
para-Hermitian connection only on para-K\"ahler manifolds.
One then defines the D-bracket of vector fields on $M$ in
terms of a canonical para-Hermitian connection which derives from the
Levi-Civita connection, and its skew symmetrization gives the
C-bracket which has the local expression \eqref{eq:Cbracketflat} when
the para-Hermitian structure is flat. The bundle isomorphism above
then sends the C-bracket on $TM$ to the Courant bracket of the exact
Courant algebroid on 
$TF\oplus T^*F$~\cite{freidel,svoboda,Borngeometry}. This result clarifies the
differences between generalized geometry and doubled geometry, and
states a precise relation between them. 

Different dual spacetimes
are obtained as different para-complex structures, i.e. \emph{polarizations}, on
the tangent bundle $TM$. In particular, suitable spacetimes are
Lagrangian submanifolds of the doubled space $M$ which, when endowed
with a generalized metric $\cH$ encoding the data of the background
fields, gives rise to the notion of a \emph{Born
  geometry}~\cite{reciproc,metastring}. Futhermore, a first appearence
of geometric and
non-geometric fluxes appears in this formalism through deformations of para-Hermitian
structures by $B$-transformations~\cite{svoboda}, which are
endomorphisms of $TM$ that preserve the metric $\eta$ but twist the
almost para-complex structure, the fundamental 2-form, and the
D-bracket. These constructions thus far fulfill the
requirements of a global formulation for the kinematics of double
field theory, whereas a global description of its dynamics is still
lacking.

\paragraph{Overview of Results and Outline.}
This paper is concerned with the description of physical spacetimes,
their backgrounds and their duality transformations in the setting of
para-Hermitian geometry; we set the stage for this in
Section~\ref{sec:ParaHermitianGeometry} by first briefly reviewing aspects
of para-Hermitian geometry, following
\cite{freidel,svoboda,Borngeometry} for the most part, and elucidating the
discussion to
describe global aspects of polarizations and T-duality in this framework. Our aim is to focus on some explicit classes
of examples in which (almost) para-Hermitian structures naturally
arise, and from them extract features of duality transformations and
the fluxes which characterize each polarization. Our examples are mostly
known already in the literature, though perhaps less widely in the
context of para-Hermitian geometry, which is the geometric impetus
that we emphasise through the physics of these
examples. Through these explicit settings we can then extract general
properties of dualities and how generalized fluxes appear on para-Hermitian
manifolds, and understand better their global physical features.

As a first class of examples, we show that the Lagrangian (and
Hamiltonian) description of classical dynamical systems leads to a
natural para-K\"ahler structure on the tangent bundle $M=T\cQ$ of the
underlying configuration space $\cQ$. Together with the dynamical
almost Hermitian structure determined by the Lagrangian function, the
para-K\"ahler structure gives a natural instance of Born geometry. The
Finsler geometry of regular Lagrangian dynamical systems is also
connected to double field theory and generalized geometry
in~\cite{Vai1,Vaisman2013}, whereas in this paper we emphasize their
geometric aspects in terms of para-K\"ahler structures. We describe
the Legendre transform to the cotangent bundle $T^*\cQ$ as a member of
the most general family of para-K\"ahler structures on phase space,
which are obtained as deformations of the canonical para-K\"ahler
structure by symmetric $(0,2)$-tensor fields $C$ on $T^*\cQ$ which
are coefficients of torsion-free non-linear connections on the tangent
bundle $T(T^*\cQ)$. This
generalizes the constructions of~\cite{Vai1,Romaniuc2013,Vaisman2013}
in which $C$ is taken to be a natural lift to $T^*\cQ$ of the Christoffel symbols
of the Levi-Civita connection of a Riemannian metric $g$ on the configuration manifold $\cQ$. The
para-K\"ahler geometry underlying classical dynamical systems in
both the Lagrangian and Hamiltonian formalisms is described in Section~\ref{sec:dynamicalpara}.

We then turn to the extension of this construction to cases where a
Lagrangian (and Hamiltonian) is not globally defined. We consider a
prototypical class of examples which includes, as specific cases,
the motion of charged particles in magnetic fields generated by
distributions of magnetic monopoles and the motion of closed strings in
locally non-geometric $R$-flux backgrounds~\cite{lustcoc,kup}. We provide an interpretation
of these dynamical systems in terms of deformations of para-K\"ahler
structures in order to explain how fluxes emerge in this case, which also
lends a novel perspective to their inherent nonassociativity. In
particular, we can recover the almost symplectic 2-form
$\omega_B$ defining the twisted Poisson brackets of the dynamical
system as the fundamental 2-form
$\omega_B=\eta_0\,K_B=\eta_0\,K_0+2\,\eta_0\,B$ of the almost para-Hermitian
structure $(K_B,\eta_0)$ on $T(T^*\cQ)$, where $K_B$ determines the
splitting $T(T^*\cQ)=L_{\tt v}\oplus L_{\tt h}^B$ with\footnote{Notation:
  $Q^i=\frac{\partial}{\partial p_i}$ and $P_i=\frac\partial{\partial
      q^i}$, where $(q^i,p_i)$ are local Darboux coordinates on
    $T^*\cQ$. We write $\partial_if=P_i(f)$ and $\tilde\partial^if=Q^i(f)$
for any function $f\in C^\infty(T^*\cQ)$.} $L_{\tt v}={\rm Span}_{C^\infty(T^*\cQ)}\{Q^i\}$ and $L_{\tt h}^B={\rm
    Span}_{C^\infty(T^*\cQ)}\{D_i=P_i+B_{ij}\, Q^j\}$, and $\eta_0$ is a flat metric of
  Lorentzian signature; it follows that $L_{\tt v}$ is an integrable
  distribution while $L_{\tt h}^B$ is not. The almost para-complex
  structure $K_B$ can be regarded as a $B$-transformation of the
  para-complex structure $K_0$ with integrable eigenbundles $L_{\tt
    v}$ and $L_{\tt h}^0={\rm
    Span}_{C^\infty(T^*\cQ)}\{P_i\}$, which illustrates the general
  feature that a $B$-transformation does not preserve
  integrability. 

If the endomorphism $B$ depends only on the
  configuration space coordinates $q$, then we show that the only non-vanishing
  D-bracket with respect to $K_B$ is 
\be\nonumber
\llbracket D_i,D_j\rrbracket^{\tt D}_B =
\frac32\, \partial_{[i}B_{jk]}\ Q^k \ .
\ee
This bracket is thus related
to the $H$-flux and is an element of $\Gamma(L_{\tt v})$ (the precise
relation can be found in~\cite{svoboda} and Section~\ref{sec:Btransformations} below), which means
that the $H$-flux obstructs the integrability of the $B$-transformed
distribution and a foliated manifold with local momentum coordinates
$p$ does not exist. Assuming more generally that $B$ may also depend 
on the fiber coordinates $p$, we show that the D-bracket becomes
\be\nonumber
\llbracket D_i,D_j\rrbracket^{\tt D}_B =
\frac32\, \big((\partial_{[i}B_{jk]} +
B_{[im}\,\tilde\partial^mB_{jk]})\ Q^k + \tilde\partial^kB_{ij}\ D_k\big) \ ,
\ee
where now integrability is obstructed by the covariant $H$-flux
$\mathscr{H}_{ijk}=\partial_{[i}B_{jk]}+B_{[i|m}\,\tilde{\partial}^m
B_{|jk]}$, which is the $\Gamma(L_{\tt v})$-component of the
D-bracket. 

This construction is then extended to give the full set of
geometric and non-geometric fluxes of (local) double field theory,
extending the considerations of~\cite{svoboda}.
In particular, we
recast Born reciprocity in terms of these deformations and extend
the discussion to describe generalized fluxes as deformations of the
para-K\"ahler geometry of the cotangent bundle $T^*\cQ$ involving both $B$- and
$\beta$-transformations. This gives an alternative perspective in the
setting of para-Hermitian geometry to closed string noncommutativity
and nonassociativity in non-geometric flux backgrounds
(see~\cite{Szabo2018} for a review). The appearence of
nonassociativity in this way from changes of polarizations and flux
deformations of the canonical para-K\"ahler structure of closed string
zero modes was also discussed in~\cite{Freidel2017}, and our complimentary detailed construction further elucidates their meaning in terms of violations of weak integrability. The role of polarizations of the para-Hermitian geometry on the cotangent bundle $T^*\cQ$ was also emphasised by~\cite{Aschieri2015,Aschieri2017} in relating phase space and spacetime nonassociativity. The description of fluxes in non-Lagrangian
dynamical systems and in the dual $R$-flux model is the subject of Section~\ref{sec:nonassociativity}.

We finally consider the broad classes of non-trivial doubled manifolds
provided by Drinfel'd doubles and doubled twisted tori, in which
para-K\"ahler structures can be naturally defined, and describe their
different polarizations as duality transformations of these
structures along with the related fluxes. In
Section~\ref{paradrinfeld} we recall the well-known para-Hermitian
geometry of Drinfel'd doubles and compute the corresponding
D-brackets to illustrate how fluxes arise. Drinfel'd doubles extend
the case of cotangent bundles and provide non-abelian generalizations
of the standard flat para-K\"ahler doubled tori $T^{2d}$, while their
different Manin triple polarizations generalize the different ways of
embedding $T^d\subset T^{2d}$ which are related by the action of the
T-duality group ${\sf O}(d,d;\IZ)$. The para-Hermitian geometry of
Drinfel'd doubles treats non-abelian T-duality using ${\sf
  O}(d,d)$-type structures as in~\cite{osten}, whereas more general
polarizations than those provided by Manin triples define a modified
non-abelian T-duality group and enable the introduction of generalized
fluxes. In
Section~\ref{sec:doubledtorus} we apply the framework of
para-Hermitian geometry to describe doubled twisted tori, which
comprise a class of well-known examples of global non-trivial doubled geometry,
and analyse their different polarizations together with their
backgrounds. This gives a more intrinsic
perspective on the T-duality transformations relating almost para-Hermitian
structures on the doubled twisted torus which are discussed in~\cite{hulled}.

\section{Para-Hermitian Geometry, Polarizations and T-Duality\label{sec:ParaHermitianGeometry}}

${\sf O}(d,d)$-transformations of supergravity can be described within the mathematical framework of generalized geometry. They include T-dualities between string backgrounds, which live in the discrete subgroup ${\sf O}(d,d;\mathbb{Z})$, and hence they must be included in the geometric structure of an ${\sf O}(d,d)$-invariant theory. 
The goal of double field theory is the description of an effective theory which manifestly possesses this invariance. Doubled geometry is the proposed mathematical framework in which such a theory should be formulated, but it currently lacks a global description, despite bearing similarities with generalized geometry~\cite{geonongeo,tfold}.
As first proposed in \cite{Vai2,Vai1}, para-Hermitian geometry may provide the natural framework in which a global formulation of double field theory can be achieved, since it has an intimate connection with generalized geometry; in this setting ``doubled spacetime'' is synonymous with ``para-Hermitian manifold''. 
In this section we will briefly review aspects of para-Hermitian geometry that we need in this paper, following~\cite{freidel,svoboda,Borngeometry} where a description of the foundations of double field theory is provided, as well as a first instance of how to introduce fluxes into this framework. 

Para-Hermitian geometry first played a central role in the discussion
of T-duality in~\cite{geonongeo}, where para-Hermitian structures are defined on the fibers of a
doubled torus bundle $M\to W$ to provide the first geometric definition of a T-fold; the various
polarizations of the $T^{2d}$ fibers give different T-dual backgrounds. The almost product structures
associated to each T-dual polarization are related by ${\sf O}(d, d)$-transformations. In the setting
of para-Hermitian geometry, we will clarify below in which sense a
physical spacetime can only be recovered locally for a globally non-geometric background by relating it to the integrability of the distribution
associated to it as a choice of the almost product structure in a general global formulation. Locally non-geometric backgrounds on the other hand, where not even a local spacetime description is possible, are particular polarizations of a doubled twisted torus in which the base space $W\times \tilde W$ is also doubled~\cite{Dabholkar2005} and will be characterized globally in the following in terms of non-integrable distributions. We return to the specific examples of doubled twisted tori using this general formalism in Section~\ref{sec:doubledtorus}.

\subsection{Para-Hermitian Manifolds}

Throughout this paper all manifolds are assumed to be smooth.
\theoremstyle{definition}
\begin{definition}
An \emph{almost product structure} on a manifold $M$ is a $(1,1)$-tensor field $K \in \mathrm{End}(TM)$ such that $K^2=\mathds{1}$. The pair $(M,K)$ is an \emph{almost product manifold}.
\end{definition}
Fixing a coordinate chart $(U, \phi)$ on $M$ with local coordinates $x^i$, an almost product
structure can be written as $K=K_i{}^j\, \de x^i
\otimes\frac{\partial}{\partial x^j}$ on $U$ with $K_i{}^m\, K_m{}^j=\delta_i{}^j$.
In this definition the analogy with almost \emph{complex} manifolds is clear, i.e. even-dimensional manifolds endowed with a $(1,1)$-tensor field $J$ such that $J^2=-\mathds{1}.$ This analogy is a useful guide to understanding the structures introduced in the following and will be often recalled by the terminology we adopt.

\theoremstyle{definition}
\begin{definition}
An \emph{almost para-complex manifold} is an almost product manifold $(M,K)$ with $M$ of even dimension such that the two eigenbundles $L_+$ and $L_-$ associated, respectively, with the eigenvalues $+1$ and $-1$ of $K$ have the same rank. A splitting of the tangent bundle $TM=L_+\oplus L_-$ of a manifold $M$ into the Whitney sum of two sub-bundles $L_+$ and $L_-$ of the same fiber dimension is an \emph{almost para-complex structure} on $M$.
\end{definition}
Recall that a \emph{$\sf G$-structure} on a $2d$-dimensional manifold $M$, for a subgroup ${\sf G} \subset {\sf GL}(2d,\mathbb{R})$, is a $\sf G$-sub-bundle of the frame bundle $FM$, i.e a reduction on the frame bundle of the structure group ${\sf GL}(2d,\mathbb{R})$ to $\sf G$. Using this notion, we can rephrase the definition of almost para-complex structure by saying that it is a $\sf G$-structure on $M$ with structure group ${\sf G}={\sf GL}(d,\mathbb{R})\times {\sf GL}(d,\mathbb{R})$.  

Using the almost product structure, we can define two projection operators
\begin{align*}
\Pi_+&=\frac{1}{2}\,(\mathds{1}+K): \Gamma(TM) \longrightarrow \Gamma(L_+) \ , \\[4pt]
\Pi_-&= \frac{1}{2}\,(\mathds{1}-K): \Gamma(TM) \longrightarrow \Gamma(L_-) \ .
\end{align*}
Then we are naturally led to study the integrability of the sub-bundles $L_+$ and $L_-$.
\theoremstyle{definition}
\begin{definition}
An almost product structure $K$ is \emph{(Frobenius) integrable} if
its sub-bundles $L_+$ and $L_-$ are both integrable,
i.e. $[\Gamma(L_+),\Gamma(L_+)]\subseteq\Gamma(L_+)$ and $[\Gamma(L_-),\Gamma(L_-)]\subseteq\Gamma(L_-)$. An integrable almost product structure is a \emph{product structure}. A \emph{para-complex structure} is an integrable almost para-complex structure, i.e. a product structure with ${\rm rank}(L_+)={\rm rank}(L_-)$. 
\end{definition}
By Frobenius' Theorem, this means that the manifold $M$ admits two foliations $\mathcal{F}_+$ and $\mathcal{F}_-$, such that $L_+=T\mathcal{F}_+$ and $L_-=T\mathcal{F}_-$.

Another way to characterize the integrability of an almost product structure is through the definition of the Nijenhuis tensor field, continuing the analogy with almost complex structures.
\theoremstyle{definition}
\begin{definition}
The \emph{Nijenhuis tensor field} of an almost product structure $K$ is the map $N_K: \Gamma(TM)\times \Gamma(TM) \rightarrow \Gamma(TM)$ given by
\begin{equation*}
N_K(X,Y)=[X,Y]+[K(X),K(Y)] -K\big([K(X),Y] + [X,K(Y)]\big)\ ,
\end{equation*}
for all vector fields $ X,Y \in \Gamma(TM).$
\end{definition}
\theoremstyle{theorem} 
\begin{theorem}
An almost product structure $K$ on a manifold $M$ is integrable if and only if $N_K(X,Y)=0$ for all $X,Y \in \Gamma(TM).$ 
\end{theorem}

Using the projection tensors $\Pi_+$ and $\Pi_-$, together with $K=\Pi_+-\Pi_-$, we can decompose the Nijenhuis tensor as 
\begin{equation*}
N_K(X,Y)=N_{\Pi_+}(X,Y)+N_{\Pi_-}(X,Y) \ ,
\end{equation*}
where 
\begin{equation}
N_{\Pi_+}(X,Y)=\Pi_-\big([\Pi_+(X),\Pi_+(Y)]\big) \qquad \mbox{and} \qquad N_{\Pi_-}(X,Y)=\Pi_+\big([\Pi_-(X),\Pi_-(Y)]\big) \ .
\label{parnij}
\end{equation}
From \eqref{parnij} we evidently have $N_{\Pi_+}(X,Y)\in \Gamma(L_-)$ and $N_{\Pi_-}(X,Y)\in \Gamma(L_+)$. Hence the two components of the Nijenhuis tensor obstruct the closure of the Lie bracket of vector fields restricted to $L_+$ and $L_-$, respectively. 
In particular, $N_{\Pi_+}$ and $N_{\Pi_-}$ are independent of each other. For instance, we may have $N_{\Pi_+}(X,Y)=0$ and $N_{\Pi_-}(X,Y) \neq 0$, so that the almost para-complex structure is only partially integrable ($N_K(X,Y)$ is still non-vanishing) and it admits only one foliation $\mathcal{F}_+$ such that $L_+=T\mathcal{F}_+$. In this case we call $(M,K)$ an \emph{$L_+$-para-complex manifold}, i.e.~there is a splitting of the tangent bundle $TM$ into two distributions with the same rank such that only the eigenbundle associated to the eigenvalue $+1$ is integrable. Exchanging the roles of $N_{\Pi_+}$ and $N_{\Pi_-}$, we obtain an analogous situation in which the sub-bundle $L_-$ is integrable, i.e. it admits a foliation $\mathcal{F}_-$ such that $L_-=T\mathcal{F}_-$, and in this case we call $(M,K)$ an \emph{$L_-$-para-complex manifold}.

Following the analogy with complex geometry, we introduce a compatible metric on almost para-complex manifolds, giving a counterpart of almost Hermitian manifolds.
\theoremstyle{definition} 
\begin{definition}
An \emph{almost para-Hermitian manifold} $(M,K,\eta)$ is an almost para-complex manifold, i.e. a manifold $M$ of even dimension $2d$ endowed with a $(1,1)$-tensor field $K \in \mathrm{End}(TM)$ such that $K^2=\mathds{1}$, together with a metric $\eta$ of Lorentzian signature $(d,d)$ which is compatible with the tensor $K$ in the sense that $$K^{\rm t}\, \eta\, K =-\eta \ .$$  
\end{definition}
The compatibility condition can also be written as $$\eta\big(K(X),K(Y)\big)=-\eta(X,Y) \ , $$ or equivalently 
\begin{equation}
\eta\big(K(X),Y\big)+\eta\big(X,K(Y)\big)=0 \ , 
\label{compcon}
\end{equation}
for all $X,Y \in \Gamma(TM)$.
The condition \eqref{compcon} implies that the distributions $L_+$ and $L_-$ are maximally isotropic with respect to $\eta$: For any $ X_+,Y_+ \in \Gamma(L_+),$ we have $K(X_+)=X_+$ and $K(Y_+)=Y_+$, and \eqref{compcon} gives $\eta(X_+,Y_+)=0$, i.e. $L_+$ is isotropic and, since ${\rm rank}(L_+)=d$, it is a maximally isotropic sub-bundle of $TM$. The same argument applies to $L_-$.

From \eqref{compcon} we also deduce the existence of a non-degenerate 2-form field $\omega$ on $M$ given by $$\omega(X,Y)=\eta\big(K(X),Y\big) \ , $$ for all $X,Y \in \Gamma(TM)$, called the \emph{fundamental 2-form}; it defines an \emph{almost symplectic structure}, since it is generally not closed. Because of this definition, we have
\begin{equation}
\omega(X_+,Y_+)=0 \ , 
\label{ome1}
\end{equation}
for all $X_+,Y_+ \in \Gamma(L_+)$, and 
\begin{equation}
\omega(X_-,Y_-)=0 \ , 
\label{ome2}
\end{equation}
for all $X_-,Y_- \in \Gamma(L_-)$. If the fundamental 2-form $\omega$ is \emph{symplectic}, i.e. it is moreover closed: ${\rm d}\omega=0$, then  $(M,K,\eta)$ is called an \emph{almost para-K\"{a}hler manifold}. In this case, the conditions \eqref{ome1} and \eqref{ome2} imply that $L_+$ and $L_-$ are \emph{Lagrangian} sub-bundles.

An \emph{almost para-Hermitian structure} $(K,\eta)$ on a manifold $M$
can be regarded as a $\sf G$-structure on $M$ given by a reduction of
the structure group of $TM$ from ${\sf GL}(2d,\mathbb{R})$ to the
subgroup which preserves both $\eta$ and $\omega$: $${\sf G}={\sf
  O}(d,d)\cap {\sf Sp}(2d,\mathbb{R})={\sf GL}(d,\mathbb{R}) \ . $$

\begin{interlude}\label{int:Odd}
We denote by ${\sf O}(d,d)(M)$ 
the infinite-dimensional pseudo-orthogonal group of tangent bundle
automorphisms $\vartheta\in{\rm End}(TM)$ which preserve the Lorentzian metric:
$\eta(\vartheta (X),\vartheta (Y))=\eta(X,Y)$ for all
$X,Y\in\Gamma(TM)$. This is the natural group of isometries of the almost para-Hermitian
manifold $(M,K,\eta)$ which is identified as its
\emph{continuous T-duality group}; any element
$\vartheta\in{\sf O}(d,d)(M)$ can be regarded as a smooth
map $\vartheta:M\to{\sf O}(d,d)$. We denote by
${\sf SO}(d,d)(M)$ the Lie subgroup which also preserves the
canonical orientation of $M$ provided by its
fundamental 2-form $\omega$; its Lie algebra $\mathfrak{so}(d,d)(M)$ consists of
endomorphisms $\tau\in{\rm End}(TM)$ such that $\eta(\tau
(X),Y)=-\eta(X,\tau (Y))$ for all $X,Y\in\Gamma(TM)$. Any element
$\tau\in\mathfrak{so}(d,d)(M)$ can be decomposed with respect to the
splitting $TM=L_+\oplus L_-$ as
\be\nonumber
\tau = \bigg(
\begin{matrix}
A & B_- \\
B_+ & -A^{\rm t}
\end{matrix} \bigg) \ ,
\ee
where $A\in{\rm End}(L_+)$ with transpose $A^{\rm t}\in{\rm End}(L_-)$
defined through $\eta(A(X),Y)=\eta(X,A^{\rm t}(Y))$, while $B_+:\Gamma(L_+)\to
\Gamma(L_-)$ and $B_-:\Gamma(L_-)\to \Gamma(L_+)$ are skew morphisms in the sense that
$\eta(B_\pm(X),Y)=-\eta(X,B_\pm(Y))$. By identifying $L_-$ with
$L_+^*$ using the Lorentzian metric $\eta$, we can regard $B_+$ as a 2-form
$B\in\bigwedge^2 \Gamma(L_+^*)$ called a \emph{$B$-field}
and $B_-$ as a bivector $\beta\in\bigwedge^2 \Gamma(L_+)$, so that as a
vector space
\be\nonumber
\mathfrak{so}(d,d)(M) = {\rm End}(L_+) \oplus
\mbox{$\bigwedge^2$}\, \Gamma(L_+^*) \oplus
\mbox{$\bigwedge^2$}\, \Gamma(L_+) \ .
\ee
\end{interlude}

Integrability of an almost para-Hermitian structure can be described
as well. If the eigenbundles $L_+$ and $L_-$ of $K$, such that
$TM=L_+\oplus L_-,$ are both integrable then the triple $(M,K,\eta)$
is called a \emph{para-Hermitian manifold}. If in addition the
fundamental 2-form $\omega=\eta\, K$ is closed, then $(M,K,\eta)$ is
said to be a \emph{para-K\"{a}hler} (or \emph{bi-Lagrangian})
\emph{manifold}, in which case it has two transverse Lagrangian
foliations with respect to the symplectic structure $\omega.$  

In this framework, we can also describe partial integrability of the sub-bundles of $TM$ and partial closure of the fundamental 2-form $\omega$ (i.e. $\mathrm{d}\omega=0$ on $L_+$ or $L_-$), giving rise to an assortment of new possible combinations.
In particular, if $(M,K,\eta)$ is an almost para-Hermitian manifold with an integrable sub-bundle $L_+$, then it is said to be an \emph{$L_+$-para-Hermitian manifold}; there is an analogous definition replacing $L_+$ with $L_-$. In general, it can be shown that $$\mathrm{d}\omega\big(\Pi_+(X),\Pi_+(Y),\Pi_+(Z)\big)= \sum_{(X,Y,Z)}\, \eta\big(N_{\Pi_+}(X,Y),Z\big) \ ,$$ where the sum runs over all cyclic permutations $(X,Y,Z)$ of the three vector fields $X,Y,Z\in\Gamma(TM)$. This shows that non-integrability of the distribution $L_+$ obstructs the closure of $\omega$; in particular, integrability of $L_+$ implies that the pullback of the 3-form ${\rm d}\omega$ to the foliation $\mathcal{F}_+$ vanishes.
In order to describe the geometry of such structures, a suitable connection is needed.
\theoremstyle{definition}
\begin{definition}
A \emph{para-Hermitian connection} $\nabla$ on an almost para-Hermitian manifold $(M,K,\eta)$ is a connection on $TM$ preserving $\eta$ and $\omega$: $$\nabla \eta=\nabla \omega =0 \ .$$
\end{definition}

\theoremstyle{proposition} 
\begin{proposition}
Let $(M,K,\eta)$ be an almost para-Hermitian manifold with fundamental 2-form $\omega$, and let $\nabla^{\tt LC}$ be the Levi-Civita connection of $\eta$. Then the covariant Levi-Civita derivative of $\omega$ satisfies
\begin{align}
{\nabla}^{\tt LC}_X\omega\big(\Pi_+(Y),\Pi_-(Z)\big)&= 0 \ , \label{propr1} \\[4pt]
{\nabla}^{\tt LC}_X \omega (Y,Z)&=\eta\big( {\nabla}^{\tt LC}_X K(Y),Z\big) \ , \label{propr2} \\[4pt]
\mathrm{d}\omega (X,Y,Z)&= \sum_{(X,Y,Z)}\, {\nabla}^{\tt LC}_X \omega(Y,Z) \ , \label{propr3}
\end{align}
for all $ X,Y,Z \in \Gamma(TM).$
\end{proposition}
The proof can be found in \cite{freidel}. This proposition stresses
the role of the Levi-Civita connection $\nabla^{\tt LC}$ in
para-Hermitian geometry. In particular, it implies that $\nabla^{\tt
  LC}$ is a para-Hermitian connection if and only if $(M,K,\eta)$ is
an almost para-K\"ahler manifold.\footnote{In this sense,
  para-K\"ahler manifolds provide the closest versions of the
  conventional flat space formulations of double field theory in the absence of fluxes.}

On any  almost para-Hermitian manifold $(M,K,\eta),$ the \emph{canonical para-Hermitian connection} is obtained from the Levi-Civita connection as $$\nabla^{\tt can}=\Pi_+\,{\nabla}^{\tt LC}\,\Pi_+ + \Pi_-\, {\nabla}^{\tt LC}\,\Pi_- \ .$$ Equivalently, the canonical connection $\nabla^{\tt can}$ is defined as the connection satisfying $$\eta(\nabla^{\tt can}_X Y,Z)= \eta({\nabla}^{\tt LC}_X Y, Z)- \frac{1}{2}\, {\nabla}^{\tt LC}_X \omega\big(Y, K(Z)\big) \ ,$$ for all vector fields $X,Y,Z$.
This connection is the key ingredient for the study of the D-bracket, which we introduce below since it gives a new interpretation of the generalized fluxes in double field theory~\cite{svoboda}.

\begin{interlude}
The splitting of the tangent bundle $TM$ gives rise to a decomposition of tensors analogous to the type decomposition in complex geometry. In particular, there is such a decomposition for differential forms. We denote $\bigwedge ^{(+k, -0)}\Gamma(T^* M)= \bigwedge^k \Gamma(L_+^*)$ and $\bigwedge ^{(+0, -k)}\Gamma(T^* M)= \bigwedge^k \Gamma(L_-^*),$ so that any $k$-form on $M$ is decomposed according to the splitting
$$ \mbox{$\bigwedge^k$}\, \Gamma(T^*M)= \underset{m+n=k}{\bigoplus} \
\mbox{$\bigwedge^{(+m,-n)}$}\, \Gamma(T^*M) \ .$$
The fundamental 2-form $\omega$ of an almost para-Hermitian manifold is a $(+1,-1)$-form with respect to the almost para-Hermitian structure $(K,\eta),$ since both $L_+$ and $L_-$ are Lagrangian with respect to $\omega.$
\end{interlude}

\subsection{Brackets and Algebroids}

As discussed in \cite{freidel,svoboda}, the D-bracket is needed to formulate a precise relation between $L_{\pm}$-para-Hermitian manifolds and (standard) Courant algebroids. The definition of a particular D-bracket also gives a global formulation of the D-bracket of double field theory.

Let us first describe how a bracket on vector fields can be associated to any connection on an almost para-Hermitian manifold.
\theoremstyle{definition}
\begin{definition}
Let $(M,K,\eta)$ be an almost para-Hermitian manifold and $\nabla$ a connection on $TM$. The \emph{bracket $\llbracket\cdot,\cdot\rrbracket^\nabla$ associated to} $\nabla$ is defined by
\be
\eta(\llbracket X,Y \rrbracket ^{\nabla},Z)=\eta(\nabla_X Y - \nabla_Y X, Z)+\eta(\nabla_Z X,Y) \ , \label{combra}
\ee
and the \emph{$\Pi_{\pm}$-projected brackets $\llbracket\cdot,\cdot\rrbracket_\pm^\nabla$ associated to $\nabla$} by
\be
\eta(\llbracket X,Y \rrbracket ^{\nabla}_{\pm},Z)=\eta(\nabla_{\Pi_{\pm}(X)} Y - \nabla_{\Pi_{\pm}(Y)} X , Z)+\eta(\nabla_{\Pi_{\pm}(Z)} X,Y) \ ,
\label{probra}
\ee
for all $X,Y,Z \in \Gamma(TM).$
\end{definition}
Since $\Pi_++\Pi_-=\mathds{1}$, it follows that $$\llbracket X,Y \rrbracket ^{\nabla} = \llbracket X,Y \rrbracket_+^{\nabla}+\llbracket X,Y \rrbracket_-^{\nabla}$$ for all $X,Y \in \Gamma(TM).$
The second term on the right-hand side in both \eqref{combra} and \eqref{probra} is usually called a \emph{twist}. In particular, for a torsion-free connection $\nabla$, the associated bracket is given by the Lie bracket $[X,Y]$ of vector fields plus a twist term determined by the connection itself.
The covariant derivative $\nabla_{\Pi_{\pm}(X)}$ is also called a \emph{projected derivative} and it gives a \emph{projected differential} $\mathrm{D}:C^\infty(M)\to \Gamma(TM)$ defined by $\eta(X,\mathrm{D}f)=\Pi_{\pm}(X)(f)$ for all $X \in \Gamma(TM)$ and $ f\in C^{\infty}(M)$. The projected derivative is a derivation, i.e.~it satisfies the Leibniz rule $$\nabla_{\Pi_{\pm}(X)}(f\, Y)=\Pi_{\pm}(X)(f)\,Y+f\,\nabla_{\Pi_{\pm}(X)}Y \ .$$
A full Riemannian characterization of the projected geometry is given in \cite{freidel}, including the projected Riemann tensor and projected torsion. 

In order to extend what we did so far, we introduce a new class of brackets by weakening the compatibility condition following \cite{Vai2,Vai1}.
\theoremstyle{definition}
\begin{definition}
Let $(M, K, \eta)$ be an almost para-Hermitian manifold. A \emph{metric-compatible bracket} is a bilinear operation $\llbracket \cdot, \cdot \rrbracket : \Gamma(TM) \times \Gamma(TM) \rightarrow \Gamma(TM)$ satisfying
\begin{align*}
X\big(\eta(Y,Z)\big)&=\eta(\llbracket X, Y \rrbracket,Z)+\eta(Y, \llbracket X,Z \rrbracket) \ , \\[4pt]
\llbracket X, f\, Y \rrbracket &= f\, \llbracket X,Y \rrbracket + X(f)\, Y \ , \\[4pt]
\eta(Y, \llbracket X,X \rrbracket )& = \eta(\llbracket Y,X \rrbracket, X) \ ,
\end{align*}
for all $X,Y,Z\in\Gamma(TM)$ and $f\in C^\infty(M)$.
\end{definition}
The triple $(TM, \eta, \llbracket \cdot, \cdot \rrbracket)$ defines a \emph{metric algebroid} with anchor given by the identity map. Note that the anchor is not required to satisfy any compatibility condition between the metric-compatible bracket and the Lie bracket; metric algebroids are related to the pre-DFT algebroids of~\cite{membrane} that naturally arise in the extensions of the Courant algebroids of generalized geometry to double field theory. 

Following \cite{Borngeometry}, we introduce a generalized notion of integrability for a metric algebroid.
\theoremstyle{definition}
\begin{definition}
Let $(TM, \eta, \llbracket \cdot, \cdot \rrbracket)$ be a metric algebroid and $K\in{\rm End}(TM)$ a tensor field such that $$K^2=\pm\,\mathds{1} \qquad \mbox{and} \qquad \eta\big(K(X),Y\big)=-\eta\big(X, K(Y)\big) \ ,$$ for all $X,Y\in\Gamma(TM)$. The \emph{generalized Nijenhuis tensor} $\mathcal{N}_K:\Gamma(TM)\times\Gamma(TM)\to\Gamma(TM)$ associated with $K$ is given by
$$\mathcal{N}_K(X,Y)=K^2\big(\llbracket X,Y \rrbracket\big) +\llbracket K(X),K(Y)\rrbracket-K\big(\llbracket K(X),Y\rrbracket+\llbracket X, K(Y) \rrbracket\big) \ .$$
\end{definition}
This allows one to define a D-bracket in the context of para-Hermitian geometry. 
\theoremstyle{definition}
\begin{definition}
A \emph{D-bracket} on an almost para-Hermitian manifold $(M, K,\eta)$ is a metric-compati{-}ble bracket $\llbracket \cdot,\cdot \rrbracket$ for which the generalized Nijenhuiis tensor associated with $K$ vanishes: $\mathcal{N}_K(X,Y)=0$ for all $X,Y\in\Gamma(TM)$. The triple $(K,\eta,\llbracket \cdot,\cdot \rrbracket)$ is a \emph{D-structure} on $M$. \\
A D-structure is \emph{canonical} if it satisfies $$\llbracket \Pi_+(X),\Pi_+(Y) \rrbracket= \Pi_+\big([\Pi_+(X), \Pi_+(Y)]\big) \qquad \mbox{and} \qquad \llbracket \Pi_-(X),\Pi_-(Y) \rrbracket= \Pi_-\big([\Pi_-(X), \Pi_-(Y)]\big) \ .$$
\end{definition}
It is shown in \cite{Borngeometry} that there exists a unique canonical D-bracket on $(M,K, \eta)$ which is given by the bracket defined in \eqn{combra} associated with the canonical para-Hermitian connection $\nabla^{\tt can}.$ We will refer to this bracket as the \emph{D-bracket associated with} $K,$ and denote it by $\llbracket \cdot,\cdot \rrbracket^{\tt D}.$

For a flat manifold the D-bracket has the local expression of the D-bracket of double field theory, as shown in \cite{Vai2}. The \emph{C-bracket} is obtained as the skew-symmetrization of the D-bracket: 
\be\label{eq:Cbracket}
\llbracket X,Y \rrbracket ^{\tt C}= \frac{1}{2}\, \big( \llbracket X,Y \rrbracket^{\tt D} - \llbracket Y,X \rrbracket^{\tt D}\big) \ .
\ee
It is shown in \cite{svoboda} that the relation between the D-bracket and the bracket associated to the Levi-Civita connection, the \emph{${\nabla}^{\tt LC}$-bracket}, involves the exterior derivative of the fundamental 2-form $\omega$, such that the D-bracket of an almost para-Hermitian manifold is given by
\begin{align}
\eta (\llbracket X,Y \rrbracket ^{\tt D}, Z)=&\ \eta(\llbracket X,Y \rrbracket ^{{\nabla}^{\tt LC}},Z) \nonumber\\ &\, - \frac{1}{2} \, \big(\de \omega^{(+3,-0)}+\de \omega^{(+2,-1)}- \de \omega^{(+0, -3)}- \de \omega^{(+1,-2)}\big)(X,Y,Z) \ . \label{dciv}
\end{align}
The difference between a D-bracket and a ${\nabla}^{\tt LC}$-bracket is also called \emph{generalized torsion} of the ${\nabla}^{\tt LC}$-bracket. For an almost para-K\"ahler manifold, $\nabla^{\tt can}=\nabla^{\tt LC}$ and hence $\llbracket X,Y \rrbracket ^{\tt D} = \llbracket X,Y \rrbracket ^{{\nabla}^{\tt LC}}.$ 

From this discussion there also emerges a new interpretation of the section condition of double field theory, proven in~\cite{svoboda}.
\theoremstyle{proposition}
\begin{proposition}\label{flatd}
Let $(M, K,\eta)$ be a flat para-Hermitian manifold and let $X_+, Y_+ \in \Gamma(L_+)$ be vector fields which are parallel along $\mathcal{F}_-.$ Then $$ \llbracket X_+,Y_+ \rrbracket ^{\tt D} =\llbracket X_+,Y_+ \rrbracket^{\de} _+ \ ,$$ or equivalently $$\llbracket X_+,Y_+ \rrbracket^{\de}_- =0 \ .$$
\end{proposition}
Here we see that the section condition $\llbracket X_+,Y_+ \rrbracket^{\de}_- =0$ restricts the vector fields to be sections over the foliation $\mathcal{F}_+$.  

\subsection{Flux Deformations of Para-Hermitian Structures}\label{sec:Btransformations}

We shall now define special isometries relating two different almost para-Hermitian structures on the same manifold $M$ and describe how the D-bracket transforms under their action. In this description we will see strong similarities with the transformations proposed in generalized geometry \cite{gualtieri:tesi,hitch}. We will find that some geometric and non-geometric fluxes appear in this discussion as obstructions to a weaker notion of integrability.

We first need the notion of $B$-transformation for an almost para-Hermitian manifold.
\theoremstyle{definition}
 \begin{definition}
Let $(M, K, \eta)$ be an almost para-Hermitian manifold.  A $B_+$-\emph{transformation} is an isometry of $TM$ given by
$$
e^{B_+}= 
\bigg(\begin{matrix}
\mathds{1} & 0 \\
B_+ & \mathds{1}
\end{matrix}\bigg)
\ \in \ {\sf O}(d,d)(M) \ ,
$$
where we have chosen the splitting $TM= L_+ \oplus L_-$ and $B_+ : \Gamma(L_+) \rightarrow \Gamma(L_-)$ is a skew map in the sense that it satisfies $\eta(B_+(X),Y)=- \eta(X, B_+(Y)).$ 
\end{definition}

A $B_+$-transformation of the almost para-complex structure $K$ is then given by $$K \longmapsto K_{B_+}= e^{B_+}\, K\, e^{-B_+} \ .$$
In the splitting $TM= L_+ \oplus L_-$, the tensor $K$ is given by 
$$
K=
\bigg(\begin{matrix}
\mathds{1} & 0 \\
0 & -\mathds{1}
\end{matrix}\bigg)
\ ,
$$
and hence the transformed almost para-complex structure takes the form 
\be
K_{B_+}=
\bigg(\begin{matrix}
\mathds{1} & 0 \\
2B_+ & -\mathds{1}
\end{matrix}\bigg)
\ . \label{matrK}
\ee
One easily has $K_{B_+}^2=\mathds{1}$, while the skew property of $B_+$ is required for the compatibility condition $\eta(K_{B_+}(X),K_{B_+}(Y))=-\eta(X,Y)$ to be satisfied.

The endomorphism $B_+$ is given either by a 2-form $b_+$ or by a bivector $\beta_-$ defined by $$\eta\big(B_+(X),Y\big)=b_+(X,Y)= \beta_-\big(\eta(X), \eta(Y)\big) \ .$$
The 2-form $b_+$ is of type $(+2,-0)$, while the bivector $\beta_-$ is of type $(+0,-2)$ with respect to $K.$
This is relevant to understanding how the fundamental 2-form $\omega$ changes under a $B_+$-transformation: $$\omega \longmapsto \omega_{B_+}= \eta\, K_{B_+} = \omega + 2b_+\ ,$$
so that such transformations may not preserve the closure of the fundamental 2-form.
A completely analogous discussion can be carried out for a \emph{$B_-$-transformation}, defined by a skew map $B_-: \Gamma(L_-) \rightarrow \Gamma(L_+).$

The main effect of a $B_+$-transformation is that the splitting $TM=L_+\oplus L_-$ changes, i.e.~$e^{B_+}:L_+\oplus L_- \rightarrow L^{B_+}_+\oplus L^{B_+}_-,$ which implies that the potential Frobenius integrability of the original splitting may not be preserved in its image under $e^{B_+}$. The transformed projections are given by
\be
\Pi_+^{B_+}= \frac{1}{2}\,\big(\mathds{1}+K_{B_+}\big)=
\bigg(\begin{matrix}
\mathds{1} & 0 \\
B_+ & 0
\end{matrix}\bigg)
\qquad \mbox{and} \qquad
\Pi_-^{B_+}= \frac{1}{2}\,\big(\mathds{1}-K_{B_+}\big)=
\bigg(\begin{matrix}
0 & 0 \\
-B_+ & \mathds{1}
\end{matrix}\bigg)
\ . \nonumber
\ee
Hence, decomposing any vector field as 
\be
X=
\bigg(\begin{matrix}
X_+ \\
X_-
\end{matrix}\bigg)
\ \in \ \Gamma(TM) \ , \nonumber
\ee
where $X_+ \in \Gamma(L_+)$ and $X_- \in \Gamma(L_-),$ the new distributions are obtained by using the transformed projections to get
$$\Pi^{B_+}_+(X)= X_+ + B_+(X_+) \qquad \mbox{and} \qquad \Pi^{B_+}_-(X)= X_- - B_+(X_+) \ ,$$
where $\Pi^{B_+}_-(X) \in \Gamma(L_-)$ since $B_+$ maps $\Gamma(L_+)$ to $\Gamma(L_-)$. Thus $L^{B_+}_-=L_-.$
On the other hand, the same reasoning applied to $\Pi^{B_+}_+(X)$ shows that it is not an element of $\Gamma(L_+),$ i.e. $L^{B_+}_+\neq L_+.$
Therefore only the $-1$-eigenbundle is preserved by a $B_+$-transformation, while the $+1$-eigenbundle changes; in particular, if $L_+$ is integrable, then integrability of $L^{B_+}_+$ is generally violated.

In order to compare two different almost para-Hermitian structures on the same manifold, a weaker notion of integrability is introduced. The main difference from the usual notion of Frobenius integrability is the replacement of the Lie bracket of vector fields with the D-bracket.  
\theoremstyle{definiton}
\begin{definition} \label{weint}
Let $(M, K,\eta)$ be an almost para-Hermitian manifold with associated D-bracket $\llbracket \cdot ,\cdot \rrbracket^{\tt D}$. An isotropic (with respect to $\eta$) distribution $\mathcal{D}$ is \emph{weakly integrable} if it is involutive under the D-bracket: $$ \llbracket \Gamma(\mathcal{D}) ,\Gamma(\mathcal{D}) \rrbracket^{\tt D} \subseteq \Gamma(\mathcal{D}) \ .$$ 
\end{definition}
For example, the eigenbundles $L_+$ and $L_-$ of $K$ are always weakly integrable. It is clear from this definition that the notion of weak integrability depends on the choice of the almost para-Hermitian structure, as this choice represents the reference almost para-Hermitian structure which defines the D-bracket. Hence we can formulate a notion of compatibility based on this relative integrability.
\theoremstyle{definition}
\begin{definition}
Let $(K,\eta)$ and $(K',\eta)$ be two almost para-Hermitian structures on a manifold $M.$ Then $K'$ is \emph{compatible with} $K$ if the eigenbundles of $K'$ are weakly integrable with respect to~$K$. 
\end{definition}

Any almost para-complex structure $K$ is always compatible with itself.
We can thus analyze the weak integrability of a $B_+$-transformed almost para-complex structure $K_{B_+}$ with respect to the original structure $K$. For this, we note that the D-bracket of sections of the $+1$-eigenbundle of $K_{B_+}$ is given by
\be\nonumber
\eta\big(\llbracket \Pi_+^{B_+}(X), \Pi_+^{B_+}(Y) \rrbracket^{\tt D},  \Pi_+^{B_+}(Z) \big) = \Big(\de_+ b_+ + \big(\mbox{$\bigwedge^3$} \eta\big)[\beta_-,\beta_-]^{\tt S}_-\Big)\big(\Pi_+^{B_+}(X),\Pi_+^{B_+}(Y),\Pi_+^{B_+}(Z)\big) \ ,
\ee
where $\de_+$ is the Lie algebroid differential of $L_+$ and $[\cdot , \cdot ]^{\tt S}_-$ is the Schouten-Nijenhuis bracket of $L_-.$
The $+1$-eigenbundle of $K_{B_+}$  is weakly integrable if and only if $$ \eta\big(\llbracket \Pi_+^{B_+}(X), \Pi_+^{B_+}(Y) \rrbracket^{\tt D},  \Pi_+^{B_+}(Z) \big)=0 \ .$$ This implies that the Maurer-Cartan equation
\be
\de_+ b_+ + \big(\mbox{$\bigwedge^3$} \eta\big)[\beta_-,\beta_-]^{\tt S}_-=0 \label{maucar}
\ee
has to be satisfied in order for $K_{B_+}$ to be compatible with $K$. If this equation is not satisfied, then the components of $\de_+ b_+$ can be interpreted as fluxes, as shown in \cite{svoboda}, with $\de_+ b_+ + \big(\mbox{$\bigwedge^3$} \eta\big)[\beta_-,\beta_-]^{\tt S}_-$ giving the covariant NS--NS $H$-flux. A $B_-$-transformation gives the corresponding dual non-geometric $R$-flux.

It is shown in~\cite{svoboda} that the D-bracket $\llbracket
\cdot,\cdot \rrbracket^{\tt D}_{B_+}$ associated to $K_{B_+}$ is
related to the original D-bracket associated to an almost
para-K\"ahler structure $K$ by
\be\label{eq:DbracketB}
\eta\big(\llbracket X,Y\rrbracket^{\tt D}_{B_+},Z\big)=\eta\big(\llbracket X,Y\rrbracket^{\tt D},Z\big) - \de b_+(X,Y,Z) \ ,
\ee
where the $(+3,-0)$-component of $\de b_+$ (with respect to the
splitting defined by $K_{B_+}$) coincides with the left-hand side of
\eqref{maucar}, while the $(+2,-1)$-component gives the dual
non-geometric $Q$-flux. In this framework, the geometric $f$-flux also
arises generally via the D-bracket $\llbracket \cdot,\cdot \rrbracket^{\tt D}$
in the usual way through diagonal isometries of the
tangent bundle $TM=L_+\oplus L_-$,
\be\label{eq:nonholonomic}
\bigg( \begin{matrix}
A & 0 \\ 0 & \big(A^{-1}\big)^{\rm t}
\end{matrix} \bigg) \ \in \ {\sf O}(d,d)(M) \ ,
\ee
with $A\in{\rm End}(L_+)$, which preserve $K$ and rotate 
frames on the sub-bundles $L_+$ and $L_-$; here $L_-$ is identified
with $L_+^*$ using the Lorentzian metric $\eta$.

\subsection{Recovering the Physical Spacetime Background\label{sec:polarization}}

Let us now briefly describe the physical interpretation of the
formalism thus far, in particular how para-Hermitian geometry recovers
the usual expectations of the doubled geometry of double field theory
and reconciles them with generalized geometry. Building on the local
description in~\cite{geonongeo}, an almost
para-Hermitian structure $(K,\eta)$ on a $2d$-dimensional manifold
$M$, i.e. a splitting of the tangent bundle $TM=L_+\oplus L_-$ into
maximally isotropic sub-bundles, is also called a
\emph{polarization}. To make contact with the generalized geometry of
the standard Courant algebroid, only one of the two distributions of
$T M$ is required to be integrable, in which case $(M,K,\eta)$ is not almost para-K\"ahler (equivalently $\nabla^{\tt LC}\omega\neq0$). 
If $(M,K,\eta)$ is an $L_+$-para-Hermitian
manifold, then $L_+=T\cF_+$ is the tangent bundle of a foliation $\cF_+$ of $M$ of dimension
$d$; the Lagrangian submanifold
$\mathcal{F}_+$ is then the physical spacetime.\footnote{If
  $(M,K,\eta)$ is a para-Hermitian manifold, so that $L_-=T\cF_-$ is
  also integrable, then the foliation $\cF_-$ may be interpreted as
  the auxiliary ``dual'' manifold of the physical spacetime
  $\cF_+$. However, this situation will only arise in some very special
  instances in this paper and is not a general feature of a global doubled geometry.} The tangent bundle $TM$ can be regarded as
a metric algebroid over $\mathcal{F}_+$ with anchor map given by the
$L_+$-para-complex projection
$\Pi_+:\Gamma(TM)\to\Gamma(T\mathcal{F}_+)$. The ${\sf
  O}(d,d)(M)$-invariant metric $\eta$ identifies $L_-$ with
$L_+^*=T^*\mathcal{F}_+$, and $TM$ is identified with the generalized
tangent bundle\footnote{Globally, the
  generalized tangent bundle on $\cF_+$ is the vector bundle
  $E\to\cF_+$ defined by the exact sequence
\be\nonumber
0\longrightarrow T^*\cF_+\longrightarrow E\longrightarrow
T\cF_+\longrightarrow 0
\ee
of bundles on $\cF_+$, where the map $\rho:E\to T\cF_+$ is an anchor. This global description must be used whenever $\cF_+$ is endowed with a non-trivial NS--NS $H$-flux.}
$\mathbb{T}\mathcal{F}_+:=T\mathcal{F}_+\oplus
T^*\mathcal{F}_+$ of $\mathcal{F}_+$ via the projection
isomorphism $${\sf p}_+:\Gamma(TM) \longrightarrow
\Gamma(\mathbb{T}\mathcal{F}_+) \ , \quad X\longmapsto {\sf
  p}_+(X)=\Pi_+(X)+\eta\big(\Pi_-(X)\big) \ , $$ with inverse
${\sf p}^{-1}_+:X_++\alpha_+\mapsto X_++\eta^{-1}(\alpha_+)$ for $X_+\in
\Gamma(T\cF_+)$ and $\alpha_+\in\Gamma(T^*\cF_+)$. This isomorphism sends the metric $\eta$ and the fundamental 2-form $\omega$ to the duality pairings
\begin{align*}
\eta\big({\sf p}^{-1}_+(X_++\alpha_+),{\sf p}^{-1}_+(Y_++\beta_+)\big) &= \beta_+(X_+)+\alpha_+(Y_+) \ , \\[4pt] \omega\big({\sf p}^{-1}_+(X_++\alpha_+),{\sf p}^{-1}_+(Y_++\beta_+)\big) &= \beta_+(X_+)-\alpha_+(Y_+) \ .
\end{align*}
It also sends the canonical D-structure on $M$ to the standard Courant algebroid over $\mathcal{F}_+$, with the $\Pi_+$-projected D-bracket $\llbracket\cdot,\cdot\rrbracket^{\tt D}_+$ on sections of $TM$ mapping to the Dorfman bracket $[\cdot,\cdot]^{\tt D}_{\mathcal{F}_+}$ on sections of $\mathbb{T}\mathcal{F}_+$:
\be\nonumber
{\sf p}_+\big(\llbracket X,Y\rrbracket^{\tt D}_+\big) = [{\sf p}_+(X),{\sf p}_+(Y)]^{\tt D}_{\mathcal{F}_+} \ .
\ee
This result is established in detail in~\cite{freidel,svoboda,Borngeometry}; an analogous result is given by~\cite{membrane} in terms of the C-bracket \eqref{eq:Cbracket} on $TM$ and the Courant bracket on $\mathbb{T}\mathcal{F}_+$ for the particular case where $M=T^*\cQ$ is the total space of the cotangent bundle of a $d$-dimensional manifold $\cQ$, which in subsequent sections we will describe as a key example of an almost para-Hermitian manifold.

In order to lead to a string background, we also need to specify how to recover the physical background fields of supergravity on $\mathcal{F}_+$, such as the spacetime metric $g$ and the Kalb-Ramond field $B$. This requires a dynamical augmentation of the kinematical data of an almost para-Hermitian structure, for which we follow~\cite{metastring,Borngeometry}.
\begin{definition}
A \emph{generalized metric} on an almost para-Hermitian manifold $(M,K,\eta)$ is a Riemannian metric $\mathcal{H}$ on $M$ which is compatible with the ${\sf O}(d,d)(M)$-invariant metric $\eta$ and the fundamental 2-form $\omega$ in the sense that
\be\nonumber
\eta^{-1}\,\mathcal{H} = \mathcal{H}^{-1}\, \eta \qquad \mbox{and} \qquad \omega^{-1}\,\mathcal{H} = -\mathcal{H}^{-1}\,\omega \ .
\ee
The triple $(\eta,\omega,\mathcal{H})$ is a \emph{Born geometry} on $M$ and $(M,\eta,\omega,\mathcal{H})$ is a \emph{Born manifold}.
\end{definition}
A generalized metric can also be regarded as an almost Hermitian
metric relative to $\omega$, while a Born geometry can be regarded as a $\sf G$-structure on $M$ with $${\sf G}={\sf O}(d,d)\cap{\sf Sp}(2d,\mathbb{R})\cap {\sf O}(2d) = {\sf O}(d) \ . $$ It is shown in~\cite{Borngeometry} that there always exists a choice of frame on $TM=L_+\oplus L_-$ in which the generalized metric can be brought into the diagonal form
\be\label{eq:diaghermmetric}
\mathcal{H}_0 = \bigg( \begin{matrix}
g_+ & 0 \\ 0 & g_+^{-1}
\end{matrix} \bigg) \ ,
\ee
where $g_+$ is a metric on the sub-bundle $L_+$ and we have identified $L_-=L_+^*$ using the Lorentzian metric $\eta$. In the case that $(M,K,\eta)$ is $L_+$-para-Hermitian, this shows that the generalized metric encodes a choice of Riemannian metric on the physical spacetime submanifold $\cF_+$ and provides an ${\sf O}(d){\times}{\sf O}(d)$-structure on the generalized tangent bundle $\mathbb{T}\cF_+$.

Having discussed how to obtain the conventional spacetime description
from a choice of polarization, it is then natural to understand the
meaning of changing polarization in the framework of para-Hermitian
geometry, extending the notions introduced in~\cite{geonongeo}.
\theoremstyle{definition}
\begin{definition}\label{def:polarizationchange}
A \emph{change of polarization} on an almost para-Hermitian manifold
$(M,K, \eta)$ is an isometry $\vartheta \in {\sf O}(d,d)(M)$ mapping the almost para-Hermitian structure $(K, \eta)$ into $(K_\vartheta, \eta)$ with $K_\vartheta = \vartheta^{-1} \, K \, \vartheta$.
\end{definition}
From this definition it is easy to check that $(K_\vartheta,\eta)$ is also an almost para-Hermitian structure on $M$, i.e.~$K_\vartheta^2=\unit$ and $K_\vartheta^{\rm t} \, \eta \, K_\vartheta=-\eta$, and that the fundamental 2-form transforms into $$\omega\longmapsto\omega_\vartheta=\eta\, K_\vartheta = \vartheta^{\rm t}\, \omega\, \vartheta \ .$$
Such transformations do not generally preserve the (Frobenius or weak)
integrability of the eigendistributions, or the closure of the
fundamental 2-form. In this sense, the choice of polarization contains
all information about fluxes and the spacetime background. We will
show explicitly later on, through some prototypical examples, that the fluxes appear as obstructions to weak integrability with respect to a reference para-K\"ahler structure. On the other hand, the background geometry arises from a choice of generalized metric $\cH$, i.e. a Born geometry $(\eta,\omega,\cH)$ on $M$, which transforms under a change of polarization into
\be\nonumber
\cH\longmapsto\cH_\vartheta = \vartheta^{\rm t} \, \cH \, \vartheta \ .
\ee
Importantly, these transformations describe T-dualities (and other
symmetries) on an almost para-Hermitian manifold, interpreted as a
doubled spacetime: The smooth map $\vartheta:M\to{\sf O}(d,d)$ acts by an element of the continuous
T-duality group ${\sf O}(d,d)$. 
We have already encountered a special class of changes of polarization, namely the $B$-transformations of Section~\ref{sec:Btransformations}. In this case, a $B_+$-transformation changes the polarization by $\vartheta=e^{-B_+}$; in particular, the diagonal generalized metric \eqref{eq:diaghermmetric} is mapped to
\be\label{eq:geometricH}
\cH_{B_+} = \big(e^{-B_+}\big)^{\rm t} \, \cH_0 \, e^{-B_+} = \bigg( \begin{matrix}
g_+- b_+\, g_+^{-1}\, b_+ & b_+\, g_+^{-1} \\ -g_+^{-1}\, b_+ & g_+^{-1}
\end{matrix} \bigg) \ .
\ee
When $(M,K,\eta)$ is an $L_+$-para-Hermitian manifold, this is the familiar form from generalized geometry of the generalized metric on the physical spacetime $\cF_+$ which unifies the target space metric $g_+$ and the Kalb-Ramond 2-form field $b_+=\eta\, B_+$.

For later use, let us spell out the form of such transformations in
local coordinates, and in particular show that any two almost para-Hermitian
structures on the same manifold $M$ with the same compatible metric
$\eta$ are related by a change of
polarization in the sense of Definition~\ref{def:polarizationchange}. Let $(M, K, \eta)$ be an almost para-Hermitian
manifold whose eigendistributions $L_+$ and $L_-$ are locally
spanned, in a given open contractible chart on $M$, by vector fields $Z_i$ and $\tilde Z^i$:
$\Gamma(L_+)={\rm Span}_{C^\infty(M)}\{ Z_i \}$ and $\Gamma(L_-)={\rm Span}_{C^\infty(M)}\{
\tilde{Z}^i\}$. Let $\Theta^i$ and $\tilde{\Theta}_i$ be the
respective dual 1-forms, so that we can write the local expression of
the almost para-Hermitian structure as\footnote{Throughout implicit summation over repeated upper and lower indices is understood.}
 $$K=Z_i \otimes \Theta^i -
\tilde{Z}^i \otimes \tilde{\Theta}_i \qquad \mbox{and} \qquad \eta=
\eta_i^j\,\big(\Theta^i \otimes \tilde{\Theta}_j + \tilde{\Theta}_j
\otimes \Theta^i\big)$$ with $\omega= \eta^i_j \, \tilde{\Theta}_i \wedge \Theta^j.$
Given another almost para-Hermitian structure $(K', \eta)$ on $M$,
we write the corresponding eigendistributions locally in the same 
chart as
$\Gamma(L'_+)={\rm Span}_{C^\infty(M)}\{ Z'_i \}$ and
$\Gamma(L'_-)={\rm Span}_{C^\infty(M)}\{ \tilde{Z}^{\prime\,i}\}$, with dual
1-forms $\Theta^{\prime\,i}$ and $\tilde{\Theta}'_i$. We can then
write the ${\sf O}(d,d)(M)$-transformation from $K$ to $K'$
as $$\vartheta=Z_i \otimes \Theta^{\prime\,i} + \tilde{Z}^i \otimes
\tilde{\Theta}'_i \ ,$$ whose inverse is thus given by  
$$\vartheta^{-1}= Z'_i \otimes \Theta^i + \tilde{Z}^{\prime\,i}
\otimes \tilde{\Theta}_i \ ,$$ and whose transpose is $$\vartheta^{\rm
  t}=\tilde{\Theta}'_i \otimes \tilde{Z}^i + \Theta^{\prime\,i}
\otimes Z_i \ .$$ It is easily checked that $$\eta'=\vartheta^{\rm
  t}\, \eta\, \vartheta=  \eta_i^j\,\big(\Theta^{\prime\,i} \otimes
\tilde{\Theta}'_j + \tilde{\Theta}'_j \otimes
\Theta^{\prime\,i}\big)=\eta \ ,$$ which is indeed compatible with $K'.$

The fact that changes of polarization generally induce flux
deformations of the almost para-Hermitian structure, and hence may
spoil (Frobenius or weak) integrability of the eigendistributions,
means that a conventional spacetime description is not always
possible~\cite{hulled}. While geometric fluxes give twisted distributions which are
globally (weakly) integrable, some flux deformations preserve weak
integrability only locally and the foliations are not globally
defined; such fluxes are said to be \emph{globally
non-geometric}. Others spoil integrability altogether, so that not even
a local geometric spacetime picture can emerge; such fluxes are called
\emph{locally non-geometric}. In the following we will spell this
picture out explicitly in several concrete classes of backgrounds, and
in particular obtain a new geometric impetus on the point of view that
non-geometric backgrounds are noncommutative and nonassociative
spacetimes~\cite{Bouwknegt2004,Blumenhagen2010,Lust2010,Blumenhagen2011,Condeescu2012,Mylonas2012,Blumenhagen2013,Mylonas2013,Aschieri2015,Freidel2017}. In
the case where the polarizations are related by T-dualities or other
symmetries of string theory, they give physically equivalent string
backgrounds. In this paper we do not address the general problem of
which changes of polarization $\vartheta\in{\sf O}(d,d)(M)$ yield proper string symmetries.

\section{Dynamical Para-K\"ahler Structures}\label{sec:dynamicalpara}

In the following we will describe some dynamical systems in which
para-Hermitian structures naturally arise, giving a more elementary
appearence of para-Hermitian geometry than in the construction of a globally well-defined setting for the kinematics of double field theory.
A clarifying class of examples of para-Hermitian geometry comes from Lagrangian dynamics, i.e. from the tangent lift of the dynamics to the tangent bundle of a configuration space with a sufficiently regular function defined on this bundle which encodes the equations of motion. In this section we give a new interpretation to a widely discussed subject, commonly known as Finsler geometry which describes the geometry arising from regular functions on a manifold (such as Lagrangians and Hamiltonians), as an instance of para-K\"ahler geometry; a discussion of Lagrangian and Hamiltonian geometry in terms of Finsler geometry can be found in~\cite{bucataru}. 

\subsection{Newtonian Dynamical Systems and Their Lifts}

In order to understand the lifting procedure, we need a precise definition of a dynamical system. In this paper we will focus on Newtonian dynamical systems \cite{calcvar}.
 \theoremstyle{definition}
\begin{definition}
A \emph{Newtonian dynamical system} is given by a $d$-dimensional
manifold $\cQ$, called \emph{configuration space}, and a second order
differential equation given, in a local chart of $\cQ$ with coordinates $ q=( q^i )$, by
\be
\frac{\de^2 q^i}{\de t^2}= \Phi^i (q, \dot{q}) \ , \label{secord}
\ee
with $t$ a real parameter, $\dot{q}{}^i=\frac{\de q^i}{\de t}$ and $\Phi^i (q, \dot{q})$ a function of $(q^i, \dot{q}{}^i)$ assigning a time evolution law.\\
A \emph{trajectory} of the dynamical system is a curve $g : \mathbb{R}
\rightarrow \cQ$ given, in a local chart $(U, \phi)=(U,q^i)$
on $\cQ$, by $\phi \circ g : \mathbb{R} \ni t \mapsto ( q^i(t)) \in \mathbb{R}^d$ such that $q^i(t)$ are solutions of the differential equation \eqref{secord}.  
\end{definition}

The differential equation \eqref{secord} does not separate the trajectories on $\cQ,$ i.e. there are an infinite number of trajectories passing through each point in $\cQ,$ and hence a different description of the dynamical system is needed in order to find a unique solution to \eqref{secord} for any set of initial conditions.
Roughly speaking,  we need to find an equivalent system of \emph{first order} differential equations by enlarging the space on which they are defined so that there are enough initial conditions to formulate a well-posed Cauchy problem, and hence to obtain a unique solution. From a geometric point of view, this means that we have to find, on this enlarged manifold $M$, a vector field $\Sigma \in \Gamma(TM)$ with components locally defined by first order differential equations, whose integral curves can be \emph{projected} onto the trajectories of the dynamical system on $\cQ$. This leads to the definition of a lift of the dynamics.
\theoremstyle{definition}
\begin{definition}
A \emph{lift of the dynamics} is the association of an equivalent
first order dynamical system field $\Sigma \in \Gamma(TM)$ on a
\emph{carrier manifold} $M$ to the Newtonian dynamical system on
$\cQ$. The inverse procedure of mapping integral curves of the first
order dynamical system field $\Sigma$ to trajectories of the original system is \emph{projection}.
\end{definition}
This definition shows that fiber bundles with base space the
configuration space $\cQ$ are natural choices for lifting the
dynamics. In particular, the bundle projection plays a crucial role in
the description of the geometry of such lifts: If we consider as
carrier manifold $M$ the total space of a fiber bundle $E$ with smooth structure induced by that of the base manifold $\cQ$, the
projection is naturally defined by the surjective map $\pi: E
\rightarrow \cQ.$ The tangent map $T\pi: TE \rightarrow T\cQ$ induced
by the projection defines a splitting $TE= L_{\tt v}(E) \oplus L_{\tt
  h}(E),$ where $L_{\tt v}(E)= \ker(T\pi)$ is called the \emph{vertical
  sub-bundle} and $L_{\tt h}(E)$ is the complementary \emph{horizontal
  sub-bundle}. Any such splitting of $TE$ can be regarded as an almost
product structure on $E$ for which $L_{\tt v}(E)$ and $L_{\tt h}(E)$ are its eigenbundles.

The vertical sub-bundle is defined entirely by the projection. In order to understand how the horizontal sub-bundle encodes the information about the dynamics, we will now describe the canonical lift to the tangent bundle of the configuration space.
\theoremstyle{definition}
\begin{definition}
The \emph{canonical lift} on $T\cQ$ of a Newtonian dynamical system on $\cQ$ is given by the correspondence to the differential equations \eqref{secord} of a \emph{second order vector field} $\Sigma \in \Gamma( T(T\cQ))$ such that:  
\begin{enumerate}
\item[(a)] Integral curves of $\Sigma$ are obtained as tangent lifts of curves on $\cQ$, i.e. $h(t)= Tg(t,1) \in T\cQ$ where $h(t)$ is an integral curve of $\Sigma$ and  $g : \mathbb{R}\rightarrow \cQ$.
\item[(b)] The bundle projection $\pi: T\cQ \rightarrow \cQ$ defines a trajectory $t\mapsto\pi \circ h (t) \in \cQ$ of the dynamical system on $\cQ.$
\end{enumerate}
\end{definition}
This uniquely defines the second order vector field $\Sigma.$ In a local chart $(\pi^{-1}(U),q^i,v^i)$ on $T\cQ$ induced by a local chart $(U,\phi)=(U,q^i)$ on $\cQ$, its expression is $$\Sigma= v^i \, \frac{\partial}{\partial q^i} + \Phi^i (q,v)\,\frac{\partial}{\partial v^i} \ ,$$ such that the equivalent system of first order differential equations is 
$$ \frac{\de q^i}{\de t} = v^i \qquad \mbox{and} \qquad \frac{\de v^i}{\de t}  = \Phi^i (q,v) \ .$$

In this case, an equivalent statement is that $X\in \Gamma(T(T\cQ))$ is
a vertical vector field if its action on functions which are constant
along the fibers vanishes, i.e. $\pounds_X (\pi^* f)=0$ for all $ f
\in C^{\infty}(\cQ)$, where $\pounds$ denotes the Lie derivative. Using
the identity $\pounds_{[X,Y]}=\pounds_X\, \pounds_Y -\pounds_Y\,
\pounds_X,$ it follows that $\pounds_{[X,Y]}(\pi^* f )=0$ for all $ f \in C^{\infty}(\cQ)$ if $X,Y \in \Gamma(L_{\tt v})$ are vertical vector fields. Hence $[X,Y] \in \Gamma(L_{\tt v}(T\cQ))$ and $L_{\tt v}(T\cQ)$ is an involutive distribution. Thus it is Frobenius integrable, and so it describes the foliation of $T\cQ$ with the fibers as leaves.

On the other hand, the \emph{vertical lift} $X_{\tt v}\in \Gamma(L_{\tt v}(T\cQ))$ of a vector field $X \in \Gamma(T\cQ)$ is the infinitesimal generator of translations along the fibers, i.e. the one-parameter group of diffeomorphisms defined by $\mathbb{R} \ni t \mapsto (q, t\, X \rvert_q) \in T\cQ.$  This defines a map $\rho: \Gamma(T\cQ) \rightarrow \Gamma(L_{\tt v}(T\cQ))$ which in local coordinates reads $$\rho: X=X^i \, \frac{\partial}{\partial q^i} \longmapsto X_{\tt v}= (\pi^* X^i )\, \frac{\partial}{\partial v^i} \ ,$$ where the components $\pi^* X^i$ are functions which are constant along the fibers. Thus $\big\{ \frac{\partial}{\partial v^i}\big\}$ locally spans $\Gamma(L_{\tt v}(T\cQ))$, and so $[X_{\tt v}, Y_{\tt v}]=0$ for all $X_{\tt v}, Y_{\tt v} \in \Gamma(L_{\tt v}(T\cQ)).$

In order to describe the horizontal distribution induced by $\Sigma$ on $T(T\cQ)$, we need to introduce the vertical endomorphism of $T(T\cQ)$.
\theoremstyle{definition}
\begin{definition}
The \emph{vertical endomorphism} $S \in \mathrm{End}(T(T\cQ))$ is the $(1,1)$-tensor field which is the composition of the vertical lift and the tangent projection: $S= \rho \circ T\pi,$ or equivalently the endomorphism of $T(T\cQ)$ which makes the diagram
\begin{center}
\begin{tikzcd}
T(T\cQ) \arrow{r}{T\pi}  \arrow{rd}{S} 
  & T\cQ \arrow{d}{\rho} \\
    & T(T\cQ)
\end{tikzcd}
\end{center}
commute.
\end{definition}
The tensor $S$ is called the vertical endomorphism because when acting
on vector fields, $\ker(S)={\rm im}(S)=\Gamma(L_{\tt v}(T\cQ))$ and
$S^2=0$. This also implies that $S$ is integrable, i.e. it has
vanishing Nijenhuis tensor $N_S=0$, so that $S$ defines a nilpotent
structure, and that in local coordinates it is given by $$S= \frac{\partial}{\partial v^i} \otimes \de q^i \ .$$

It can be shown~\cite{calcvar} that $(\pounds_{\Sigma}S )^2
=\mathds{1}.$  It is also shown in \cite{calcvar} that $L_+=L_{\tt
  v}(T\cQ)$ is the $+1$-eigenbundle of $\pounds_{\Sigma}S.$ The
horizontal sub-bundle $L_-=L_{\tt h}(T\cQ)$ is therefore the
$-1$-eigenbundle of $\pounds_{\Sigma}S$ and its elements, as
\emph{horizontal lifts} of vector fields $X \in \Gamma(T\cQ),$ take
the form $X_{\tt h} = \frac{1}{2}\, ([X_{\tt v}, \Sigma]+X^{\tt c}),$
where $X^{\tt c}$ is the complete lift\footnote{Let $X \in
  \Gamma(T\cQ)$ and let $\varphi_t$ be the local one-parameter group
  of diffeomorphisms on $\cQ$ generated by $X.$ The infinitesimal
  generator $X^{\tt c} \in \Gamma(T(T\cQ))$ of the local one-parameter
  group of diffeomorphisms on $T\cQ$ defined by the canonical lift
  $\varphi^{\tt c}_t=T\varphi_t$ is called the \emph{complete lift} of
  $X.$ It defines a canonical injection
  $\Gamma(T\cQ)\to\Gamma(T(T\cQ))$ by the directional derivatives
  $X^{\tt c}(\pi^*f)=\pi^* X(f)$ for $f\in C^\infty(\cQ)$. In local coordinates, if $X=X^i\,\frac\partial{\partial q^i}$
  then
\be\nonumber
X^{\tt c} = X^i\,\frac\partial{\partial q^i} + v^i\,\frac{\partial
  X^j}{\partial q^i}\, \frac\partial{\partial v^j} \ .
\ee
} of $X.$ The rank of $L_{\tt v}(T\cQ)$ is $d$, and since $L_{\tt
  h}(T\cQ)$ is the complementary sub-bundle of $L_{\tt v}(T\cQ)$ it
also has rank $d.$ Thus the second order vector field $\Sigma$,
together with the naturally defined maps $\rho$ and $T\pi$, define an
almost para-complex structure $K=\pounds_{\Sigma}S$ on $M=T\cQ.$ As we
have seen, the $+1$-eigenbundle $L_+=L_{\tt v}(T\cQ)$ is Frobenius
integrable and the corresponding foliation $\cF_+$ is canonically
identified with the fibers of $T\cQ$, i.e. the space of velocities $v$.

In local coordinates a horizontal lift reads $$X_{\tt h}= (\pi^* X^i)\, D_i \qquad \mbox{with} \quad D_i= \Big(\frac{\partial}{\partial q^i}\Big)_{\tt h}= \frac{\partial}{\partial q^i}+ \frac{1}{2}\,\frac{\partial \Phi^k}{\partial v^i}\, \frac{\partial}{\partial v^k} \ ,$$ where $\{D_i\}$ is a local basis spanning $\Gamma(L_{\tt h}(T\cQ)).$
This easily shows that $L_{\tt h}(T\cQ)$ is not an integrable distribution, as $$[D_i, D_j]= \frac{1}{2}\, \bigg( \frac{\partial^2 \Phi^k}{\partial q^i \, \partial v^j}-  \frac{\partial^2 \Phi^k}{\partial v^i \, \partial q^j}+ \frac{1}{2}\, \Big(  \frac{\partial^2 \Phi^k}{\partial v^j \, \partial v^m} \, \frac{\partial \Phi^m}{\partial v^i}- \frac{\partial^2 \Phi^k}{\partial v^i \, \partial v^m} \, \frac{\partial \Phi^m}{\partial v^j}\Big) \bigg) \, \frac{\partial}{\partial v^k} \ ,$$ so that $[D_i, D_j] \in \Gamma(L_{\tt v}(T\cQ)).$

We can now obtain the local expressions of the 1-forms $\tau^i$ and $ \alpha^i$ dual to the spans of $\Gamma(L_{\tt v}(T\cQ))$ and $\Gamma(L_{\tt h}(T\cQ)).$ Imposing the duality conditions $$ \tau^i\Big(\frac{\partial}{\partial v^j}\Big)= \alpha^i(D_j)= \delta^i{}_j \qquad \mbox{and} \qquad \tau^i(D_j)=\alpha^i\Big(\frac{\partial}{\partial v^j}\Big)=0 \ ,$$ we obtain 
$$ \alpha^i = \de q^i \qquad \mbox{and} \qquad \tau^i= \de v^i - \frac{1}{2} \, \frac{\partial \Phi^i}{\partial v^j} \, \de q^j \ .$$
Hence the local expression of the dynamical almost para-complex structure is $$K=\pounds_{\Sigma}S= \frac{\partial}{\partial v^i} \otimes \tau^i - D_i \otimes \de q^i \ .$$
Similarly the projections $\Pi_{\pm}= \frac{1}{2}\,(\mathds{1}\pm \pounds_{\Sigma}S)$ are locally given by $$\Pi_+= \frac{\partial}{\partial v^i}\otimes \tau^i \qquad \mbox{and} \qquad \Pi_-=D_i \otimes \de q^i \ .$$
Finally it is straightforward to compute that the Nijenhuis tensor associated to $K=\pounds_{\Sigma}S$ is $$N_K= 2\,[D_i,D_j]\otimes \de q^i \otimes \de q^j \ ,$$ showing once more that the complete integrability of the almost para-complex structure is violated by the horizontal eigenbundle. 

\begin{remark}
This construction is reminescent of the definition of a linear
connection on principal and vector bundles, which also introduces a
splitting of the tangent bundle of the total space. However the
condition here defining the horizontal sub-bundle is different and generally weaker than the usual condition for connections on an associated vector bundle. Because of this such a construction is called a \emph{non-linear connection}. The properties of non-linear connections on the tangent bundle are discussed in \cite{Crampin1,Crampin2}. 
\end{remark}

\subsection{Lagrangian Dynamics and Born Geometry}

Our aim now is to connect the completely general discussion above to a
specific case, the description of Newtonian dynamical systems admitting a regular Lagrangian. In this framework we will encounter a simple instance of para-K\"ahler geometry.

We first review some useful notions about the Lagrangian formalism. 
\theoremstyle{definition}
\begin{definition}
Let $\mathcal{L} \in C^{\infty}(T\cQ)$ be a smooth function. The 1-form $\theta_{\mathcal{L}}= S(\de \mathcal{L})$ is the \emph{Cartan 1-form} associated with $\mathcal{L}$. The closed 2-form $\Omega_{\mathcal{L}}=-\de \theta_{\mathcal{L}}$ is the \emph{Lagrangian 2-form}. The function $\mathcal{L}$ is a \emph{regular Lagrangian} if and only if $\Omega_{\mathcal{L}}$ is non-degenerate, and hence a symplectic form.
\end{definition}
From this definition it follows that the Cartan 1-form is horizontal: If $X \in \Gamma(T(T\cQ)),$ then $S(X) \in \Gamma(L_{\tt v}(T\cQ))$ is a vertical vector field, so $\theta_{\mathcal{L}}(S(X))= S (\de \mathcal{L}) (S(X))=  \de \mathcal{L} (S^2 (X))=0.$

The Euler-Lagrange equations read
as $$\pounds_{\Sigma}\theta_{\mathcal{L}}- \de \mathcal{L}=0 \ .$$ By
applying the Cartan formula
\be\nonumber
\pounds_\Sigma=\de\,\imath_\Sigma+\imath_\Sigma\,\de
\ee
for the Lie derivative we obtain $$\imath_{\Sigma} \Omega_{\mathcal{L}}= \de \mathcal{E}_{\mathcal{L}} \ ,$$ where $\mathcal{E}_{\mathcal{L}}= \imath_{\Sigma}\theta_{\mathcal{L}}-\mathcal{L}.$ This shows that $\Sigma$ is the Hamiltonian vector field of the Hamiltonian function $\mathcal{E}_{\mathcal{L}}$, which is globally defined on $T\cQ$ because it is directly derived from the Lagrangian function.

In local coordinates $(q^i, v^i)$ on $T\cQ$, the Cartan 1-form
reads $$\theta_{\mathcal{L}}= \frac{\partial \mathcal{L}}{\partial
  v^i} \, \de q^i \ , $$ so that the Hamiltonian function is given by
\be\nonumber
\cE_\cL = v^i\,\frac{\partial\cL}{\partial v^i} - \cL \ ,
\ee
and the Lagrangian 2-form is $$\Omega_{\mathcal{L}}=\frac{\partial^2 \mathcal{L}}{\partial v^i \, \partial v^j} \, \de q^i \wedge \de v^j+\frac{1}{2}\,\Big( \frac{\partial^2 \mathcal{L}}{\partial v^i\, \partial q^j} -\frac{\partial^2 \mathcal{L}}{\partial q^i\, \partial v^j} \Big) \, \de q^i \wedge \de q^j \ .$$
The local expression of $\theta_{\mathcal{L}}$ explicitly shows that it is a horizontal 1-form, while the local form of $\Omega_{\mathcal{L}}$ gives another formulation of the regularity requirement for the Lagrangian: $\Omega_{\mathcal{L}}$ is non-degenerate if $\ker(\Omega_{\mathcal{L}}) =\{ X \in \Gamma(T(T\cQ)) : \imath_X \Omega_{\mathcal{L}}=0 \}=0.$ Hence, given any vector field $X=X^i\, \frac{\partial}{\partial q^i}+ \tilde{X}^i\, \frac{\partial}{\partial v^i} \in \Gamma(T(T\cQ)),$ we compute  
$$\imath_X \Omega_{\mathcal{L}}= \frac{\partial^2 \mathcal{L}}{\partial v^i\, \partial v^j} \, X^i \, \de v^j + \bigg( \Big(\frac{\partial^2 \mathcal{L}}{\partial v^i \, \partial q^j}- \frac{\partial^2 \mathcal{L}}{\partial q^i \, \partial v^j}\Big)\, X^i -\frac{\partial^2 \mathcal{L}}{\partial v^i\, \partial v^j} \, \tilde{X}^i \bigg) \, \de q^j \ .$$
This shows that $\imath_X \Omega_{\mathcal{L}} \neq 0$ for any
$X\neq0$ if and only if $\det\big( \frac{\partial^2
  \mathcal{L}}{\partial v^i \,\partial v^j} \big) \neq 0.$ Thus a
Lagrangian $\mathcal{L}$ is regular if and only if the Hessian matrix $\big(\frac{\partial^2 \mathcal{L}}{\partial v^i\, \partial v^j}\big)$ has maximum rank $d$.

It can also be shown \cite{calcvar} that 
\be
\Omega_{\mathcal{L}}(X_{\tt v},Y_{\tt v})=\Omega_{\mathcal{L}}(X_{\tt h},Y_{\tt h})=0 \qquad \mbox{and} \qquad \Omega_{\mathcal{L}}(X_{\tt v},Y_{\tt h})+\Omega_{\mathcal{L}}(X_{\tt h},Y_{\tt v})=0 \ , \label{omepro}
\ee
for all $ X_{\tt v}, Y_{\tt v} \in \Gamma(L_{\tt v}(T\cQ))$ and $ X_{\tt h}, Y_{\tt h} \in \Gamma(L_{\tt h}(T\cQ)).$ This implies that the local expression of the Lagrangian 2-form can be written as
$$\Omega_{\mathcal{L}}=\eta_{ij}\, \de q^i \wedge \tau^j \ ,$$
where $\eta_{ij}=\Omega_{\mathcal{L}}\big(D_i, \frac{\partial}{\partial v^j} \big)=\eta_{ji}.$
The vanishing conditions \eqn{omepro} give the compatibility between the almost para-complex structure $K=\pounds_{\Sigma}S$ and the Lagrangian 2-form.
 \theoremstyle{proposition}
\begin{proposition}
Let $\mathcal{L}$ be a regular Lagrangian on $T\cQ$, $\Omega_{\mathcal{L}}$ the associated Lagrangian 2-form and $K=\pounds_{\Sigma}S$ the dynamical almost para-complex structure on $T\cQ$. Then the $(0,2)$-tensor field $\eta_\cL$ defined by $$\eta_\cL(X,Y)= \Omega_{\mathcal{L}}\big(K(X),Y\big)$$ for $X,Y\in\Gamma(T(T\cQ))$ is a metric tensor with Lorentzian signature $(d,d)$, i.e. the vertical and horizontal distributions $L_{\tt v}(T\cQ)$ and $L_{\tt h}(T\cQ)$ are maximally isotropic with respect to $\eta_\cL.$ Thus $(T\cQ, K,\eta_\cL)$ is an $L_{\tt v}(T\cQ)$-para-K\"ahler manifold with fundamental 2-form $\Omega_{\mathcal{L}}$.
\end{proposition}
\begin{proof}
The non-degeneracy of $\eta_\cL$ follows from the non-degeneracy of both $K$ and $\Omega_{\mathcal{L}}.$
Recall that $K(X_{\tt v})=X_{\tt v}$ and $ K(X_{\tt h})=-X_{\tt h}.$  It follows that
\begin{align*}
\eta_\cL(X_{\tt v},Y_{\tt v})&= \Omega_{\mathcal{L}}\big(K(X_{\tt v}),Y_{\tt v}\big)=\Omega_{\mathcal{L}}(X_{\tt v},Y_{\tt v})=0 \ , \\[4pt]
\eta_\cL(X_{\tt h},Y_{\tt h})&=\Omega_{\mathcal{L}}\big(K(X_{\tt h}),Y_{\tt h}\big)=-\Omega_{\mathcal{L}}(X_{\tt h},Y_{\tt h})=0 \ , \\[4pt]
\eta_\cL(X_{\tt v},Y_{\tt h})&= \Omega_{\mathcal{L}}\big(K(X_{\tt v}),Y_{\tt h}\big)=\Omega_{\mathcal{L}}(X_{\tt v},Y_{\tt h})=\eta_{ij}\, X_{\tt v}^i\, Y_{\tt h}^j \ ,
\end{align*}
in which we used \eqn{omepro} and the local expression of $\Omega_{\mathcal{L}}$. From the last equation, we infer that $\eta_\cL$ is symmetric, since $\eta_{ij}$ is symmetric. Therefore $\eta_\cL$ defines a metric compatible with $K$ and the Lagrangian  2-form. It also follows that the two eigenbundles $L_{\tt v}(T\cQ)$ and $L_{\tt h}(T\cQ)$ of $K$ are isotropic with respect to $\eta_\cL$ and, since they are both of rank $d$, they are maximal. We have already seen that only $L_{\tt v}(T\cQ)$ is an integrable eigenbundle.
Thus $(T\cQ, K=\pounds_{\Sigma}S,\eta_\cL)$ is an $L_{\tt v}(T\cQ)$-para-K\"ahler manifold, since $\Omega_{\mathcal{L}}$ is symplectic for a regular Lagrangian.
\end{proof}
In local coordinates, the metric $\eta_\cL$ takes the
form $$\eta_\cL=\eta_{ij}\, \big(\de q^i \otimes \tau^j + \tau^i \otimes \de
q^j \big) \qquad \mbox{with} \quad \eta_{ij}= \frac{\partial^2
  \mathcal{L}}{\partial v^i \, \partial v^j} \ , $$ showing once more the importance of the regularity condition for the Lagrangian.

Any dynamical system described by a regular Lagrangian $\mathcal{L}$ induces an almost K\"ahler structure on the tangent bundle $T\cQ$ (see e.g.~\cite{calcvar}). The almost complex structure $I$ on the tangent bundle, associated with the second order vector field $\Sigma \in \Gamma(T(T\cQ)),$ is given by $$I=S+\frac{1}{2}\,K\,\Pi_+ \ ,$$ where  
$K=\pounds_{\Sigma}S$ is the almost para-complex structure associated with $\Sigma$ and $\Pi_+= \frac12\,(\mathds{1}+K)$ its vertical projector. It is easy to show that $I(X_{\tt v})= - X_{\tt h}$ and $I(X_{\tt h})=X_{\tt v}$, where $X_{\tt h}$ and $ X_{\tt v}$ are, respectively, the horizontal and vertical lift of a vector field $X \in \Gamma(T\cQ).$ In local coordinates, by fixing the splitting $T(T\cQ)=L_{\tt v}(T\cQ)\oplus L_{\tt h}(T\cQ)$ the almost complex structure $I$ reads $$I= \frac{\partial}{\partial v^i}\otimes \de q^i - D_i \otimes \tau^i \ .$$
Given the properties of the Lagrangian 2-form $\Omega_{\mathcal{L}},$
it follows that it is compatible with the almost complex structure
$I,$ i.e. they satisfy the
relation $$\Omega_{\mathcal{L}}\big(I(X),Y\big)+\Omega_{\mathcal{L}}\big(X,I(Y)\big)=0
\ ,$$ for all $ X,Y \in \Gamma(T(T\cQ)).$ We can then introduce the Hermitian metric $$\mathcal{H}_\cL(X,Y) = \Omega_{\mathcal{L}}\big(I(X),Y\big) \ ,$$ such that 
$$\mathcal{H}_\cL\big(I(X), I(Y)\big)=\mathcal{H}_\cL(X,Y) \ , $$
for all $X,Y \in \Gamma(T(T\cQ))$. In local coordinates, it has the expression $$\mathcal{H}_\cL=\eta_{ij}\,\big(\de q^i \otimes \de q^j + \tau^i \otimes \tau^j\big) \ ,$$ so that for a regular Lagrangian it defines a Riemannian metric on $T\cQ$.

It follows that the Lagrangian almost para-K\"ahler structure and
almost K\"ahler structure have the same fundamental 2-form. With this
data it is then straightforward to show that the almost para-K\"ahler structure\footnote{It is redundant here to write all three of the tensors, but we wish to stress what objects are involved in the structures under consideration.} $(K, \eta_\cL, \Omega_{\mathcal{L}})$ and the almost K\"ahler structure $(I, \mathcal{H}_\cL,\Omega_{\mathcal{L}})$ on $T\cQ$ satisfy the relations $$\eta^{-1}_\cL\,\mathcal{H}_\cL= \mathcal{H}^{-1}_\cL\,\eta_\cL \qquad \mbox{and} \qquad \Omega_{\mathcal{L}}^{-1}\,\mathcal{H}_\cL=-\mathcal{H}_\cL^{-1}\,\Omega_{\mathcal{L}} \ .$$ Thus $(\eta_\cL, \Omega_{\mathcal{L}}, \mathcal{H}_\cL)$ is a {Born geometry} and the tangent bundle $T\cQ$ of the configuration space $\cQ$, for a dynamical system arising from a regular Lagrangian, is a {Born manifold}. The chiral structure $(\eta_\cL, J_\cL)$, introduced in the usual way by $J_\cL=\eta_\cL^{-1}\,\mathcal{H}_\cL,$ makes the triple $(I,J_\cL,K)$ an \emph{almost para-quaternionic structure} \cite{ivanov,Borngeometry}.  
Generalizing this construction, we can infer that any almost para-Hermitian structure and almost Hermitian structure {having the same fundamental 2-form} and the same splitting of the tangent bundle give rise to a Born geometry.

\begin{example}\label{ex:Riemanniantangent}
A particular instance of this geometry \cite{calcvar} is given by 
geodesic motion on any Riemannian manifold $\cQ$ with metric tensor
$g=g_{ij}\, \de q^i\otimes\de q^j.$
In this case the Lagrangian function is given by
$\mathcal{L}=\frac{1}{2}\,g_{ij}(q) \,v^i\,v^j$ and the second order
vector field is $$\Sigma= v^i\,\frac{\partial}{\partial
  q^i}-{\mit\Gamma}^i_{km}(q) \,v^k\, v^m\, \frac{\partial}{\partial v^i}
\ ,$$ where ${\mit\Gamma}^i_{km}$ are the Christoffel symbols of the
Levi-Civita connection compatible with $g.$ The eigenbundle $L_{\tt h}(T\cQ)$ here is the horizontal
distribution of the Levi-Civita connection and has a local frame given by the vector fields
\be\nonumber
D_i = \frac\partial{\partial q^i} - {\mit\Gamma}^k_{ij}(q) \, v^j\,
\frac\partial{\partial v^k} \qquad \mbox{with} \quad [D_i,D_j] =
R^k{}_{ijm}(q) \, v^m\, \frac\partial{\partial v^k} \ ,
\ee
where $R^k{}_{ijm}$ are the components of the Riemann curvature tensor
of $g$; in other words, the horizontal distribution is locally spanned
by the tangent vectors of the paths in $T\cQ$ defined via parallel
transport of a vector $v\in T_q\cQ$ along paths through $q$ in $\cQ$. Thus
$(T\cQ,K,\eta_\cL)$ is a para-K\"ahler manifold if and only if $g$ is a flat metric.
In general $L_{\tt v}(T\cQ)$ is, as always,
canonically identified with the tangent bundle of the space of velocities, while here the
inverse metric $g^{-1}=g^{ij}\,\frac\partial{\partial
  q^i}\otimes\frac\partial{\partial q^j}$ identifies $L_{\tt h}(T\cQ)$ with the cotangent
bundle. The fundamental Lagrangian 2-form is
\be\nonumber
\Omega_{\cL}= g_{ij}(q) \, \de q^i\wedge\tau^j \qquad \mbox{with}
\quad \tau^i = \de v^i+{\mit\Gamma}^i_{jk}(q) \, v^k\, \de q^j \ ,
\ee
while the generalized metric is the Sasaki metric on $T\cQ$~\cite{Vai1}:
\be\nonumber
\cH_\cL = g_{ij}(q)\, \big(\de q^i \otimes \de q^j + \tau^i \otimes
\tau^j\big) \ ,
\ee
This is similar to the example presented in \cite{freidel}, in
which an arbitrary connection compatible with $g$ is used in the
definition of the horizontal sub-bundle, with the difference that in
our case $K$ is always an almost para-K\"ahler structure for any
choice of metric $g$.
\end{example}

It also follows generally that the D-bracket and the bracket associated to the
Levi-Civita connection compatible with $\eta_{\cL}$ coincide, since $(T\cQ,
K, \eta_{\cL})$ is almost para-K\"ahler. 

\subsection{Para-K\"ahler Geometry of Phase Spaces}

We shall now describe how the Legendre transform of a Lagrangian function allows one to import the para-K\"ahler structure from the tangent bundle $T\cQ$ to the cotangent bundle $T^*\cQ$ of a configuration space $\cQ$.
Assuming the existence of a regular Lagrangian also implies that the Legendre transform is well-defined as the fiber derivative of the Lagrangian function $\mathcal{L}$, hence all the structures defined thus far can also be introduced on the cotangent bundle $T^*\cQ.$ Let us recall the definition of Legendre transform \cite{Marsden}, for a more general statement see \cite{Marmo}.
\theoremstyle{definition}
\begin{definition}
The \emph{Legendre transform} of a Lagrangian function $\mathcal{L}\in C^{\infty}(T\cQ)$ is the fiber derivative $$\mathrm{F}\mathcal{L}: T\cQ \longrightarrow T^*\cQ$$ given, at any point $q\in \cQ,$ by $$\big(\mathrm{F}\mathcal{L}(v)\big)(z)=\frac{\de}{\de t}\mathcal{L}(q, v+t\,z)\Big\rvert_{t=0} \qquad \mbox{with} \quad \mathrm{F}\mathcal{L}(v)\in T^*_q \cQ\ , $$ for all $v,z \in T_q \cQ$. This transformation is fiber preserving, and in a local chart it reads
$$\mathrm{F}\mathcal{L}: (q^i, v^i) \longmapsto (q^i, p_i) \qquad \mathrm{with} \quad p_i := \big(\mathrm{F}\mathcal{L}(v)\big)_i= \frac{\partial \mathcal{L}}{\partial v^i} \ .$$
\end{definition}
If $\mathcal{L}$ is a regular Lagrangian, then the Legendre transform
defines a local diffeomorphism between $T\cQ$ and $T^*\cQ.$  A regular
Lagrangian for which this diffeomorphism is globally defined will be called \emph{hyper-regular}. From now on, we will assume hyper-regularity of the Lagrangian.

The pushforward $\mathrm{F}\mathcal{L}_* : \Gamma(T(T\cQ)) \rightarrow \Gamma(T(T^*\cQ))$ in local coordinates, for any vector field $X= X^i\, \frac{\partial}{\partial q^i} + \tilde X^i\, \frac{\partial}{\partial v^i} \in \Gamma(T(T\cQ)),$ is given by $$X \longmapsto X_{\mathcal{L}}=X^i\, \frac{\partial}{\partial q^i} + \pounds_X \Bigl( \frac{\partial \mathcal{L}}{\partial v^i}\Bigr)\bigg\rvert_{(q,p)}\, \frac{\partial}{\partial p_i} \ .$$  
Using this definition, it is easy to check that there exists a unique
1-form $\theta_0 \in \Omega^1(T^*\cQ)$ such that
$(\mathrm{F}\mathcal{L})^* \theta_0 = \theta_{\mathcal{L}},$ where
$\theta_{\mathcal{L}}$ is the Cartan 1-form on $T\cQ.$ In local
coordinates it is given by $\theta_0 = p_i\, \de q^i.$ A similar
statement holds for the Lagrangian 2-form $\Omega_{\mathcal{L}}$:
There exists a unique closed 2-form $\omega_0$ on $T^*\cQ$ that pulls
back to $\Omega_{\mathcal{L}}$, which is just the canonical 2-form
given locally in Darboux coordinates by $$\omega_0=\de p_i \wedge \de
q^i \ .$$  

Having defined $\mathrm{F}\mathcal{L}_*$, we can also show how the
splitting of $T(T\cQ)$ pushes forward to a splitting of $T(T^*\cQ)$:
The basis vectors $\frac{\partial}{\partial v^i}$ locally spanning
$\Gamma (L_{\tt v}(T\cQ))$ push forward to $V_i= \frac{\partial^2
  \mathcal{L}}{\partial v^i\, \partial v^j}\big\rvert_{(q,p)}\,
\frac{\partial}{\partial p_j}.$ Since $\mathcal{L}$ is regular, the
matrix $\big(\frac{\partial^2 \mathcal{L}}{\partial v^i \partial
  v^j}\big\rvert_{(q,p)}\big)$ acts as a ${\sf
  GL}(d,\IR)$-transformation of the vertical distribution, hence the
basis spanning the vertical sub-bundle of $T(T^*\cQ)$ can be written
as $Q^i=\frac{\partial}{\partial p_i}$; we will see that using this
transformed basis does not change the almost para-complex structure
$K$ on $T^*\cQ$, as expected. Thus $\mathrm{F}\mathcal{L}$ preserves verticality. Similarly, the horizontal distribution on $T(T\cQ)$ pushes forward to the horizontal distribution on $T(T^*\cQ)$ spanned by 
$$H_i=\frac{\partial}{\partial q^i} + \pounds_{D_i}\Bigl( \frac{\partial \mathcal{L}}{\partial v^j} \Bigr)\bigg\rvert_{(q,p)}\,\frac{\partial}{\partial p_j}=:P_i + N_{ij}\,Q^j \ ,$$
which defines a non-linear connection. It is shown in \cite{Hamiltonrom} that the connection coefficients $N_{ij}$ are symmetric for a hyper-regular Lagrangian. This will play a crucial role for the interpretation given in the following.

We can also write down the corresponding dual 1-forms $$ \zeta_i= \de p_i - N_{ji}\,\de q^j \qquad \mbox{and} \qquad \de q^i \ ,$$ so that the almost para-complex structure induced by the splitting on $T(T^*\cQ)$ is 
\be
K_N=Q^i \otimes \zeta_i - H_i \otimes \de q^i= Q^i \otimes \de p_i - P_i \otimes \de q^i -2\, N_{ji}\,P^i \otimes \de q^j \ .\label{kleg}
\ee 
Note that, had we considered the ${\sf GL}(d,\IR)$-transformed vertical basis, $K_N$ would not have changed, since the vertical dual 1-form would have changed by the inverse matrix of the transformed vertical basis.

We now need to take a step back and describe the natural para-K\"ahler structures of the cotangent bundle, in order to show how they are intrinsically related to this construction.
We have already discussed above the natural splitting of the tangent bundle of any fiber bundle $E\to \cQ$ induced by its projection map. For the cotangent bundle $T^*\cQ,$ this implies that $T(T^*\cQ)=L^0_{\tt v}(T^*\cQ)\oplus L^C_{\tt h}(T^*\cQ),$ which in local Darboux coordinates is given by
 $$\Gamma\big(L^0_{\tt v}(T^*\cQ)\big)={\rm Span}_{C^\infty(T^*\cQ)}\big\{ Q^i\big\} \qquad \mbox{and} \qquad \Gamma\big(L^C_{\tt h}(T^*\cQ)\big)={\rm Span}_{C^\infty(T^*\cQ)}\big\{ E_i=P_i + C_{ij}\,Q^j \big\} \ ,$$
where $C: \Gamma(T(T^*\cQ)) \rightarrow \Gamma(L^0_{\tt v}(T^*\cQ))$ is any
local map. This is the most general form that can be assigned to the
horizontal distribution, since $C$ is defined in each patch but does
not necessarily transform as a tensor. This is also called a
non-linear connection, since it is obtained by the requirement that
the horizontal sub-bundle is annihilated by vertical distributions,
but without further assumptions,\footnote{For instance, we may assume
  that the vertical 1-forms transform in a prescribed way under the adjoint action of a subgroup of ${\sf GL}(d, \mathbb{R}),$ which is another way to define a linear connection.} a freedom of choice is left in the definition of a vertical 1-form that is given by all the possible choices of $C.$  Hence the almost para-complex structure on $T^*\cQ$ is locally given by 
\be
K_C= Q^i \otimes \tau_i - E_i \otimes \de q^i \ , \label{Npro}
\ee
where $\tau_i= \de p_i - C_{ji}\,\de q^j$ and $\de q^i$ are the dual 1-forms to $Q^i$ and $ E_i$ respectively.

Since $T^*\cQ$ also has a natural symplectic structure given by the canonical 2-form $\omega_0=\de p_i \wedge \de q^i,$ we may ask if there exists a compatibility condition between $K_C$ and $\omega_0$ which endows $T^*\cQ$ with the structure of an almost para-K\"ahler manifold. In other words, we want to construct a metric $\eta_C(X,Y)= \omega_0 (K_C(X),Y)$ for which $L^0_{\tt v}(T^*\cQ)$ and $L^C_{\tt h}(T^*\cQ)$ are maximally isotropic distributions.
We first compute the $(0,2)$-tensor $\eta_C$ by expressing $K_C$ in Darboux coordinates to get
\be
\eta_C=\omega_0\, K_C= \de q^i \otimes \de p_i + \de p_i \otimes \de q^i -2\, C_{ji}\,\de q^i \otimes \de q^j \ .\label{Nmet}
\ee
We notice that $\eta_C$ is non-degenerate because $\omega_0$ and $K_C$ are. We can easily see that the vertical distribution $L^0_{\tt v}(T^*\cQ)$ is isotropic with respect to $\eta_C$: $$\eta_C(X_{\tt v},Y_{\tt v})=\eta_C(Y_{\tt v},X_{\tt v})=0 \ .$$ We further have $$\eta_C(X_{\tt h},Y_{\tt v})=\eta_C(Y_{\tt v},X_{\tt h})= (X_{\tt h})^i \, (Y_{\tt v})_i \ ,$$ where $X_{\tt h}= (X_{\tt h})^i\, E_i\in \Gamma(L^C_{\tt h}(T^*\cQ))$ and 
$Y_{\tt v}=(Y_{\tt v})_i\, Q^i \in \Gamma(L^0_{\tt v}(T^*\cQ)).$
Requiring isotropy of the horizontal distribution, we obtain
$$\eta_C(E_i, E_j)=C_{ij}+C_{ji}-2\,C_{ji}=0 \ .$$ Hence $C_{ij}$ must be symmetric, or equivalently the map $$\eta_C\, C: \Gamma\big(T(T^*\cQ)\big)\times \Gamma\big(T(T^*\cQ)\big) \longrightarrow C^{\infty}(T^*\cQ)$$ must be symmetric in each local trivialization. This condition also implies the symmetry of the tensor $\eta_C$, thus $\eta_C$ is a metric for which $L^0_{\tt v}(T^*\cQ)$ and $L^C_{\tt h}(T^*\cQ)$ are maximally isotropic sub-bundles. Thus with these conditions, $(T^*\cQ,K_C,\eta_C)$ is an $L_{\tt v}^0(T^*\cQ)$-para-K\"ahler manifold.

This result may be interpreted as follows. Any collection of locally
symmetric $(0,2)$-tensor fields $C$ on $T^*\cQ$ corresponds to a
different splitting of $T(T^*\cQ)$ and a different metric $\eta_C$,
which together give an almost para-K\"ahler structure for which the
canonical 2-form $\omega_0$ on $T^*\cQ$ is the fundamental 2-form. We
say that such an almost para-K\"ahler structure on $T^*\cQ$ is
$\omega_0$-\emph{compatible}. The $\omega_0$-compatible 
para-K\"ahler structure corresponding to $C=0$ is called
\emph{canonical}, since it has the Levi-Civita connection as its
canonical para-Hermitian connection.\footnote{We consider the
  canonical splitting here up to constant $C$.}  In this case both
$L^0_{\tt v}(T^*\cQ)$ and $L_{\tt h}^0(T^*\cQ)$, given respectively by
$\Gamma(L^0_{\tt v}(T^*\cQ))={\rm Span}_{C^\infty(T^*\cQ)}\{Q^i\}$ and $\Gamma(L_{\tt
  h}^0(T^*\cQ))={\rm Span}_{C^\infty(T^*\cQ)}\{P_i\},$ are integrable distributions
and Lagrangian with respect to $\omega_0.$ The para-complex structure
is locally given by $K_0=Q^i \otimes \de p_i - P_i \otimes \de q^i$
and $\eta_0= \de q^i \otimes \de p_i + \de p_i \otimes \de q^i$ is a
flat metric. The deformation of the canonical splitting by
$C:\Gamma(L_{\tt h}^0)\to\Gamma(L_{\tt v}^0)$ is somewhat analogous to
the deformations of almost para-Hermitian structures that we discussed
in Section~\ref{sec:Btransformations}, with the important difference
that here the deformation is not realized by an ${\sf O}(d,d)(T^*\cQ)$-transformation.

Going back to the splitting induced by a dynamical system, we see that
the almost para-complex structure $K_N$ from \eqn{kleg} is compatible
with $\omega_0$ and together they define an almost para-K\"ahler
structure on $T^*\cQ$; in this case we identify $C_{ij}=N_{ij}$. Let
us also stress that a bijective Legendre transform maps a Lagrangian
almost para-K\"ahler structure on $T\cQ$ into one of the members of
the class of $\omega_0$-compatible almost para-K\"ahler structures on
$T^*\cQ,$ and the generalized metric $\cH_{\cL}$ on $T\cQ$ to a
generalized metric $\cH_C$ on $T^*\cQ$.

\begin{example}
As in the setting of Example~\ref{ex:Riemanniantangent}, consider the case that $(\cQ,g)$ is a Riemannian manifold of dimension $d$ with the hyper-regular Lagrangian function $\cL=\frac12\, g_{ij}\, v^i\, v^j$. In this case the connection coefficients
\be\nonumber
C_{ij}=N_{ij}={\mit\Gamma}_{ij}^k(q) \, p_k
\ee
coincide with the Christoffel symbols and again the 
sub-bundle $L_{\tt h}^C(T^*\cQ)$ is the horizontal distribution of the Levi-Civita
connection of $g$. As previously $(T^*\cQ,K_C,\eta_C)$ is a
para-K\"ahler manifold if and only if $g$ has vanishing
curvature. A generalized metric on $T^*\cQ$ is defined
in the generic case by the Sasaki metric~\cite{Vaisman2013}
\be\nonumber
\cH_C = g_{ij}(q) \, \de q^i\otimes\de q^j + g^{ij}(q) \,
\tau_i\otimes\tau_j \qquad \mbox{with} \quad \tau_i = \de p_i -
{\mit\Gamma}^k_{ij}(q) \, p_k\, \de q^j \ .
\ee
This is a special case of the 6-parameter family of natural almost
para-Hermitian structures constructed in~\cite{Romaniuc2013} as the general natural lifts of the metric $g$ from the configuration manifold $\cQ$ to the total space $T^*\cQ$ of its cotangent bundle.
\end{example}

\section{Dynamical Nonassociativity and Generalized Fluxes\label{sec:nonassociativity}} 

In this section we will describe some examples of how the description
of dynamical systems in terms of para-Hermitian geometry goes beyond
systems admitting Lagrangian functions which are regular; in these
settings the almost para-K\"ahler structures are relaxed to almost
para-Hermitian structures. In
particular, we look at certain dynamical systems which do not even
admit a Lagrangian function due to the presence of fluxes which induce
a nonassociative deformation of the phase space Poisson algebra. This
will pave the way, within a physically elementary setting, to a
general understanding of how to incorporate geometric and
non-geometric fluxes as deformations of the local para-K\"ahler geometry of the phase space $T^*\cQ$ of any configuration space $\cQ$. 

\subsection{Para-Hermitian Geometry of a Non-Lagrangian System}\label{beyond}

We consider a particular dynamical system in which non-regularity means the lack of a globally defined (regular) Lagrangian or Hamiltonian function. Our goal is to demonstrate how a para-Hermitian structure can be introduced on the cotangent bundle in order to compensate the lack of a regular Hamiltonian, and obtain a geometric description of the dynamics of this system. The configuration space is $\cQ=\mathbb{R}^d$ with $d\geq3$ and the equations of motion are given by
\be
\frac{\de^2 q^i}{\de t^2}=2\,\delta^{ik}\, B_{kj}\, \frac{\de q^j}{\de t} \ , \label{magn}
\ee
where $B=\frac12\, B_{ij}(q)\,\de q^i\wedge\de q^j$ is any 2-form on $\cQ$ (regarded as a skew-symmetric map). For $d=3$ this is the Lorentz force law describing the dynamics of a classical spinless point particle with unit mass and electric charge moving in a magnetic field $\mathcal{B}^i=\varepsilon^{ijk}\,B_{jk}$, where the force exerted by the electric field of the charged particle is neglected; our main interest is the case where $\mathcal{B}$ is generated by a smooth distribution of magnetic monopoles (see e.g.~\cite{Jackiw1984,Gunaydin1985,lustcoc,Bojowald2014,Bojowald2015,Szabo2017,kup,Bunk2018} for other treatments of this dynamical system). In the following we work in arbitrary dimensionality since our later considerations will be naturally adapted to this general setting, but the reader interested in concrete examples may wish to keep this special case in mind.

By lifting the dynamics on the cotangent bundle $T^*\cQ$, we obtain the second order vector field                              
\be\nonumber
\Sigma=\delta^{ij}\, p_j \, \frac{\partial}{\partial q^i}+ 2\,\delta^{jk}\, B_{ij}\, p_k \, \frac{\partial}{\partial p_i} \ \in \ \Gamma\big(T(T^*\cQ)\big) \ .
\ee
The corresponding integral curves yield the system of first order differential equations $$ \frac{\de q^i}{\de t}=\delta^{ij}\, p_j \qquad \mbox{and} \qquad \frac{\de p_i}{\de t}= 2\,B_{ij}\, \delta^{jk}\, p_k \ , $$ which is equivalent to \eqn{magn}. A similar expression is obtained for the lift on the tangent bundle $T\cQ.$ However, such dynamical systems do not generally admit any (global) Lagrangian or Hamiltonian function, hence a fiber derivative connecting the two lifts cannot be defined as previously. 

The splitting $T(T^*\cQ)= L^0_{\tt v}(T^*\cQ)\oplus L^B_{\tt h}(T^*\cQ)$ induced by $\Sigma$ is given by
\begin{align*}
\Gamma\big(L^0_{\tt v}(T^*\cQ)\big)&={\rm Span}_{C^\infty(T^*\cQ)}\Big\{
Q^i=\frac{\partial}{\partial p_i}\Big\} \ , \\[4pt] \Gamma\big(L^B_{\tt
  h}(T^*\cQ)\big)&={\rm Span}_{C^\infty(T^*\cQ)}\big\{D_i= P_i +
B_{ij}\,Q^j\big\} \ ,
\end{align*}
with $P_i=\frac\partial{\partial q^i}$.
The Lie algebra defining these distributions is 
$$[D_i, D_j]=\big(\partial_i B_{jk}-\partial_j B_{ik}\big)\,
Q^k \ , \quad [D_i , Q^j]=0 \qquad \mbox{and} \qquad [Q^i,
Q^j]=0 \ ,$$ where $\partial_i$ denotes the partial derivative with
respect to $q^i$. This shows that the horizontal distribution $L^B_{\tt
  h}(T^*\cQ)$ is not involutive, while the vertical distribution $L^0_{\tt v}(T^*\cQ)$
is integrable and can be identified with the tangent bundle of the
fibers of $T^*\cQ$.

The respective dual 1-forms to the basis vector fields $Q^i$ and $D_i$ are 
$$ \tilde\theta_i= \de p_i + B_{ij}\, \de q^j \qquad \mbox{and} \qquad \de q^i \ ,$$
thus we can write the almost para-complex structure $K_B$ defined by the splitting $T(T^*\cQ)=L^0_{\tt v}(T^*\cQ)\oplus L^B_{\tt h}(T^*\cQ)$, i.e. such that $K_B|_{L^0_{\tt v}(T^*\cQ)}=\unit$ and $K_B|_{L^B_{\tt h}(T^*\cQ)}=-\unit$:
\be
K_B= Q^i \otimes \tilde\theta_i - D_i \otimes \de q^i= Q^i \otimes \de p_i - P_i \otimes \de q^i + 2\, B_{ij}\,Q^i \otimes \de q^j \ . \label{almK}
\ee
We now define a Lorentzian metric $\eta_B$ on $T^*\cQ$ in order to obtain a suitable almost symplectic 2-form which gives the canonical equations of motion. The introduction of such a metric can be regarded as a way to get around the problem of the non-existence of a global Hamiltonian.
The flat Lorentzian metric with $\eta_B(X_{\tt v},Y_{\tt
  h})=\eta_B(Y_{\tt h},X_{\tt v})$ and $\eta_B(X_{\tt v},Y_{\tt
  v})=\eta_B(X_{\tt h},Y_{\tt h})=0$ in local coordinates
reads $$\eta_B= \de q^i \otimes \tilde\theta_i + \tilde\theta_i \otimes \de q^i \ ,$$
and may be regarded as a lift of the natural flat Euclidean metric
defined on the configuration space $\cQ=\mathbb{R}^d$. Then $(T^*\cQ, K_B,
\eta_B)$ is an $L^0_{\tt v}$-para-Hermitian manifold. We stress that a lack of a Hamiltonian or Lagrangian function translates into a weakening of the properties of the carrier manifold $M$, i.e. $T\cQ$ (or $T^*\cQ$) endowed with a regular Lagrangian (or Hamiltonian) is an almost para-K\"ahler manifold, while $T\cQ$ (or $T^*\cQ$) without a regular function is only an almost para-Hermitian manifold.

Our main goal now is to obtain the almost symplectic 2-form describing the coordinate algebra on the phase space from the geometry of the phase space itself. In this case, such a 2-form is given by the fundamental 2-form of para-Hermitian geometry, i.e. $\omega_B(X,Y)=\eta_B(K_B(X),Y).$ In local coordinates $(q^i, p_i)$ it reads 
\be
\omega_B= \tilde\theta_i \wedge \de q^i= \de p_i \wedge \de q^i + 2\, B_{ij}\, \de q^i \wedge \de q^j \ , \label{omb}
\ee
and its inverse leads to the coordinate algebra 
$$ \{ q^i, q^j\}_B=0 \ , \quad \{ q^i, p_j \}_B= \delta^i{}_j \qquad \mbox{and} \qquad \{p_i, p_j\}_B=2\, B_{ij}(q) \ .$$
These define twisted Poisson brackets which have non-zero Jacobiators amongst the
fiber momentum coordinates given by $$\{ p_i, p_j, p_k\}_B=
3\,H_{ijk}(q) \ ,$$ where $H=\de B=\frac1{3!}\, H_{ijk}(q)\,\de
q^i\wedge\de q^j\wedge\de q^k$ is a 3-form on $\cQ$ with
$H_{ijk}=\partial_{[i}B_{jk]}$. The nonassociativity of the coordinate
algebra is related to the lack of closure of the fundamental 2-form
$\omega_B$: $$\de \omega_B= 2\,H\ ,$$ in which we see the emergence of
$H$-flux.\footnote{It is natural to think of this flux as the geometric NS--NS $H$-flux, and later on it will indeed be identified in that way; for $d=3$ the $H$-flux can be
  interpreted as a field of magnetic charges in the present context.} This algebra is
associative only when the $H$-flux vanishes; for $d=3$ this is the
classical Maxwell theory, where $\partial_i \mathcal{B}^i=0.$ In general, the dynamical vector field $\Sigma$ is Hamiltonian with respect to $\omega_B$ for the locally defined Hamiltonian function $\mathcal{E}=\frac{1}{2}\, \delta^{ij}\, p_i\, p_j.$

In the setting of para-Hermitian geometry, the nonassociativity of the coordinate algebra means the violation of the weak integrability condition (Definition \ref{weint}).  In order to show how locally non-geometric fluxes obstruct the relative weak integrability in this example, we compute the D-bracket on $(T^*\cQ, K_B, \eta_B)$ and the bracket associated to the Levi-Civita connection compatible with $\eta_B$ (which is the D-bracket when $B=0$),  and then compare them using \eqn{dciv}. Here $\eta_B$ is a flat metric, hence the Levi-Civita connection has vanishing Christoffel symbols, and we also know that $\llbracket Q^i, Q^j \rrbracket^{\tt D}= \llbracket Q^i, Q^j \rrbracket^{{\nabla}^{\tt LC}}=0$ from Proposition~\ref{flatd}, i.e. $L^0_{\tt v}(T^*\cQ)$ is a weakly integrable distribution. By computing
$$\eta_B\big({\nabla}^{\tt LC}_{D_i}D_j-{\nabla}^{\tt LC}_{D_j}D_i, Z\big)=  Z^k\,\big(\partial_i B_{jk}-\partial_j B_{ik}\big) \qquad \mbox{and} \qquad \eta_B\big({\nabla}^{\tt LC}_Z D_i, D_j\big)= Z^k\, \partial_k B_{ij} \ ,$$ where $Z=Z^i\,D_i+\tilde Z_i\,Q^i\in \Gamma(T(T^*\cQ))$, we find that the bracket associated to the Levi-Civita connection for two horizontal basis elements is given by
\be\nonumber
\eta_B\big(\llbracket D_i, D_j \rrbracket^{{\nabla}^{\tt LC}}, Z\big)= H_{ijk} \, Z^k \ .
\ee
Considering the canonical para-Hermitian connection, which for the almost para-complex structure $K_B$ is not the Levi-Civita connection, the D-bracket is given by
\be\nonumber
\eta\big(\llbracket D_i, D_j \rrbracket_B^{\tt D}, Z\big)=0 \ ,
\ee
and thus $L^B_{\tt h}(T^*\cQ)$ is also weakly integrable with respect to $K_B$, as it should be. The difference between the D-bracket and the ${\nabla}^{\tt LC}$-bracket is exactly measured by $\de \omega_B=2H,$ as formulated by~\eqn{dciv}. 

Following \cite{svoboda}, we may give a new perspective on this
dynamical nonassociativity, based on the flux deformations of almost
para-Hermitian structures that we discussed in
Section~\ref{sec:ParaHermitianGeometry}. The almost para-Hermitian
structure $(K_B, \eta_B)$ on $T^*\cQ$ can be regarded as a deformation
via a $B_-$-transformation of the canonical para-K\"ahler structure
$(K_0, \eta_0),$ where $\eta_B=\eta_0$ since by definition $e^{B_-} \in {\sf O}(d,d)(T^*\cQ)$, whereas the closure of $\omega_0$ is no longer preserved by $\omega_B=\eta_B\, K_B.$
In the present case the map $B_-: \Gamma(L_{\tt h}^0(T^*\cQ)) \rightarrow
\Gamma(L^0_{\tt v}(T^*\cQ))$ is defined by $$B_-= B_{ij}\, Q^i
\otimes \de q^j$$ and it satisfies the skew condition $$b_-=\eta_B\,
B_-= B_{ij}\, \de q^i \wedge \de q^j \ ,$$ with $b_-$ having only a
$(+0,-2)$-component with respect to the canonical
splitting\footnote{In Section~\ref{sec:ParaHermitianGeometry} we
  considered a $B_+$-transformation associated to a $(+2,-0)$-form
  $b_+$. In the present case we consider instead a
  $B_-$-transformation which is associated to a 2-form $b_-$ on the
  base manifold $\cQ$, hence it becomes a $(+0,-2)$-form since the
  coordinates on the base manifold are the adapted coordinates of the
  horizontal eigenbundle $L_{\tt h}^0,$ which has eigenvalue $-1$.}
$T(T^*\cQ)=L^0_{\tt v}(T^*\cQ)\oplus L^0_{\tt h}(T^*\cQ)$. In this sense the
horizontal distribution of the dynamical splitting of $T(T^*\cQ)$ can
be regarded as the graph $\Gamma(L^B_{\tt
  h}(T^*\cQ))=\mathrm{Graph}_{T(T^*\cQ)}(B_-)=\{Z+B_-(Z):Z\in \Gamma(T(T^*\cQ))\},$ using standard terminology from generalized geometry.
Thus  $K_B=e^{B_-} \, K_0 \, e^{-B_-}$ in \eqn{almK} is exactly given by \eqn{matrK}, and the fundamental 2-form is given by $\omega_B=\omega_0+2b_-$ as confirmed by its local form \eqn{omb}. This description also confirms that fluxes are generally a relative notion obstructing the compatibility of two (almost) para-Hermitian structures in the form of structure constants of the D-bracket algebra of vector fields.

\subsection{Born Reciprocity and the $R$-Flux Model}\label{sec:R-flux}

An important application of the non-Lagrangian dynamical system discussed above comes from applying the duality transformation $(q^i, p_i) \mapsto (p_i, -q^i)$ of order~4 \cite{lustcoc}, sometimes called \emph{Born reciprocity}. 
Born reciprocity is a symplectomorphism of the canonical 2-form
$\omega_0$ on $T^*\cQ,$ but it does not preserve the canonical
para-K\"ahler structure because it sends $K_0\mapsto -K_0$ and
$\eta_0\mapsto-\eta_0.$ In fact, this transformation can be understood
as a composition of deformations of para-Hermitian structures: It
sends the 2-form $B=\frac12\,B_{ij}(q)\, \de q^i\wedge\de q^j$ on the 
configuration space $\cQ$ to the 2-form $\beta=\frac12\,\beta^{ij}(p)\,
\de p_i\wedge \de p_j$ on the fiber spaces of the cotangent bundle
$\pi:T^*\cQ\to \cQ$, and correspondingly the $H$-flux
$H_{ijk}=\partial_{[i}B_{jk]}$ to the $R$-flux\footnote{For example,
  in $d=3$ dimensions the 2-form $B_{ij}=\frac13\,\rho\,
  \varepsilon_{ijk}\, q^k$, which can be interpreted as a magnetic
  field sourced by a
  uniform distribution $\rho$ of magnetic charges, is mapped to
  $\beta^{ij}=\frac13\,\rho\,\varepsilon^{ijk}\, p_k$.}
$R^{ijk}=\tilde\partial{}^{[i}\beta^{jk]}$, where $\tilde\partial^i$
denotes the partial derivative with respect to $p_i$. It therefore sends the map $B_-: \Gamma(L_{\tt h}^0)(T^*\cQ) \rightarrow \Gamma(L^0_{\tt v})(T^*\cQ)$ to the map $\beta_+: \Gamma(L^0_{\tt v})(T^*\cQ) \rightarrow \Gamma(L_{\tt h}^0)(T^*\cQ)$ defined by $\beta_+= \beta^{ij}\, P_i \otimes \de p_j$.
From a more general point of view, we may define Born reciprocity in
para-Hermitian geometry as a morphism sending a $B_-$-transformation
into a $\beta_+$-transformation, where a $\beta_+$-transformation is
any map $\beta_+: \Gamma(L^0_{\tt v}(T^*\cQ)) \rightarrow \Gamma(L_{\tt h}^0(T^*\cQ))$
whose composition with the metric $\eta_0$ is a 2-form.\footnote{In
  the present case a $\beta_+$-transformation is associated to a
  2-form along the fibers, hence it is a $(+2,-0)$-form with respect
  to the canonical splitting.} $\beta_+$-transformations will be discussed more generally later on.

In terms of its action on the almost para-Hermitian structure $(K_B,
\eta_B)$ on $T^*\cQ$, Born reciprocity can be regarded as the change of polarization $\vartheta=e^{\beta_+}\, e^{-B_-}$:
$$\vartheta:K_B \xrightarrow{e^{-B_-}} K_0 \xrightarrow{\ e^{\beta_+}\
} K_{\beta} \ ,$$ where $\beta$ is a bivector field of type $(+0,
-2)$, or equivalently a 2-form of type $(+2, -0)$ depending only on
the fiber directions.\footnote{Here the bigrading is with respect to
  the canonical para-K\"ahler structure.} This change of polarization exchanges the role of the twisted distribution between the horizontal and vertical sub-bundle of the canonical para-K\"ahler structure: There is a new splitting $T(T^*\cQ)=L_{\tt v}^{\beta}(T^*\cQ)\oplus L_{\tt h}^0(T^*\cQ)$, where the vertical distribution $L^0_{\tt v}(T^*\cQ)$ is twisted to $L^\beta_{\tt v}(T^*\cQ)$ with $$\Gamma\big(L_{\tt v}^{\beta}(T^*\cQ)\big)= {\rm Span}_{C^\infty(T^*\cQ)}\big\{ \tilde D^i=Q^i + \beta^{ij}\,P_j\big\} \ ,$$ while the horizontal distribution $L_{\tt h}^0(T^*\cQ)$ is unchanged. Thus the Lie algebra of the new splitting is  
$$[P_i, P_j]= 0 \ , \quad [P_i, \tilde D^j]=0 \qquad \mbox{and} \qquad [\tilde D^i,\tilde D^j]=\big(\tilde{\partial}^i \beta^{jk}-\tilde{\partial}^j \beta^{ik}\big)\, P_k \ .$$
The twisted sub-bundle is still maximally isotropic with respect to the metric $\eta_0,$ thus $(K_{\beta}, \eta_0)$ is an almost para-Hermitian structure on $T^*\cQ$ (a $\beta_+$-transformation of the canonical para-K\"ahler structure as we showed above) with fundamental 2-form $$\omega_{\beta}= \de p_i \wedge \de q^i + 2\, \beta^{ij}\, \de p_i \wedge \de p_j= \omega_0 + 2\,\tilde{b}_+ \ , $$ where 
$\tilde{b}_+= \eta_0 \, \beta_+ = \beta^{ij}\, \de p_i \wedge \de p_j$
is a $(+2,-0)$-form with respect to the canonical splitting. The
inverse of the 2-form $\omega_\beta$ yields the local coordinate
algebra with twisted Poisson brackets
\be\nonumber
\{q^i,q^j\}_\beta=2\,\beta^{ij}(p) \ , \quad \{q^i,p_j\}_\beta=\delta^i{}_j \qquad \mbox{and} \qquad \{p_i,p_j\}_\beta=0 \ ,
\ee
which together with the non-vanishing Jacobiators
\be\nonumber
\{q^i,q^j,q^k\}_\beta=3\,R^{ijk}(p)
\ee
exhibit a nonassociative deformation of the configuration space $\cQ$. 

This dynamical system is called the \emph{$R$-flux model} and is
purported to describe the phase space dynamics of closed strings propagating in
locally non-geometric $R$-flux
backgrounds~\cite{Lust2010,Blumenhagen2011,Mylonas2012,Aschieri2015}. It
is also possible to introduce a Born geometry on $T^*\cQ$ in this
setting, analogously to what we did in Section~\ref{sec:polarization}:
Starting from the generalized metric
$$
\widetilde{\cH}_0 = \bigg(\begin{matrix}
g^{-1} & 0 \\ 0 & g
\end{matrix} \bigg)
$$
with respect to the canonical polarization $T(T^*\cQ) = L^0_{\tt
  v}(T^*\cQ)\oplus L^0_{\tt h}(T^*\cQ)$, where $g$ is a Riemannian metric on $\cQ$, one
computes the change of metric in the new $\beta$-twisted polarization
to be
\be\label{eq:tildeHbeta}
\widetilde{\cH}_{\beta_+} = \big(e^{-\beta_+}\big)^{\rm t} \,
\widetilde{\cH}_0 \, e^{-\beta_+} = \bigg( \begin{matrix}
g^{-1} -\beta \, g \, \beta & - \beta\,g \\ g\,\beta & g
\end{matrix} \bigg) \ .
\ee
This is the correct global parameterization for the generalized metric in a
non-geometric polarization familiar from generalized geometry and
double field theory, which is a particular T-duality transformation of
the generalized metric \eqref{eq:geometricH} of a geometric polarization with $g_+=g$
and $b_+ = B$~\cite{Grana2008,Hull2009,Andriot2012}.

Repeating the analysis of Section~\ref{beyond}, we find that the ${\nabla}^{\tt LC}$-bracket describes the emergence of $R$-flux as an obstruction to the weak integrability of $L_{\tt v}^{\beta}(T^*\cQ)$ with respect to the canonical para-K\"ahler structure: $$\llbracket \tilde D^i, \tilde D^j \rrbracket^{{\nabla}^{\tt LC}}= R^{ijk}\, P_k \ ,$$ while the horizontal distribution is weakly integrable. The D-bracket with respect to $K_{\beta}$ vanishes, and the difference between the D-bracket and the ${\nabla}^{\tt LC}$-bracket is measured by $\de \omega_{\beta}= 2\, \de \tilde{b}_+=2R.$ 

Since the notion of flux in this context appears as an obstruction to relative
integrability of two (almost) para-Hermitian structures, we can also compute the D-bracket associated to the almost para-complex structure $K_B$ from \eqn{almK} of the vectors spanning $\Gamma(L_{\tt h}^0(T^*\cQ))$ and $\Gamma(L_{\tt v}^{\beta}(T^*\cQ)),$ which will demonstrate how the fluxes thus far obtained can be generalized. The D-bracket associated with $K_B$ is defined using the canonical para-Hermitian connection $$\nabla^{\tt can}=\Pi_+^B\, {\nabla}^{\tt LC}\,\Pi_+^B + \Pi_-^B\, {\nabla}^{\tt LC}\, \Pi_-^B$$ where 
\begin{align*}
\Pi^B_+&=\frac12\,\big(\mathds{1}+K_B\big)=Q^i\otimes \de p_i +B_{ij}\,Q^i \otimes \de q^j \ , \nonumber \\[4pt] \Pi^B_-&=\frac12\,\big(\mathds{1}-K_B\big)=P_i\otimes \de q^i +B_{ij}\, Q^j \otimes \de q^i \ ,
\end{align*}
and ${\nabla}^{\tt LC}$ is the Levi-Civita connection of $\eta_0$.
Then we obtain the D-brackets
\be 
\llbracket P_i, P_j \rrbracket^{\tt
  D}_B=\mathscr{H}_{ijk}\, \tilde D^k + {\mathscr{F}_{ij}}^k\, P_k
\qquad \mbox{and} \qquad \llbracket \tilde D^i, \tilde D^j \rrbracket^{\tt D}_B=
\mathscr{Q}^{ij}{}_k\, \tilde D^k +\mathscr{R}^{ijk}\, P_k \ , \nonumber
\ee
where 
\begin{align}
\mathscr{H}_{ijk}&=-3\,\partial_{[i}B_{jk]} \ , \nonumber \\[4pt]
{\mathscr{F}_{ij}}^k&=\beta^{km}\, \mathscr{H}_{mij} \ , \nonumber \\[4pt]
\mathscr{Q}^{ij}{}_k &= \beta^{im}\, \beta^{jl}\, \mathscr{H}_{mlk}+ B_{km}\,\tilde{\partial}^m \beta^{ij} \ , \nonumber \\[4pt]
\mathscr{R}^{ijk}&=3\, \tilde{\partial}^{[i}\beta^{jk]}+ 3\, B_{lm}\,\beta^{[i|l}\, \tilde{\partial}^m \beta^{|jk]}+ \beta^{il}\, \beta^{jm}\, \beta^{kn}\, \mathscr{H}_{lmn} \ . \nonumber 
\end{align}
These structure constants are precisely the generalized fluxes of double field theory associated to an NS--NS background written in a holonomic basis (see e.g.~\cite{membrane}), after noticing that $K_{\beta}$ can be obtained as a $(-B_-){+}\beta_+$-deformation of $K_{B}$, and requiring that $B$ and $\beta$ do not depend on $p_i$ and $q^i$ respectively. 

\subsection{Fluxes from $B$- and $\beta$-Transformations}

The $R$-flux model can be extended by considering more general $B+\beta$-transformations in order to formulate generalized fluxes as obstructions to compatibility between (almost) para-Hermitian structures.
Generalized fluxes on a cotangent bundle may be interpreted in the
sense of \cite{membrane}, where the cotangent bundle of an arbitrary
manifold is the doubled target space of a membrane sigma-model for
double field theory which involves geometric and non-geometric fluxes
as components of a generalized Wess-Zumino term in the membrane action. Here we will see how the complete expressions of $H$-, $f$-, $Q$- and $R$-fluxes in double field theory emerge from suitable twists of the canonical para-K\"ahler structure on $T^*\cQ,$ in local coordinates $(q^i, p_i),$ for any $d$-dimensional manifold~$\cQ$.

\paragraph{Geometric Fluxes from $\boldsymbol B$-Transformations.}
As a starting example, following \cite{svoboda} we may consider a slightly more general $B$-transformation than that considered in Section~\ref{beyond} by allowing a further dependence on the fiber coordinates $p_i$, and acting on the canonical para-K\"ahler structure on $T^*\cQ$ in the usual way. In this case the tensors $(K_B, \eta_0, \omega_B)$ take the same forms as given in Section~\ref{beyond}, with the differences that the Lie algebra of the two distributions is now
\begin{align}
[D_i, D_j]&=\big(\partial_i B_{jk}-\partial_j B_{ik}+ B_{il}\,\tilde{\partial}^l B_{jk}-B_{jl}\,\tilde{\partial}^l B_{ik}\big) \, Q^k \ , \nonumber \\[4pt]
 [ D_i,Q^j]&=-\tilde{\partial}^j B_{ik} \,
                            Q^k \ , \label{eq:LiealgHf} \\[4pt] [Q^i, Q^j]&=0 \ , \nonumber
\end{align}
and the flux of the fundamental 2-form becomes
\be
\de \omega_B=2 \, \big(\partial_i B_{jk}+ B_{im}\, \tilde{\partial}^m B_{jk}\big) \, \de q^i \wedge \de q^j \wedge \de q^k + 2 \, \tilde{\partial}^i B_{jk} \, \tilde\theta_i \wedge \de q^j \wedge \de q^k \ , \label{deomb}
\ee
when expressed in the splitting $ T(T^*\cQ)= L^0_{\tt v}(T^*\cQ)\oplus L^B_{\tt h}(T^*\cQ).$ The closure of $\omega_B$ is obstructed by a covariant $H$-flux and an $f$-flux with the bivector field $\beta$ set to zero.

By \eqn{deomb} the bracket associated to the Levi-Civita connection also changes: With $Z= Z^i\, D_i + \tilde{Z}_i\, Q^i \in \Gamma(T(T^*\cQ)),$ we compute
\begin{align*}
\eta\big({\nabla}^{\tt LC}_{D_i}D_j-{\nabla}^{\tt LC}_{D_j}D_i, Z\big) &= Z^k \, \big(\partial_i B_{jk}-\partial_j B_{ik}+B_{im}\,\tilde{\partial}^m B_{jk}-B_{jm}\,\tilde{\partial}^m B_{ik}\big) \ , \\[4pt]
\eta\big({\nabla}^{\tt LC}_Z D_i, D_j\big)&= Z^k \, \big(\partial_k B_{ij}+B_{km}\,\tilde{\partial}^m B_{ij}\big) + \tilde{Z}_k\, \tilde{\partial}^k B_{ij} \ ,
\end{align*}
to obtain
\be
\eta\big(\llbracket D_i, D_j \rrbracket^{{\nabla}^{\tt LC}}, Z\big)= 3\,Z^k \,  \big(\partial_{[i}B_{jk]}+B_{[i|m}\,\tilde{\partial}^m B_{|jk]}\big) + \tilde{Z}_k\,\tilde{\partial}^k B_{ij} \ . \label{dsec}
\ee
This still agrees with \eqref{deomb} (up to the usual factor $\frac{1}{2}$). The D-bracket again vanishes:
\be\nonumber
\eta\big(\llbracket D_i, D_j \rrbracket^{\tt D}_B, Z\big)= 0 \ ,
\ee
so that the obstruction to (relative) weak integrability of $L^B_{\tt h}(T^*\cQ)$ is characterized by the components of the covariant $H$-flux, i.e. the $H$-flux without the section condition. The horizontal component of the ${\nabla}^{\tt LC}$-bracket in \eqn{dsec} takes the form of the $f$-flux when the bivector field $\beta$ vanishes~\cite{membrane}. In other words, such fluxes also appear as the $(+0,-3)$-component of \eqn{deomb} with respect to $K_0$:
$$\de \omega_B^{(+0,-3)_0}=2 \, \big(\partial_i B_{jk}+ B_{im}\,\tilde{\partial}^m B_{jk}\big)\, \de q^i \wedge \de q^j \wedge \de q^k \ ,$$
and as the $(+1,-2)$-component
\be\nonumber
\de \omega_B^{(+1,-2)_0}= 2 \, \tilde{\partial}^i B_{jk} \, \de p_i \wedge \de q^j \wedge \de q^k \ .
\ee
In particular, this now implies that the Jacobiators $\{q^i,p_j,p_k\}_B$ are also non-vanishing with the additional nonassociativity induced by the dependence of the $f$-flux on the dual fiber coordinates $p_i$.
We can also see how the fluxes are related to the Lie algebra \eqref{eq:LiealgHf} of the maximally isotropic distributions $L_{\tt h}^B(T^*\cQ)$ and $L_{\tt v}^0(T^*\cQ)$: The $H$-flux appears as the vertical component of the Lie bracket of two horizontal basis vectors, while the $f$-flux is exactly given by the vertical component of the bracket of a vertical and a horizontal basis vector. The vertical distribution remains weakly integrable with vanishing D-bracket and ${\nabla}^{\tt LC}$-bracket; this is due to the vanishing of the non-geometric $Q$- and $R$-fluxes because $\beta$ has been set to zero here. We have thus shown how geometric fluxes appear in the context of para-Hermitian geometry, and in particular as certain deformations of para-Hermitian structures; this provides an explicit example of the general global formulation of fluxes given in \cite{svoboda}.

The lack of integrability in this case also means that $b=\eta_0\,B$ does not satisfy the Maurer-Cartan equation \eqn{maucar}: The obstruction to integrability in \eqn{dsec} is given by $$\Big(\de_+b+ \big(\mbox{$\bigwedge^3$} \eta\big)[b,b]^{\tt S}_-\Big)(D_i,D_j,Z)= \de\omega_B^{(+3,-0)_B}(D_i,D_j,Z) \ ,$$ where $[b,b]^{\tt S}_-$ is the Schouten-Nijenhuis bracket of $b$ regarded as a bivector field and the bigrading on the right-hand side is with respect to the twisted para-complex structure $K_B.$ The relation \eqref{eq:DbracketB} between D-brackets is also easily verified in the present case, since the D-bracket associated to the canonical para-K\"ahler structure is the bracket associated to the Levi-Civita connection of the flat metric $\eta_0$. 

\paragraph{Non-Geometric Fluxes from $\boldsymbol\beta$-Transformations.}
The $B$-transformations considered above preserve the natural splitting induced by the projection map, i.e. they only twist the horizontal distribution while preserving the vertical distribution, since the cotangent projection does not uniquely define a horizontal sub-bundle. On the other hand, a skew transformation $\beta,$ i.e. $\eta_0\,\beta$ is a 2-form, which does not preserve the natural splitting arising from the projection $\pi: T^*\cQ \rightarrow \cQ,$ leads to the emergence of locally non-geometric $R$-flux, as we saw in the discussion of the $R$-flux model in Section~\ref{sec:R-flux}. For this, we explain the notion of $\beta$-transformation on the cotangent bundle, in order to provide a more general interpretation of the $R$-flux model.

Let us consider the canonical para-K\"ahler structure $(K_0, \eta_0)$ on the cotangent bundle $T^*\cQ$ for which $T(T^*\cQ)=L^0_{\tt v}(T^*\cQ) \oplus L_{\tt h}^0(T^*\cQ)$, and define a $\beta$-transformation on it by a map $\beta: \Gamma(L^0_{\tt v}(T^*\cQ)) \rightarrow \Gamma(L_{\tt h}^0(T^*\cQ)),$ so that in local coordinates $\beta= \beta^{ij}\, P_j \otimes \de p_i.$ 
The splitting is twisted to $T(T^*\cQ)=L_{\tt v}^{\beta}(T^*\cQ)\oplus L_{\tt h}^0(T^*\cQ)$ such that $\Gamma(L_{\tt h}^0(T^*\cQ))= {\rm Span}_{C^\infty(T^*\cQ)}\{ P_i \}$ and $\Gamma(L_{\tt v}^{\beta}(T^*\cQ))={\rm Span}_{C^\infty(T^*\cQ)}\{\tilde D^i= Q^i + \beta^{ij}\, P_j \}$, i.e. the almost para-complex structure $$K_{\beta}=  \tilde D^i \otimes \de p_i - P_i \otimes \theta^i$$ is defined on $T^*\cQ,$ where $\de p_i$ and 
$ \theta^i= \de q^i + \beta^{ij}\,\de p_j$ are the dual 1-forms of the vectors $\tilde D^i$ and $P_i$ respectively. Then $K_{\beta}$ is obtained as a \emph{$\beta$-transformation} of $K_0$: $$K_0 \longmapsto K_{\beta}= e^{-\beta}\,K_0\, e^{\beta} \ ,$$ where
\be
e^{\beta}=
\bigg( \begin{matrix}
\mathds{1} & \beta \\
0 & \mathds{1}
\end{matrix} \bigg) \ \in \ {\sf O}(d,d)(T^*\cQ)
 \ , \nonumber
\ee
with respect to the canonical splitting $T(T^*\cQ) = L_{\tt v}^0(T^*\cQ)\oplus L_{\tt h}^0(T^*\cQ)$. 
In this case a $\beta$-transformation does not preserve verticality, hence it ruins the natural construction of a para-Hermitian structure on any bundle since the vertical sub-bundle is intrinsically defined by the underlying structures characterizing a bundle.  This is analogous to what happens in generalized geometry, where there is a distinctive difference between $B$-transformations and $\beta$-transformations, since only $B$-transformations (with a closed 2-form) preserve the Courant bracket~\cite{gualtieri:tesi}.
 
Since $e^\beta\in{\sf O}(d,d)(T^*\cQ)$, any $\beta$-transformation is an isometry of the flat metric $\eta_0$, hence $L_{\tt v}^{\beta}(T^*\cQ)$ and $L_{\tt h}^0(T^*\cQ)$ are still maximally isotropic distributions. This shows that the metric $\eta_0$ can be written in local coordinates as $$\eta_0= \theta^i \otimes \de p_i + \de p_i \otimes \theta^i \ .$$ Thus $\eta_0$ and $K_{\beta}$ are compatible and together define an almost para-Hermitian structure.
Then the fundamental 2-form reads $$\omega_{\beta}= \eta_0\, K_{\beta}=\omega_0 + 2\,\eta_0\, \beta= \theta^i \wedge \de p_i \ ,$$ where $2\,\eta_0\, \beta$ is a $(+2,-0)$-form with respect to the canonical splitting, and its exterior derivative is
\be
\de \omega_{\beta}= \partial_k \beta^{ij}\, \theta^k \wedge \de p_j \wedge \de p_i + \big(\tilde{\partial}^k \beta^{ij}+ \beta^{km}\, \partial_m\, \beta^{ij}\big) \, \de p_k \wedge \de p_j \wedge \de p_i \ . \label{dombeta}
\ee
From \eqn{dombeta} we see that $\omega_{\beta}$ fails to be closed and
the obstruction to closure is given by $\partial_k \beta^{ij}$, which represents a globally non-geometric $Q$-flux when the $B$-field is turned off, and by $\tilde{\partial}^k \beta^{ij}+ \beta^{km}\,\partial_m \beta^{ij}$, which is a locally non-geometric $R$-flux when $B=0.$

We can now compute the D-bracket associated to $K_{\beta}$ between two basis vectors $\tilde D^i,$ since $\llbracket P_i, P_j \rrbracket^{\tt D}_{\beta}=0,$ i.e. $L_{\tt h}^0(T^*\cQ)$ is a Frobenius integrable and weakly integrable eigenbundle. This computation relates the failure of closure of $\omega_{\beta}$ to the relative concept of weak integrability of the twisted vertical distribution $L_{\tt v}^{\beta}(T^*\cQ)$ with respect to the canonical para-K\"ahler structure on $T^*\cQ.$ We first need the bracket associated to the (flat) Levi-Civita connection compatible with $\eta_0$ which is given by 
\be
\eta_0\big(\llbracket \tilde D^i, \tilde D^j \rrbracket^{{\nabla}^{\tt LC}}, Z\big)= Z^k\, \partial_k \beta^{ij} + 3\,\tilde{Z}_k\, \big(\tilde{\partial}^{[k}\beta^{ij]}+ \beta^{[k|m}\, \partial_m \beta^{|ij]}\big) \ . \label{dbrabeta}
\ee
Thus from the bracket \eqref{dbrabeta} and \eqn{dombeta}, we obtain $\llbracket \tilde D^i, \tilde D^j \rrbracket^{\tt D}_{\beta}= 0.$ 
In \eqn{dbrabeta} we clearly see that the $R$-flux obstructs weak integrability while the $Q$-flux arises as the involutive component of the ${\nabla}^{\tt LC}$-bracket, as expected on general grounds from the distinction between local versus global non-geometry.

\paragraph{Generalized Fluxes from $\boldsymbol{B{+}\beta}$-Transformations.}
Finally we describe a twist of the canonical para-K\"ahler structure on $T^*\cQ$ given by the composition of a $B$-transformation and a $\beta$-transformation, as discussed in \cite{membrane} in the context of generalized geometry. Following the discussion above, the new almost para-Hermitian structure is given by  
$$K= e^{-\beta}\,e^{B}\,K_0\, e^{-B}\, e^{\beta} \ , \quad \eta=\eta_0 \qquad \mbox{and} \qquad \omega=\eta_0\, K \ .$$ In matrix notation this reads
\be
K=
\bigg(\begin{matrix}
\mathds{1}-2\,\beta\, B & 2\,(\beta- \beta\, B\, \beta) \\
2\,B & -\mathds{1}+ 2\,B\, \beta 
\end{matrix}\bigg)
 \ , \quad
\eta_0=
\bigg(\begin{matrix}
0 & \mathds{1} \\
\mathds{1} & 0
\end{matrix}\bigg)
\qquad \mbox{and} \qquad
\omega=
\bigg(\begin{matrix}
2\,B & -\mathds{1}+ 2\,B\, \beta \\
\mathds{1}- 2\, \beta\, B & 2\,(\beta- \beta\, B\, \beta) 
\end{matrix}\bigg)
\nonumber
\ee
where the fundamental 2-form can be written as $$\omega= \omega_0  + 2 \,\eta_0\, (B\, \beta -\beta\, B)+ 2\,\eta_0\, B + 2\,\eta_0\, (\beta- \beta\, B\, \beta) \ . $$ The 2-form $\omega_0 +2\, \eta_0\, (B\,\beta - \beta\, B)$ is its $(+1,-1)$-component, $2\,\eta_0\, B$ is its $(+0,-2)$-component, and $ 2\,\eta_0\, (\beta- \beta\, B\, \beta)$ is its $(+2,-0)$-component with respect to the canonical splitting $T(T^*\cQ)=L^0_{\tt v}(T^*\cQ)\oplus L_{\tt h}^0(T^*\cQ).$

The splitting $T(T^*\cQ)= L_{\tt v}^{\beta}(T^*\cQ)\oplus L_{\tt h}^B(T^*\cQ) $ is given in a local patch by
\begin{align*}
\Gamma\big(L_{\tt v}^{\beta}(T^*\cQ)\big)&= {\rm
                                           Span}_{C^\infty(T^*\cQ)}\big\{
                                           \tilde D^i=Q^i+\beta^{ij}\,
                                           D_j\big\} \ , \\[4pt]
  \Gamma\big(L_{\tt h}^B(T^*\cQ)\big)&= {\rm
                                       Span}_{C^\infty(T^*\cQ)}\big\{
                                       D_i=P_i + B_{ij}\,Q^j\big\} \ ,
\end{align*}
with $$\eta_0(D_i, D_j)= \eta_0 (\tilde D^i, \tilde D^j)=0 \qquad \mbox{and} \qquad \eta_0(D_i, \tilde D^j)=\eta_0 (\tilde D^j, D_i)= \delta_i^j \ .$$
The Lie algebra defined by these distributions is
\begin{align}
[D_i, D_j]&= J_{ijk} \, \tilde D^k -J_{ijm} \, \beta^{mk} \, D_k \ ,
            \nonumber \\[4pt]
  [D_i,\tilde D^j]&= C_i{}^{jk}\, D_k  - W^j{}_{ik}\, \tilde D^k \
                    , \label{eq:BbetaLie} \\[4pt] [\tilde D^i, \tilde D^j]&={P^{ij}}_k \, \tilde D^k + A^{ijk} \, D_k \ , \nonumber
\end{align}
where 
\begin{align}
J_{ijk}&=2\,\big(\partial_{[i} B_{j]k} +B_{[i|l}\,\tilde{\partial}^l
         B_{|j]k}\big) \ , 
\nonumber \\[4pt]
C_i{}^{jk}&= \partial_i\beta^{jk}+B_{il}\,\tilde\partial^l\beta^{jk} + \beta^{jm}\,\beta^{lk}\,\big(\partial_iB_{ml}+\partial_mB_{il}+B_{mp}\,\tilde\partial^pB_{il}-B_{ip}\,\tilde\partial^pB_{ml}\big) \ , \nonumber \\[4pt]
W^j{}_{ik}&= \tilde\partial^jB_{ik}+B_{il}\, B_{mk}\,
            \tilde\partial^l\beta^{jm}+\beta^{jl}\,\big(\partial_iB_{lk}
            + \partial_lB_{ik}+B_{lm}\,
            \tilde\partial^mB_{ik}-B_{im}\,\tilde\partial^mB_{lk}\big)
            \ , \label{integra} \\[4pt]
{P^{ij}}_k&=2\,\big(\beta^{p[i}\,\tilde{\partial}^{j]}\,B_{pk} -
            \beta^{l[i}\,\beta^{j]m}\,\partial_l B_{mk} -
            \beta^{l[i}\,\beta^{j]p}\,B_{lm}\,\tilde{\partial}^m
            B_{pk}\big) \ , 
\nonumber \\[4pt]
A^{ijk}&= 2\,\big(\tilde{\partial}^{[i}\beta^{j]k}+
         \beta^{[i|l}\, \partial_l\, \beta^{|j]k}
         +\beta^{[i|l}\,B_{lm}\, \tilde{\partial}^m \beta^{|j]k}\big) -
         {W^{ij}}_l\, \beta^{lk} \ . \nonumber
\end{align}
Thus neither the vertical nor the horizontal distribution is integrable in this polarization. The dual 1-forms to $\tilde D^i$ and $D_i$ respectively are
$$\tilde\theta_i=\de p_i + B_{ij}\,\de q^j \qquad \mbox{and} \qquad \theta^i=\de q^i +\beta^{ij}\,\tilde\theta_j \ .$$ 
Then we can write the local expressions of $K$ and $\eta_0$ as 
$$K= \tilde D^i \otimes \tilde\theta_i - D_i \otimes \theta^i \qquad \mbox{and} \qquad \eta_0= \tilde\theta_i \otimes \theta^i + \theta^i \otimes \tilde\theta_i \ ,$$ 
so that the fundamental 2-form reads $\omega= \eta_0\, K= \theta^i \wedge \tilde\theta_i.$

Generalized fluxes emerge from the bracket associated to the original canonical para-K\"ahler structure on $T^*\cQ.$ Hence we compute the ${\nabla}^{\tt LC}$-bracket, which is the D-bracket of this structure, on the basis vectors of the twisted vertical and horizontal distributions:
\be\nonumber
\llbracket D_i, D_j \rrbracket^{{\nabla}^{\tt LC}}=\mathscr{H}_{ijk}\,
\tilde D^k + {\mathscr{F}_{ij}}^k\, D_k \qquad \mbox{and} \qquad \llbracket
\tilde D^i, \tilde D^j \rrbracket^{{\nabla}^{\tt LC}}= \mathscr{Q}^{ij}{}_k\, \tilde D^k + \mathscr{R}^{ijk}\, D_k \ ,
\ee
where 
\begin{align}
\mathscr{H}_{ijk}&=3\,\partial_{[i}B_{jk]}+3\,B_{[i|m}\,\tilde{\partial}^m B_{|jk]} \ , \nonumber \\[4pt]
{\mathscr{F}_{ij}}^k&=\tilde{\partial}^k B_{ij}+\beta^{km}\, \mathscr{H}_{mij} \ , \nonumber \\[4pt]
\mathscr{Q}^{ij}{}_k &=\partial_k \beta^{ij}+ \beta^{im}\,\beta^{jl}\,\mathscr{H}_{mlk}+ B_{km}\,\tilde{\partial}^m \beta^{ij}+ 2\,\beta^{p[i}\,\tilde{\partial}^{j]}B_{pk} \ , \nonumber \\[4pt]
\mathscr{R}^{ijk}&=3\, \tilde{\partial}^{[i}\beta^{jk]}+ 3\,\beta^{[i|l}\,\beta^{|j|m}\,\tilde{\partial}^{|k]}B_{ml}+ 3\,\beta^{[i|m}\,\partial_m \beta^{|jk]}+ 3\, B_{lm}\,\beta^{[i|l}\,\tilde{\partial}^m \beta^{|jk]}+ \beta^{il}\,\beta^{jm}\,\beta^{kn}\,\mathscr{H}_{lmn} \ . \nonumber 
\end{align}
These are precisely the generalized fluxes\footnote{The Bianchi
  identities for the fluxes follow from the Jacobi identity for the
  Lie brackets \eqref{eq:BbetaLie}.} of double field theory in a
holonomic frame obtained from the standard Courant algebroid
description~\cite{Geissbuhler2013,membrane} or from the Roytenberg
bracket \cite{Halmagyi}; their counterparts in an arbitrary
non-holonomic frame can be obtained by further applying an ${\sf
  O}(d,d)(T^*\cQ)$-transformation of the form \eqref{eq:nonholonomic} to write
the change of basis $\tilde E^a=(A^{-1})^a{}_i\,Q^i$ on
$L^0_{\tt v}$ and $E_a=A_a{}^i\,P_i$ on $L^0_{\tt h}$, where
$A\in{\rm End}(T\cQ)$ is a local ${\sf GL}(d,\IR)$-transformation inducing
geometric $f$-flux through the non-vanishing Lie brackets
$[E_i,E_j]=f_{ij}{}^k\,E_k$ with $f_{ij}{}^k = (A^{-1})^a{}_{[i}\,
(A^{-1})^b{}_{j]}\, \partial_aA_b{}^k$. In the present case these
fluxes arise as a measure of the relative weak integrability between
two (almost) para-Hermitian structures related via a composition of a
$B$-transformation and a $\beta$-transformation,\footnote{Note that an
  arbitrary change of polarization $\vartheta\in{\sf O}(d,d)(M)$ can
  be parameterized as $\vartheta=e^{-B}\, A\,e^\beta$, cf. \ Interlude~\ref{int:Odd}.} with the
$H$-flux obstructing the integrability of $L_{\tt h}^B$ and the
$R$-flux obstructing integrability of $L_{\tt v}^\beta$. This also justifies
once more the choice in \cite{membrane} of the cotangent bundle for
the doubled target space of the membrane sigma-model for double field
theory: It carries a natural para-K\"ahler structure obtained solely
from the properties of the bundle itself, while the generalized fluxes
can be encoded in its deformation and used to introduce a Wess-Zumino
type topological term in the membrane action. Such a construction can of course also be carried out on any flat para-K\"ahler manifold, with the same result; we have chosen the cotangent bundle because it naturally carries such a structure.

The D-bracket associated to the $B{+}\beta$-twisted para-Hermitian structure gives only the integrable part of the $\nabla^{\tt LC}$-bracket algebra defined by the two distributions:
$$\llbracket D_i, D_j \rrbracket^{\tt D} = J_{ijm}\,\beta^{mk}\, D_k \qquad \mbox{and} \qquad \llbracket \tilde D^i, \tilde D^j \rrbracket^{\tt D} = {P^{ij}}_k\, \tilde D^k \ .$$
These brackets also show that the fundamental 2-form $\omega$ is not
closed; the components of $\de\omega$ can be obtained directly from
$\omega=\theta^i\wedge\tilde\theta_i$, or equivalently as the difference of the D-bracket and the ${\nabla}^{\tt LC}$-bracket by using \eqref{dciv} to get
\begin{align*}\nonumber
\de\omega &= 2\,\big( \mathscr{H}_{ijk}\, \theta^i\wedge\theta^j\wedge\theta^k
- \mathscr{R}^{ijk}\,
            \tilde\theta_i\wedge\tilde\theta_j\wedge\tilde\theta_k
  \\[4pt] & \qquad +\, ({\mathscr{F}_{ij}}^k-J_{ijm}\,\beta^{mk}) \,
  \theta^i\wedge\theta^j\wedge\tilde\theta_k - (\mathscr{Q}^{ij}{}_k -
  P^{ij}{}_k) \, \tilde\theta_i\wedge\tilde\theta_j\wedge\theta^k\big) \ .
\end{align*}
The generalized fluxes then appear amongst the Jacobiators of the
corresponding twisted Poisson brackets
\be\nonumber
\{q^i,q^j\}_{B,\beta} = 2\,\big(\beta^{ij}-\beta^{ik}\,B_{kl}\,\beta^{lj}\big) \ , \quad \{q^i,p_j\}_{B,\beta} = \delta^i{}_j-2\,\beta^{ik}\,B_{kj} \qquad \mbox{and} \qquad \{p_i,p_j\}_{B,\beta} = 2\,B_{ij} \ .
\ee
Finally, we can start from a reference generalized metric \eqref{eq:tildeHbeta} and apply the change of polarization
\begin{align*}
\widetilde{\cH}_{B,\beta} &= \big(e^{-\beta}\big)^{\rm t} \, \big(e^{-B}\big)^{\rm t} \, \widetilde{\cH}_0 \, e^{-B} \, e^{-\beta} \\[4pt]
&= \bigg( \begin{matrix}
g^{-1}-\beta\,g\,\beta+(g^{-1}\, B\, \beta)_{\rm ss} + \beta\,B\,g^{-1}\,G\,\beta &g^{-1}\,B-\beta\,g+\beta\,B\,g^{-1}\,B \\ -B\,g^{-1}+g\, \beta-B\,g^{-1}\,B\,\beta & g-B\,g^{-1}\,B
\end{matrix} \bigg)
\ ,
\end{align*}
where the subscript ${}_{\rm ss}$ means the skew-symmetric part of a
$(0,2)$-tensor on $M=T^*\cQ$; this is indeed the correct form of the covariant generalized metric on $T^*\cQ$~\cite{membrane}.

\section{Para-Hermitian Geometry of Drinfel'd Doubles} \label{paradrinfeld}

We shall now move on to study other related examples of how fluxes arise in
para-Hermitian geometry, which extend our previous considerations
globally to certain classes of parallelizable manifolds. A broad class of examples of
para-Hermitian manifolds naturally arises in the form of Lie groups
that are Drinfel'd
doubles, which provide a global extension of the local geometry of the cotangent bundles
we considered previously to curved backgrounds; from a dynamical perspective, they describe duality transformations between particular field theories valued in Lie groups which are given by principal chiral models. At the same time, they automatically capture the relation with generalized geometry and provide a natural notion of non-abelian T-duality for Lie groups. Double field theory and in particular Poisson-Lie T-duality on Drinfel'd doubles has also been considered by~\cite{hulled,edwards:nonpub,Hassler2017,osten}. In this section we will
adapt the description presented in~\cite{hulled} to the formalism of
Sections~\ref{sec:ParaHermitianGeometry}
and~\ref{sec:nonassociativity} in this setting. 

\subsection{The Left-Invariant Para-Hermitian Structure} \label{paraherdri}

A \emph{(classical) Drinfel'd double} $\sfD$ is a $2d$-dimensional Lie
group whose Lie algebra $\frd$ can be given in the split form $\mathfrak{d}=\mathfrak{g} \Join \tilde{\mathfrak{g}},$ where $\mathfrak{g}$ and $\tilde{\mathfrak{g}}$ are two dual Lie subalgebras of $\mathfrak{d}$ generated, respectively, by elements $T_i$ and $\tilde{T}^i$ with $i=1, \dots,d$ satisfying 
\be
[T_i,T_j]= f_{ij}^{\ \ k}\,T_k \ , \quad [T_i, \tilde{T}^j]= f_{ki}^{\ \ j}\,\tilde{T}^k- Q^{kj}{}_i\, T_k \qquad \mbox{and} \qquad [\tilde{T}^i, \tilde{T}^j]=Q^{ij}{}_k\,\tilde{T}^k\ . \label{drinflie}
\ee
The Jacobi identity for the Lie bracket implies the algebraic Bianchi
identities
\be \nonumber
f_{[ij}{}^m\, f_{k]m}{}^n=0 \ , \quad f_{ij}{}^m\, Q^{kl}{}_m= Q^{m[k}{}_{[i}\, f_{j]m}{}^{l]} 
\qquad \mbox{and} \qquad Q^{[ij}{}_m\, Q^{k]m}{}_n=0 \ .
\ee
The Lie subalgebras $\mathfrak{g}$ and $\tilde{\mathfrak{g}}$ together define a Lie bialgebra $(\frg,\tilde\frg)$, and they
respectively generate two Lie subgroups $\sfG$ and $\tilde\sfG$ of
$\sfD$ such that $\sfD=\sfG\Join \tilde\sfG$ which are dual in
the sense that their Lie algebras are dual. This duality induces a
natural ${\sf Ad}(\sfD)$-invariant inner product on the Lie algebra of
the Drinfel'd double. The Lie brackets \eqref{drinflie} describe the
gauge algebra of a string compactification on a Poisson-Lie background
(see e.g.~\cite{edwards:nonpub}).

This is not the only possible splitting of the Lie algebra $\mathfrak{d}.$ Generally, a splitting of $\mathfrak{d}$ given in terms of two maximally isotropic subspaces is called a \emph{polarization}.\footnote{See~\cite{alekseev:articolo,kossman:articolo} for a more precise definition and comprehensive treatment of the geometry of Drinfel'd doubles.} When the polarization is given in terms of two maximally isotropic Lie subalgebras $\mathfrak{g}$ and $\tilde{\mathfrak{g}}$ as above, the triple $(\mathfrak{d},\mathfrak{g}, \tilde{\mathfrak{g}})$ is called a \emph{Manin triple}. 
In this case the duality pairing between $\mathfrak{g}$ and $\tilde{\mathfrak{g}}$ can be regarded as an invariant ${\sf O}(d,d)$-metric $\eta$ on the Lie algebra $\mathfrak{d}$ such that $$\eta(T_i,T_j)=\eta(\tilde{T}^i, \tilde{T}^j)=0 \qquad \mbox{and} \qquad \eta(T_i,\tilde{T}^j)=\eta(\tilde{T}^j, T_i)=\delta_i^j \ .$$
The splitting of $\mathfrak{d}$ can be equivalently regarded as an
invariant para-complex structure $K$ on $\frd$ such that
$\mathfrak{g}$ is its $+1$-eigenspace and $\tilde{\mathfrak{g}}$ is
its $-1$-eigenspace, so that $$K= T_i \otimes \tilde{T}^i -
\tilde{T}^i \otimes T_i \ .$$ 
This is the para-K\"ahler structure associated to the Manin triple
polarization $\mathfrak{d}=\frg\Join\tilde\frg$; alternatively, the
para-K\"ahler structure $-K$ is associated to the polarization
$\frd=\tilde\frg\Join\frg$. In the spirit of Section~\ref{sec:R-flux},
the change of polarization $K\mapsto-K$ is a type of Born reciprocity
transformation.

The Lie algebra $\mathfrak{d}$ is isomorphic to the Lie algebra of left-invariant vector fields on the Drinfel'd double $\sfD,$ which are globally defined because
the tangent and cotangent bundles of a Lie group are trivial vector
bundles. Hence we may translate this construction to the group manifold $\sfD,$ where coordinates are given by $x^I= (x^i, \tilde{x}_i)$ in a local chart. For this, we need to construct left-invariant 1-forms and vector fields. We first describe this construction in a Manin triple polarization, and then later on give the general form in an arbitrary polarization.

In order to obtain the local expression of the left-invariant 1-forms,
we fix the Iwasawa decomposition of $\sfD$ to be $\gamma=g \,
\tilde{g}$, for any element $\gamma\in\sfD$ in terms of elements
$g=\exp(x^i\,T_i)\in\sfG$ and $\tilde g=\exp(\tilde x_i\,\tilde
T^i)\in\tilde\sfG$; an equivalent discussion is possible with the dual
Iwasawa decomposition $\gamma=\tilde{g}\,g$. Then any left-invariant
1-form on $\sfD$ is valued in $\frd\otimes T^*\sfD$ and has the expression
$$\Theta= \gamma^{-1}\,\de \gamma = \tilde{g}^{-1}\, g^{-1} \, \de (g \, \tilde{g})= \tilde{g}^{-1}\,\big(g^{-1}\,\de g\big)\,\tilde{g}+ \tilde{g}^{-1}\,\big(\de \tilde{g} \, \tilde{g}^{-1}\big) \, \tilde{g} \ ,$$  
where $\lambda= g^{-1}\,\de g= \lambda^m \, T_m$ (depending only on the
coordinates $x^i$) and $\tilde{\rho}=\de\tilde{g} \, \tilde{g}^{-1}=\tilde{\rho}_m \,\tilde{T}^m$ (depending only on the coordinates $\tilde{x}_i$) are Lie algebra-valued left- and right-invariant 1-forms on $\sfG$ and $\tilde{\sfG}$ respectively. Then the left-invariant 1-form can be written as 
$$\Theta=\lambda^m \, \big(\tilde{g}^{-1} \, T_m \, \tilde{g}\big) + \tilde{\rho}_m
\, \big(\tilde{g}^{-1} \, \tilde{T}^m \, \tilde{g} \big) \ .$$ This
shows that we need the adjoint action of $\tilde{\sfG}$ on the
generators $T_M={{T_m}\choose{\tilde T{}^m}}$ of $\frd$, with the Lie
brackets $[T_M,T_N]=t_{MN}{}^P\, T_P$, which has the form~\cite{hulled}
\be\nonumber
\tilde{g}^{-1} \,
\bigg(\begin{matrix}
T_m \\
\tilde{T}^m
\end{matrix} \bigg) \,
\tilde{g}=
\bigg(\begin{matrix}
\big(\tilde{A}^{-1}\big)_m{}^n & \tilde b_{mn} \\
0 & \tilde{A}^m{}_n
\end{matrix}\bigg) 
\bigg(\begin{matrix}
T_n \\
\tilde{T}^n
\end{matrix}\bigg) \ .
\ee
Here the block matrices are defined by the adjoint action and they all
depend only on the coordinates $\tilde{x}_i$, where $\tilde b_{mn}$ is skew-symmetric because we have chosen a Manin triple polarization. Therefore the left-invariant 1-form is $$\Theta=\lambda^m \, \big(\tilde{A}^{-1}\big)_m{}^n \, T_n + \big(\lambda^m \, \tilde b_{mn}+ \tilde{\lambda}_n\big) \, \tilde{T}^n \ ,$$
where in the second term on the right-hand side we used $\tilde{g}^{-1}\, \tilde{\rho}\,\tilde{g}=\tilde{\lambda},$ where $\tilde{\lambda}=\tilde{g}^{-1}\,\de \tilde{g}=\tilde{\lambda}_m \, \tilde{T}^m,$ hence $ \tilde{\rho}_m\,\tilde{A}^m{}_n=\tilde{\lambda}_n.$ 
The Lie algebra components of $\Theta=\Theta^M\,T_M$ are given by the 1-forms
\be \label{leinfo}
\Theta^M=
\bigg(\begin{matrix}
\Theta^m \\
\tilde{\Theta}_m
\end{matrix} \bigg)
=
\bigg(\begin{matrix}
\lambda^n_i \, \big(\tilde{A}^{-1}\big)_n{}^m \, \de x^i \\
\lambda^n_i\, \tilde b_{nm} \, \de x^i + \tilde{\lambda}^i_m \, \de \tilde{x}_i 
\end{matrix} \bigg)
\ee
from which we obtain the dual left-invariant vector fields
\be
Z_M=
\bigg(\begin{matrix}
Z_m \\
\tilde{Z}^m
\end{matrix}\bigg)
=
\bigg(\begin{matrix}
\tilde{A}^n{}_m \,\big( (\lambda^{-1})^i_n\, \frac\partial{\partial x^i} - (\tilde{\lambda}^{-1})^p_i\, \tilde b_{np}\, \frac{{\partial}}{\partial\tilde x_i} \big)\\
\big(\tilde{\lambda}^{-1}\big)^m_i \, \frac\partial{\partial\tilde x_i}
\end{matrix}\bigg)
\ . \label{leid}
\ee
The Lie brackets $[Z_M, Z_N]=t_{MN}{}^{P}\, Z_P$ close to the Lie algebra $\mathfrak{d}$ from \eqn{drinflie}. 

We can now construct the left-invariant para-Hermitian structure induced by the para-Hermitian structure defined on the Lie algebra $\mathfrak{d}$ in a Manin triple polarization; see also~\cite{osten} for a similar discussion of this construction.
The (globally defined) left-invariant para-complex structure is given by
\be\nonumber
K= Z_m \otimes \Theta^m - \tilde{Z}^m \otimes \tilde{\Theta}_m = \frac\partial{\partial x^i} \otimes \de x^i - 2 \,\tilde{b}_{pn} \, \big(\tilde{\lambda}^{-1}\big)^n_i \, \lambda^p_j \, \frac\partial{\partial\tilde x_i} \otimes \de x^j - \frac\partial{\partial\tilde x_i} \otimes \de \tilde{x}_i \ , 
\ee
where the eigenbundles $L_+$ and $L_-$ of $K$ are both integrable
since they are generated, respectively, by the vector fields $Z_m
$ and $\tilde{Z}^m$ which close to the Lie subalgebras
$\mathfrak{g}$ and $\tilde{\mathfrak{g}}$ from \eqn{drinflie}. 

The left-invariant metric with Lorentzian signature induced by the duality pairing is $$\eta= \Theta^m \otimes \tilde{\Theta}_m + \tilde{\Theta}_m \otimes \Theta^m \ . $$ In this splitting (polarization), it can be regarded as the ${\sf O}(d,d)$-invariant constant metric 
\be 
\eta=
\bigg(\begin{matrix}
0 & \mathds{1} \\
\mathds{1} & 0
\end{matrix}\bigg)
\ . \nonumber
\ee
In the local coordinates $x^I=(x^i, \tilde{x}_i)$ on the group
manifold $\sfD$ it is given by 
\be \label{metrloca}
\eta=\lambda^m_i \, \big(\tilde{A}^{-1}\big)_m{}^n \, \tilde{\lambda}_n^j \, \big(\de x^i \otimes \de \tilde{x}_j + \de \tilde{x}_j \otimes \de x^i\big) \ ,
\ee
since $\tilde b_{mn}$ is skew-symmetric.

Finally, the left-invariant fundamental 2-form is thus given by $$\omega= \eta\, K= \tilde{\Theta}_m \wedge \Theta^m \ .$$ It can be shown that $\omega$ is not closed:
$$\de \omega= \de \tilde{\Theta}_m \wedge \Theta^m - \tilde{\Theta}_m \wedge \de \Theta^m \ . $$ Using the Maurer-Cartan structure equations 
$$\de \Theta^P = -\frac{1}{2} \, {t_{MN}}^P \, \Theta^M \wedge \Theta^N \ ,$$ which in polarized components read
\begin{align*}
\de \Theta^p &= -\frac{1}{2}\,{f_{mn}}^p \, \Theta^m \wedge \Theta^n - Q^{mp}{}_n \, \Theta^n \wedge \tilde{\Theta}_m \ , \\[4pt]
\de \tilde{\Theta}_p &= -{f_{mp}}^n \, \tilde{\Theta}_n \wedge \Theta^m-\frac{1}{2} \, Q^{mn}{}_p \, \tilde{\Theta}_m \wedge \tilde{\Theta}_n \ ,
\end{align*}
we obtain $$\de \omega=-\frac{1}{2}\, \big({f_{mn}}^p \, \tilde{\Theta}_p
\wedge \Theta^m \wedge \Theta^n -{Q^{mn}}_p \, \Theta^p \wedge \tilde{\Theta}_m \wedge \tilde{\Theta}_n
\big) \ .$$ 

We have
therefore shown that a Drinfel'd double $\sfD$ is endowed with a
natural para-Hermitian structure, having two Lagrangian foliations
with leaves given by $\sfG$ and $\tilde\sfG$; this result does not depend on the
choice of Iwasawa decomposition of~$\sfD.$ It is para-K\"ahler if and only if both both groups $\sfG$ and $\tilde\sfG$ (and hence $\sfD$) are abelian. Since $\tilde{\frg}$ is
dual to $\frg$, there is a bundle isomorphism $T^*\sfG\simeq
T\tilde{\sfG}$, and the splitting $T\sfD\simeq
T\sfG\oplus T\tilde{\sfG}$ appropriate to double field theory is
naturally identified with the splitting $T\sfD\simeq \mathbb{T}\sfG$
appropriate to generalized geometry on the Lie group $\sfG$. If $\sfG$ is connected and simply connected, then the Lie bialgebra $(\frg,\tilde\frg)$ makes $\sfG$ into a Poisson-Lie group~\cite{Drinfeld1983} and
endows the pair $(T\sfG,T^* \sfG)$ with the structure of a Lie
bialgebroid, while the Drinfel'd double structure $\frd=\frg \Join
\tilde\frg$ makes the generalized tangent bundle
$\mathbb{T}\sfG=T\sfG\oplus T^*\sfG$ into a Courant algebroid. Generally, a Manin triple polarization of $\frd$ gives a decomposition of $T\sfD$ into Dirac structures, i.e. integrable maximally isotropic
sub-bundles of $T\sfD$. 

Thus for a Drinfel'd double, doubled geometry naturally coincides with
generalized geometry. In this
case we recover $\cF_+=\sfG$ from the Manin triple
$(\frd,\frg,\tilde\frg)$ as the physical spacetime, identified as the
coset $\sfG=\sfD/\tilde\sfG$ by the left action of the subgroup
$\tilde\sfG$ whose isotropy group $\tilde\sfG$ is generated by the
right action. Alternatively, using instead the Manin triple
$(\frd,\tilde\frg,\frg)$ recovers $\cF_+=\tilde\sfG=\sfD/\sfG$ as the
physical spacetime, and the process of exchanging the Manin triple
polarizations $(\frd,\frg,\tilde\frg)$ and $(\frd,\tilde\frg,\frg)$ is
often called \emph{Poisson-Lie T-duality}~\cite{Klimcik1995,edwards:nonpub}. A change in a Manin
triple polarization of $\frd$ is also called a \emph{non-abelian T-duality}~\cite{osten}.

\subsection{Manin Triples as Flux Deformations of Para-K\"ahler Structures} \label{manindefor}

We shall now give a different geometric interpretation of this construction in terms of deformations of a fixed reference para-K\"ahler structure on the group manifold $\sfD$, similarly to what we did in Section~\ref{sec:nonassociativity}.
From the local coordinate expression \eqn{metrloca} of the
left-invariant metric $\eta$, we see that it can be obtained as a
transformation of the ${\sf O}(d,d)(\sfD)$-metric $\eta_0 =\de x^i \otimes
\de \tilde{x}_i + \de \tilde{x}_i \otimes \de x^i$ on the group
manifold via a sequence of transformations; however, in contrast to
the situation of Section~\ref{sec:nonassociativity}, these
transformations are generally \emph{not} valued in the T-duality group
${\sf O}(d,d)(\sfD).$  From this perspective we may infer that, since the group manifold $\sfD$ is $2d$-dimensional, a flat para-K\"ahler structure can always be defined by
$$K_0=\frac\partial{\partial x^i} \otimes \de x^i - \frac\partial{\partial\tilde x_i} \otimes \de \tilde{x}_i \qquad \mbox{and} \qquad \eta_0=\de x^i \otimes \de \tilde{x}_i + \de \tilde{x}_i \otimes \de x^i \ ,$$ 
with abelian eigenbundles of $K_0$ spanned by the vector fields $\frac\partial{\partial x^i}$ and $\frac\partial{\partial\tilde x_i}$; the canonical para-Hermitian connection of this structure is the Levi-Civita connection of $\eta_0$, and the fundamental 2-form is the closed form $\omega_0=\eta_0\,K_0=\de x^i\wedge\de\tilde x_i$. 

Then the para-Hermitian structure $(K,\eta)$ induced by a Manin triple polarization of the Lie algebra $\mathfrak{d}$ is a deformation of $(K_0,\eta_0)$ given by a composition of three types of transformations:

\noindent
$\bullet$ \ \underline{\sl ${\sf GL}(d,\mathbb{R})_+ {\times} {\sf GL}(d,\mathbb{R})_-$-transformations.} \  Such a transformation of the trivial para-K\"ahler structure $(K_0,\eta_0)$ is given by  
$$K'=K_0=\frac\partial{\partial x^i} \otimes \de x^i - \frac\partial{\partial\tilde x_i} \otimes \de \tilde{x}_i \qquad \mbox{and} \qquad \eta'=\lambda^m_i \, \tilde{\lambda}_m^j \, \big(\de x^i \otimes \de \tilde{x}_j + \de \tilde{x}_j \otimes \de x^i\big) \ .$$
The corresponding vector fields spanning the two maximally isotropic distributions are $$Z'_m =  \big(\lambda^{-1}(x)\big)^i_m \, \frac\partial{\partial x^i} \qquad \mbox{and} \qquad \tilde{Z}^{\prime\, m} = \big(\tilde{\lambda}^{-1}(\tilde{x})\big)^m_i \, \frac\partial{\partial\tilde x_i} \ ,$$ whose dual 1-forms are obtained, respectively, from
$$\lambda=g^{-1}\, \de g= \lambda^m_i \, \de x^i \, T_m=\Theta^{\prime\, m}\, T_m \qquad \mbox{and} \qquad \tilde{\lambda}=\tilde{g}^{-1}\, \de \tilde{g}= \tilde{\lambda}^i_m \, \de \tilde{x}_i\, \tilde{T}^m=\tilde{\Theta}'_m \, \tilde{T}^m \ .$$
It follows that the vectors $Z'_m$ and $\tilde{Z}^{\prime\, m}$ spanning
the two distributions close, respectively, to the Lie algebras
$\mathfrak{g}$ and $\tilde{\mathfrak{g}}$ such that
$[Z'_m,\tilde{Z}^{\prime\, n}]=0,$ which means that they separately
give two foliations of $\sfD$ with the subgroups $\sfG$ and $\tilde{\sfG}$ as
leaves. From this point of view, despite the fact that $Z'_m$ and
$\tilde{Z}^{\prime\, m}$ close to the Lie algebras defining the
Drinfel'd double, they are not obtained as global left-invariant
vector fields on $\sfD.$ The fundamental 2-form transforms to
$\omega'= \lambda^m \wedge \tilde{\lambda}_m, $ which is not closed: $$\de
\omega'=-\frac{1}{2}\,\big({f_{pq}}^m \,
\tilde{\lambda}_m \wedge \lambda^p\wedge \lambda^q - Q^{pq}{}_m \, \lambda^m \wedge \tilde{\lambda}_p \wedge
\tilde{\lambda}_q\big) \ ,$$ where we used the Maurer-Cartan structure
equations for the Lie groups $\sfG$ and $\tilde{\sfG}.$ Thus the
non-closure of the fundamental 2-form is, in this framework, related
to the non-abelian nature of the Lie algebras of the distributions. We
stress that this choice is fundamental, since it gives two canonical ``dual'' foliations, one of which can be interpreted as the physical spacetime submanifold. In this way the duality plays a fundamental role in linking the description of a physical spacetime to its dual.

\noindent
$\bullet$ \ \underline{\sl $B$-transformations.} \ Denoted here $\tilde b$, these transformations act to give
\be
\bigg(\begin{matrix}
Z''_m \\
\tilde{Z}^{\prime\prime\, m}
\end{matrix}\bigg)
=
\bigg(\begin{matrix}
\delta_m{}^n & \tilde b_{nm}\\
0 & \delta^m{}_n
\end{matrix}\bigg)
\bigg(\begin{matrix}
Z'_n \\
\tilde{Z}^{\prime\, n}
\end{matrix}\bigg)
 \ , \nonumber
\ee
which yields the globally defined left-invariant vector fields 
\be
Z_M=
\bigg(\begin{matrix}
Z_m \\
\tilde{Z}^m
\end{matrix}\bigg)
=
\bigg(\begin{matrix}
\big( (\lambda^{-1})^i_m \, \frac\partial{\partial x^i} - (\tilde{\lambda}^{-1})^p_i \, \tilde b_{mp} \, \frac\partial{\partial\tilde x_i} \big)\\
\big(\tilde{\lambda}^{-1}\big)^m_i \, \frac\partial{\partial\tilde x_i}
\end{matrix}\bigg)
 \ . \nonumber
\ee
Thus from this perspective a $B$-transformation acts by twisting the two distributions, with $K= e^{-B}\,K_0\,e^B$, and preserving the metric $\eta'$, i.e. it is an ${\sf O}(d,d)(\sfD)$-transformation. The fundamental 2-form $\omega''$ is not closed and $\de \omega''$ has the same coefficients as $\de \omega',$ written in the new dual basis. 

\noindent
$\bullet$ \ \underline{\sl ${\sf GL}(d, \mathbb{R})_+$-transformations.} \ Such a transformation rotates the distribution spanned by $Z''_m$ while preserving the other:
\be
\bigg(\begin{matrix}
Z_m \\
\tilde{Z}^m
\end{matrix}\bigg)
=
\bigg(\begin{matrix}
\tilde A_m{}^n &  0\\
0 & \delta^m{}_n
\end{matrix}\bigg)
\bigg(\begin{matrix}
Z''_n \\
\tilde{Z}^{\prime\prime\, n}
\end{matrix}\bigg)
\ , \nonumber
\ee
giving finally the basis \eqn{leid}. The combined action of the last
two transformations allows for non-vanishing brackets
$[Z'_m,\tilde{Z}^{\prime\, n}]$ and does not affect the Lie algebras
closed by the vector fields $Z'_m$ and $\tilde{Z}^{\prime\, m},$ thus
giving the desired geometric interpretation of how the Manin triple polarization is obtained in terms of deformations of para-Hermitian structures. This last transformation does not change the para-complex structure $K,$ which is affected only by a $B$-transformation, while the Lorentzian metric is finally transformed into the local expression \eqref{metrloca}. Again the fundamental 2-form $\omega$ is not closed, as shown previously, and $\de \omega$ still has the same components as $\de\omega'.$ Thus the last two transformations preserve the 3-form $\de \omega'$ in a Manin triple polarization.  

We have thus shown that a ${\sf GL}(d,\mathbb{R})_+ {\times} {\sf
  GL}(d,\mathbb{R})_-$-transformation fixes the Lie algebra closed by
the two distributions separately, thereby governing the closure of the
fundamental 2-form. A $B$-transformation (or $\beta$-transformation)
is needed to twist the two distributions together, as seen from its
action on the para-complex structure which deforms $K_0$ into $K.$
This means that any time there is a $B$- or $\beta$-transformation
involved in this deformation, we will find a non-trivial Kalb-Ramond
field on the physical spacetime submanifold; this is not the same as having non-vanishing generalized fluxes, since fluxes are governed by the non-integrability of the chosen polarization as we saw in Section~\ref{sec:nonassociativity}. Finally, a ${\sf GL}(d, \mathbb{R})_+$-transformation does not affect the para-complex structure $K$ but is needed to obtain the complete Lie algebra $\mathfrak{d}.$  
Note that, using the dual Iwasawa decomposition, the second
transformation becomes a $\beta$-transformation while the last
transformation becomes a ${\sf GL}(d,\IR)_-$-transformation rotating
only the distribution spanned by $\tilde{Z}^{\prime\prime\, m}$, and giving the same final result. Later on we will give an explicit example of this latter approach in the Drinfel'd double description of the cotangent bundle $T^*\sfG$ of a Lie group $\sfG$.

Let us now show how fluxes arise from the para-Hermitian geometry of Drinfel'd doubles.
In order to compute the D-bracket associated to the Lie algebra induced para-Hermitian structure $(K, \eta)$ and the ${\nabla}^{\tt LC}$-bracket compatible with the undeformed trivial para-K\"ahler structure $(K_0, \eta_0)$ on $\sfD$, we need the connection coefficients in the non-holonomic frame $Z_M$, which we compute to be
\begin{align}
{{\mit\Gamma}^i}_{jk} &=\frac{1}{2}\,{f_{kj}}^i \ , \quad {{\mit\Gamma}^{ij}}_k =\frac{1}{2}\, Q^{ji}{}_k \ , \quad {{\mit\Gamma}_{ij}}^k =\frac{1}{2}\, {f_{ij}}^k \qquad \mbox{and} \qquad {\mit\Gamma}^{ijk} =0 \ , \nonumber \\[4pt]
{{\mit\Gamma}_i}^j{}_k &=\frac{1}{2}\, {f_{ik}}^j \ , \quad {{\mit\Gamma}_i}^{jk} = \frac{1}{2}\, Q^{kj}{}_i \ , \quad {{\mit\Gamma}^i}_j{}^k =\frac{1}{2}\, Q^{ik}{}_j \qquad \mbox{and} \qquad {\mit\Gamma}_{ijk} =0 \ . \nonumber
\end{align}
Thus the ${\nabla}^{\tt LC}$-bracket is 
$$\llbracket Z_m, Z_n \rrbracket^{{\nabla}^{\tt LC}} =\frac{1}{2}\, {f_{mn}}^k \, Z_k \qquad \mbox{and} \qquad \llbracket \tilde{Z}^m, \tilde{Z}^n \rrbracket^{{\nabla}^{\tt LC}} = \frac{1}{2}\, Q^{mn}{}_k \, \tilde{Z}^k$$
while the D-bracket is
$$\llbracket Z_m, Z_n \rrbracket^{\tt D} ={f_{mn}}^k \, Z_k \qquad \mbox{and} \qquad \llbracket \tilde{Z}^m, \tilde{Z}^n \rrbracket^{\tt D} = Q^{mn}{}_k \, \tilde{Z}^k \ .$$
From these brackets it is clear that the left-invariant para-Hermitian
structure $(K, \eta)$ is compatible with the trivial para-K\"ahler
structure $(K_0, \eta_0)$ on $\sfD$, since $(K, \eta)$ is weakly
integrable with respect to the D-bracket of $(K_0, \eta_0).$ This
holds in any Manin triple polarization of $\sfD$.

The difference between these brackets is measured by $\de \omega.$ The lack of closure of the fundamental 2-form $\omega$ can be interpreted as a transformation from two abelian integrable distributions (the eigenbundles of the trivial para-K\"ahler structure on $\sfD$) to two non-abelian integrable distributions (the eigenbundles of the left-invariant para-Hermitian structure on $\sfD$). As $\de \omega$ is characterized by the structure constants of the Lie algebras closed by the two distributions, this impinges on the difference between the D-bracket and  the ${\nabla}^{\tt LC}$-bracket. This can be regarded as a situation in which only (globally geometric) $f$-flux and (locally geometric) $Q$-flux arise. The ${\nabla}^{\tt LC}$-bracket does not give any $H$-flux or $R$-flux in a Manin triple polarization, hence in order to induce such fluxes it is necessary to choose a polarization with non-integrable distributions, at least if the reference structure is the trivial para-K\"ahler structure; we consider this construction in detail below. Alternatively, one can accordingly change the reference para-Hermitian structure.

No matter what the polarization, in order to recover the physical background fields on the spacetime submanifold it is necessary to introduce a Born geometry $(K, \eta, \mathcal{H})$ that is left-invariant and globally defined on $\sfD.$
Following \cite{hulled}, this can be achieved by the introduction of
the globally defined (left-invariant and Riemannian) generalized
metric $$\mathcal{H}=\delta_{MN} \, \Theta^M\otimes\Theta^N = \delta_{mn}\, \Theta^m\otimes \Theta^n + \delta^{mn}\, \tilde{\Theta}_m \otimes \tilde{\Theta}_n$$ on the Drinfel'd double $\sfD$.

\begin{example}
A non-trivial example of a Born geometry for Drinfel'd doubles is
given by $\sfD={\sf SL}(2,\mathbb{C})$, regarded as a six-dimensional
real Lie group; this is studied in~\cite{mpv} in the context of the isotropic rigid rotator, whose
configuration space is the Lie group ${\sf SU}(2)$, as an alternative
carrier manifold (to the tangent bundle) for the lift of the dynamics. It has a Manin triple polarization ${\sf
  SL}(2,\IC)={\sf SU}(2)\Join{\sf SB}(2,\IC)$.\footnote{The Lie group
  ${\sf SB}(2,\IC)$ is the Borel subgroup of $2\times2$ upper
  triangular complex matrices with determinant equal to $1$.} In a suitable basis of $\mathfrak{sl}(2,\IC)$ the generators satisfy the commutation
relations
\be\nonumber
[T_i,T_j] = \frac12\,\varepsilon_{ij}{}^k\, T_k \ , \quad [ T_i,\tilde
T^j] = \frac12\,
\varepsilon_{ki}{}^j \,\tilde T^k-\frac12\,
\varepsilon^{kjl}\,\varepsilon_{l3i}\, T_k\qquad \mbox{and} \qquad [\tilde T^i,\tilde
T^j]=\frac12\, \varepsilon^{ijl}\, \varepsilon_{l3k}\,\tilde T^k
\ .
\ee
The ${\sf O}(d,d)$-invariant metric $\eta$ is obtained from the
Cartan-Killing form $\langle a,b\rangle =
2\,\mathrm{Im}\big(\mathrm{Tr}(a\,b)\big),$ for 
$a,b\in\mathfrak{sl}(2,\mathbb{C}),$ which gives the duality pairing
between the Lie subalgebras $\mathfrak{su}(2)$ and
$\mathfrak{sb}(2,\mathbb{C})$, and hence realises ${\sf SU}(2)$ and ${\sf SB}(2,\IC)$ as T-dual submanifolds of the Drinfel'd double $\sfD={\sf SL}(2,\mathbb{C})$. Writing
$F_i^\pm=\frac1{\sqrt2}\,\big(T_i\pm(\delta_{ij}\pm\varepsilon_{ij3}\,\tilde
T^j)\big)$, the isotropy conditions read as $\langle
F_i^+,F_j^+\rangle=\delta_{ij}=-\langle F_i^-,F_j^-\rangle$ and
$\langle F_i^+,F_j^-\rangle=0$. On the other hand, the generalized
metric $\mathcal{H}$ is obtained from the other natural inner product
$(a,b) = 2\,\mathrm{Re}\big(\mathrm{Tr}(a\,b)\big)$ (which does not
define a Manin triple polarization), for which one finds
\be\nonumber
\cH = \delta^{ij}\, \big(F_i^+\otimes F_j^++ F_i^-\otimes F_j^-\big)
\ .
\ee
Expanding this out with respect to the splitting $\mathfrak{sl}(2,\IC)=
\mathfrak{su}(2)\Join{\mathfrak{sb}}(2,\IC)$, and comparing with
\eqref{eq:geometricH}, then identifies the metric $g_{ij}=\delta_{ij}$
and 2-form $b_{ij}=\varepsilon_{ij3}$ on $\mathfrak{su}(2)$, which
lead to the standard round metric and Kalb-Ramond field (whose
$H$-flux is the volume form) on the 3-sphere ${\sf SU}(2)=S^3$.
See~\cite{mpv} for further details.
\end{example}

\subsection{Polarizations and Generalized Fluxes} \label{anotear}

Let us now describe arbitrary polarizations as alluded to above. This
is the situation where the doubled group $\sfD$ is now a
\emph{twisted} Drinfel’d
double. In this case, the description presented in \cite{hulled}
becomes much more complicated because of the more general form of the
adjoint action of $\tilde\sfG$, so describing the framework in which the final polarization is obtained via a chain of transformations of the trivial para-K\"ahler structure becomes highly non-trivial since it should mix rotations, $B$-transformations and $\beta$-transformations.  In the following we will describe how the choice of a polarization with non-involutive sub-bundles affects the D-brackets, showing that generalized fluxes associated with the structure constants of the Lie algebra of $\sfD$ emerge from such a choice.

The Lie algebra $\mathfrak{d}$ now splits into two maximally isotropic subspaces, i.e. we assume that the left-invariant vector fields on $\sfD$ close to the Lie algebra $[Z_M, Z_N]={t_{MN}}^P \, Z_P$ where
\begin{align*}
[Z_m, Z_n]&={f_{mn}}^k\, Z_k + H_{mnk}\, \tilde{Z}^k \ , \\[4pt] [ Z_m,\tilde{Z}^n]&= {f_{km}}^n \, \tilde{Z}^k -{Q^{kn}}_m\, Z_k \ , \\[4pt] [\tilde{Z}^m, \tilde{Z}^n]&=Q^{mn}{}_k\, \tilde{Z}^k + R^{mnk}\, Z_k \ .
\end{align*}
The Jacobi identity for this Lie bracket yields the algebraic Bianchi
identities
\begin{align*}
f_{[ij}{}^m\, f_{k]m}{}^n &= Q^{mn}{}_{[i}\, H_{jk]m} \ , \\[4pt]
f_{ij}{}^m\, Q^{kl}{}_m - Q^{m[k}{}_{[i}\, f_{j]m}{}^{l]} &= R^{mkl}\,
H_{mij} \ , \\[4pt] Q^{[ij}{}_m\, Q^{k]m}{}_n &=
f_{mn}{}^{[i}\, R^{jk]m} \ , \\[4pt]
f_{[ij}{}^m\,H_{kl]m} &= 0 = Q^{[ij}{}_m\, R^{kl]m} \ .
\end{align*}
The dual left-invariant 1-forms satisfy the Maurer-Cartan structure equations $$\de \Theta^P = -\frac{1}{2}\, {t_{MN}}^P\, \Theta^M \wedge \Theta^N \ ,$$ whose polarized components read
\begin{align*}
\de \Theta^p &= -\frac{1}{2}\,\big({f_{mn}}^p\, \Theta^m \wedge \Theta^n +R^{mnp}\, \tilde{\Theta}_m \wedge \tilde{\Theta}_n\big) - {Q^{mp}}_n\, \Theta^n \wedge \tilde{\Theta}_m \ , \\[4pt]
\de \tilde{\Theta}_p &= -{f_{mp}}^n\, \tilde{\Theta}_n \wedge \Theta^m-\frac{1}{2}\, \big({Q^{mn}}_p\, \tilde{\Theta}_m \wedge \tilde{\Theta}_n + H_{mnp}\, \Theta^m \wedge \Theta^n\big) \ .
\end{align*}
The left-invariant almost para-Hermitian structure is still well-defined and given by
\be
K= Z_m \otimes \Theta^m - \tilde{Z}^m \otimes \tilde{\Theta}_m \ , \quad \eta= \Theta^m \otimes \tilde{\Theta}_m + \tilde{\Theta}_m \otimes \Theta^m \qquad \mbox{and} \qquad \omega= \tilde{\Theta}_m \wedge \Theta^m \ , \label{paranonint}
\ee
with
\begin{align}
\de \omega&= -\frac{1}{2}\,\big(H_{mnp}\, \Theta^m \wedge \Theta^n \wedge \Theta^p - R^{mnp}\, \tilde{\Theta}_m \wedge \tilde{\Theta}_n \wedge \tilde{\Theta}_p \nonumber \\ & \hspace{1.5cm} +{f_{mn}}^p\, \Theta^m \wedge \Theta^n \wedge \tilde{\Theta}_p - {Q^{mn}}_p\, \tilde{\Theta}_m \wedge \tilde{\Theta}_n \wedge \Theta^p\big) \ .
\label{eq:domegaHRfQ}\end{align}

We can compute the ${\nabla}^{\tt LC}$-bracket, which is the D-bracket of the trivial para-K\"ahler structure on the group manifold $\sfD,$ to measure the compatibility between the trivial para-K\"ahler structure and the left-invariant almost para-Hermitian structure defined above. In order to compute the bracket associated to $\nabla^{\tt LC}$, we proceed as before by computing the connection coefficients of the Levi-Civita connection of the metric $\eta$ in the non-holonomic basis $Z_M$ to get
\begin{align}
{{\mit\Gamma}^i}_{jk} &=\frac{1}{2}\,{f_{kj}}^i \ , \quad {{\mit\Gamma}^{ij}}_k  =\frac{1}{2}\, {Q^{ji}}_k \qquad \mbox{and} \qquad {{\mit\Gamma}_{ij}}^k =\frac{1}{2}\, {f_{ij}}^k \ , \nonumber \\[4pt]
{{\mit\Gamma}_i}^j{}_k &=\frac{1}{2}\, {f_{ik}}^j \ , \quad {{\mit\Gamma}_i}^{jk} = \frac{1}{2}\, {Q^{kj}}_i \qquad \mbox{and} \qquad {{\mit\Gamma}^i}_j{}^k =\frac{1}{2}\, {Q^{ik}}_j \ , \nonumber \\[4pt]
 {\mit\Gamma}^{ijk} &=\frac{1}{2}\,\big(R^{ijk}+ R^{ikj}-R^{jki}\big) \qquad \mbox{and} \qquad {\mit\Gamma}_{ijk} =\frac{1}{2}\,\big(H_{ijk}+H_{ikj}-H_{jki}\big) \ . \nonumber 
\end{align}
Then the ${\nabla}^{\tt LC}$-bracket is given by
$$\llbracket Z_m, Z_n \rrbracket^{{\nabla}^{\tt LC}}=\frac{1}{2}\, {f_{mn}}^k\, Z_k + \frac{3}{2}\, H_{[mnk]}\, \tilde{Z}^k \qquad \mbox{and} \qquad \llbracket \tilde{Z}^m, \tilde{Z}^n \rrbracket^{{\nabla}^{\tt LC}}=\frac{1}{2}\, {Q^{mn}}_k\, \tilde{Z}^k +\frac{3}{2}\,R^{[mnk]}\, Z_k \ ,$$
whereas the D-bracket computed using \eqref{eq:domegaHRfQ} keeps only the integrable part of the ${\nabla}^{\tt LC}$-bracket algebra, as before.
It follows that the fluxes appear here as an incompatibility between the almost para-Hermitian structure $(K,\eta)$ and the trivial para-K\"ahler structure $(K_0,\eta_0)$, i.e. as an obstruction to weak integrability. This implies that such a choice of polarization gives an almost para-Hermitian structure which is not compatible with the trivial para-K\"ahler structure.

\begin{example}
Following \cite{hulled}, a simple case of this construction comes from
a Drinfel'd double $\sfD$ with Lie algebra in the non-involutive
polarization given by $$[T_i, T_j]= H_{ijk}\, \tilde{T}^k \ , \quad [ T_i,\tilde{T}^j]=0 
\qquad \mbox{and} \qquad [\tilde{T}^i, \tilde{T}^j]=0 \ .$$
The left-invariant vector fields $Z_M$ must close to the same algebra and they can be written as $$Z_m= \frac\partial{\partial x^m} -\frac{1}{2}\, H_{mnk}\, x^k \, \frac\partial{\partial\tilde x_n} \qquad \mbox{and} \qquad \tilde{Z}^m=\frac\partial{\partial\tilde x_m}$$ with dual 1-forms 
$$\Theta^m=\de x^m \qquad \mbox{and} \qquad \tilde{\Theta}_m=\de \tilde{x}_m +\frac{1}{2}\, H_{nmk}\, x^k \, \de x^n \ .$$
The almost para-Hermitian structure $(K,\eta)$ can be easily introduced as before. In this example, it is particularly interesting to note the analogy with the deformation of the canonical para-K\"ahler structure on a cotangent bundle $T^*Q$ from Section~\ref{sec:nonassociativity}: We can regard the almost para-Hermitian structure $(K, \eta)$ as a $B$-transformation of the trivial para-K\"ahler structure $(K_0,\eta_0)$ on the group manifold $\sfD$ given by $$B=-\frac{1}{2}\, H_{ijk}\, x^k\, \frac\partial{\partial\tilde x_j} \otimes \de x^i \ .$$ According to the ${\nabla}^{\tt LC}$-bracket, this is a specific case in which the NS--NS $H$-flux naturally arises and is encoded by non-closure of the fundamental 2-form:
\be\nonumber
\de \omega = -\frac12\, H_{ijk}\, \de x^i\wedge\de x^j\wedge\de x^k \ .
\ee
This construction will be useful in the example of the doubled twisted torus that we study in Section~\ref{sec:doubledtorus}.
\end{example}

\subsection{The Drinfel'd Double $T^*\sfG$} \label{specialcase}

In order to draw a parallel between the para-Hermitian geometry of dynamical systems, studied in Sections~\ref{sec:dynamicalpara} and~\ref{sec:nonassociativity}, and of Drinfel'd doubles considered in this section, we describe the special case where the doubled group is $\sfD_\sfG=T^*\sfG$ for a semisimple Lie group $\sfG$.
The cotangent bundle of a $d$-dimensional Lie group $\sfG$ is the best
known example of a Drinfel'd double; it has the structure of a
semi-direct product Lie group, whose Lie algebra $\frd_\frg$ is given by
$\mathfrak{d}_\frg= \mathfrak{g} \ltimes \mathbb{R}^d.$  We will closely
follow~\cite{hulled} in the description of the geometry of $T^*\sfG$
as a doubled Lie group.

The semi-direct product structure $\mathfrak{d}_\frg= \mathfrak{g}
\ltimes \mathbb{R}^d$ means that the Lie algebra $\mathfrak{d}_\frg$
is given by the brackets
\be\nonumber
[T_i,T_j]= {f_{ij}}^k \, T_k \ , \quad [T_i,\tilde{T}^j]= {f_{ik}}^j
\, \tilde{T}^k \qquad \mbox{and} \qquad [\tilde{T}^i, \tilde{T}^j]=0 \ . 
\ee
The natural duality pairing between $\frg$ and $\IR^d$ in this case comes from the
feature that $\frg$ is the fiber of the tangent bundle $T\sfG$ while
$\IR^d$ is the fiber of the cotangent bundle $T^*\sfG$.
Assuming $\sfG$ is semisimple, a $2d$-dimensional matrix
representation of this Lie algebra is given by 
\be 
T_i=
\bigg(\begin{matrix}
t_i & 0 \\
0 & t_i
\end{matrix}\bigg)
\qquad \mbox{and} \qquad \tilde{T}^i =
\bigg(\begin{matrix}
0 & \kappa^{ij}\, t_j \\
0 & 0
\end{matrix}\bigg)
\ , \label{genny}
\ee
where $t_i$ are $d \times d$ matrices obeying the commutation
relations $[t_i, t_j]= {f_{ij}}^k\, t_k$ and
$\kappa_{ij}=\frac{1}{2}\, {f_{il}}^k \, {f_{jk}}^l$ is the
bi-invariant Cartan-Killing metric of $\sfG$. 
Fixing the Iwasawa decomposition of a general element $\gamma\in T^*\sfG$ to be $\gamma=g \, \tilde{g}$, by exponentiating the generators \eqn{genny} we get the matrix representation
\begin{equation}
\gamma =
\bigg(\begin{matrix}
g & 0 \\
0 & g
\end{matrix}\bigg) \,
\bigg(\begin{matrix}
\mathds{1} & \tilde{x} \\
0 & \mathds{1}
\end{matrix}\bigg)
=
\bigg(\begin{matrix}
g & g\,\tilde{x} \\
0 & g
\end{matrix}\bigg)
\ , \label{dectstarg}
\end{equation}
where here $\tilde x=\tilde x_i\,\kappa^{ij}\,t_j$ is valued in the
Lie coalgebra of $\sfG$.

Hence the left-invariant 1-forms are given by
\be\nonumber
\Theta=\gamma^{-1}\, \de \gamma =
\bigg(\begin{matrix}
g^{-1} \, \de g & \de \tilde{x}+ [\tilde{x}, g^{-1}\, \de g] \\
0 & g^{-1}\, \de g
\end{matrix}\bigg)
\ .
\ee
Writing $\lambda=g^{-1}\, \de g= \lambda_i^m\, T_m \, \de x^i,$ we can give the Lie algebra components of $\Theta$ as
\be \label{tgle}
\Theta^M=
\bigg(\begin{matrix}
\Theta^m \\
\tilde{\Theta}_m
\end{matrix}\bigg)
=
\bigg(\begin{matrix}
\lambda^m_i\, \de x^i \\
\de \tilde{x}_m+ {f_{mn}}^k\, \tilde{x}_k \, \lambda^n_i \, \de x^i
\end{matrix}\bigg) \ .
\ee
From this expression we can characterize the adjoint action by confronting \eqn{tgle} with \eqn{leinfo}: Since the Lie group $\tilde{G}=\mathbb{R}^d$ is abelian we have $\tilde{\lambda}^i_m=\delta^i_m$, hence $\tilde{A}_m{}^n=\delta_m{}^n$ and $\tilde b_{mn}= {f_{mn}}^k \, \tilde{x}_k.$
By duality the left-invariant vector fields are then
\be\nonumber
Z_M=
\bigg(\begin{matrix}
Z_m \\
\tilde{Z}^m
\end{matrix}\bigg)
= 
\bigg(\begin{matrix}
(\lambda^{-1})^i_m \, \frac\partial{\partial x^i} - {f_{mn}}^k \, \tilde{x}_k \, \frac\partial{\partial\tilde x_n} \\
\frac\partial{\partial\tilde x_m}
\end{matrix}\bigg)
\ .
\ee

We can now write the para-complex structure as
\be\nonumber
K= Z_m \otimes \Theta^m - \tilde{Z}^m \otimes \tilde{\Theta}_m= \frac\partial{\partial x^i} \otimes \de x^i - \frac\partial{\partial\tilde x_i} \otimes \de \tilde{x}_i + 2\, {f_{mn}}^k\, \tilde{x}_k\, \lambda^n_j\, \frac\partial{\partial\tilde x_m} \otimes \de x^j \ .
\ee
Similarly, the compatible metric with Lorentzian signature is given by
\be \nonumber
\eta= \Theta^m \otimes \tilde{\Theta}_m + \tilde{\Theta}_m \otimes \Theta^m= \lambda^m_i\, \big(\de x^i \otimes \de \tilde{x}_m + \de \tilde{x}_m \otimes \de x^i \big) \ .
\ee
Since the group manifold is now a cotangent bundle, it can be endowed
with the canonical para-K\"ahler structure that now plays the role of
the trivial para-K\"ahler structure defined on $\sfD_\sfG.$ Assuming that
$(x^i, \tilde{x}_i)$ are Darboux coordinates, it follows that
$(T^*\sfG, K, \eta)$ is obtained as a ${\sf GL}(d,
\mathbb{R})_+$-transformation followed by a $B$-transformation of
$(T^*\sfG, K_0, \eta_0)$ defined by the map $\tilde
b_{mn}={f_{mn}}^k\, \tilde{x}_k,$ as we discussed in Section~\ref{manindefor}.
In this case the fundamental 2-form has the expression $$\omega= \tilde{\Theta}_m \wedge \Theta^m = \lambda^m_i\, \de \tilde{x}_m \wedge \de x^i$$
and the Maurer-Cartan structure equations are given by
\begin{align*}
\de \Theta^p = -\frac{1}{2}\,{f_{mn}}^p\, \Theta^m \wedge \Theta^n \qquad \mbox{and} \qquad
\de \tilde{\Theta}_p = -{f_{mp}}^n\, \tilde{\Theta}_n \wedge \Theta^m \ ,
\end{align*}
which can be used to show that $$\de \omega=-\frac{1}{2}\, {f_{mn}}^p\, \tilde{\Theta}_p \wedge \Theta^m \wedge \Theta^n \ ,$$ since now only one distribution becomes non-abelian under the transformation described above. The brackets associated to the two para-Hermitian structures considered here behave exactly as described in the general case: It suffices to put ${Q^{ij}}_k=0$ in all previous general expressions from Section~\ref{manindefor}.

\section{Para-Hermitian Geometry of Doubled Twisted Tori}\label{sec:doubledtorus}

The description of the different para-Hermitian structures on a
cotangent bundle and the para-Hermitian geometry of Drinfel'd doubles
find a common ground in the setting of doubled twisted tori, which
applies the formalism developed thus far to certain parallelizable
string backgrounds which arise as duality twisted compactifications. Beginning
with a general discussion of the underlying para-Hermitian geometry, we shall then specialize to
the specific example of the doubled twisted torus 
which is defined as the quotient of the Drinfel'd double $\sfD_\sfH=T^*\sfH$ of
the Heisenberg group $\sfH$ by a discrete cocompact subgroup. After
working out its para-Hermitian structure explicitly, we will then
demonstrate explicitly how the different polarizations of the doubled twisted torus are obtained by T-duality transformations of its natural para-Hermitian structure as a Drinfel'd double.

\subsection{Torus Bundles and Their Doubles} \label{torusbundle}

In order to apply our previous constructions to some specific
examples, let us begin by reviewing the construction of a particular
class of compact manifolds obtained as quotients of certain
non-compact Lie groups by a discrete cocompact subgroup; this is
sometimes called a \emph{twisted
  torus}~\cite{Dabholkar2002,Hull2005,hulled,edwards:nonpub} and it can be described
as a torus bundle. These arise as the $d$-dimensional internal
compactification manifolds in string theory from standard Kaluza-Klein
reduction on a $d{-}1$-dimensional torus $T^{d-1}$, followed by a
Scherk-Schwarz reduction on an additional circle $S^1$ with a
geometric duality twist around $S^1$ valued in the mapping class group ${\sf GL}(d-1,\IZ)$ of $T^{d-1}$; such a reduction is equivalent to compactification on a $T^{d-1}$-bundle over $S^1$.

We consider a $d$-dimensional Lie group $\sfG$ which is generated by the
Lie algebra $$[t_x, t_a]= {N^b}_a \, t_b \qquad \mbox{and} \qquad
[t_a, t_b]=0 \ ,$$ with $a,b=1, \dots, d-1$ and $\cM=\exp(N)\in{\sf
  GL}(d-1,\IZ)$. A $d$-dimensional matrix representation
of this Lie algebra is given by 
\be
t_x=
\bigg( \begin{matrix}
{N^a}_b & 0 \\
0 & 0
\end{matrix} \bigg)
\qquad \mbox{and} \qquad t_a=
\bigg( \begin{matrix}
0 & E_a \\
0 & 0
\end{matrix} \bigg)
\ , \nonumber
\ee
where $E_a$ is the $d{-}1$-dimensional column vector whose only
non-vanishing entry is $1$ in the $a$-th row. Then in (adapted) local
coordinates $(x,z^a)$ on the group
manifold, any element $g=\exp(x\, t_x+z^a\,t_a) \in\sfG$ can be cast in the form 
\be 
g=
\bigg( \begin{matrix}
\exp(x\,N)^a{}_b & z^a \\
0 & 1
\end{matrix} \bigg)
\ . \nonumber
\ee
Thus we obtain the left-invariant 1-forms $\Theta=g^{-1}\,
\de g = \Theta^m\, t_m$ with components
$$\Theta^x=\de x \qquad \mbox{and} \qquad \Theta^a= \exp(-x\,N)_b{}^a
\, \de z^b \ ,$$ with dual left-invariant vector fields
$$Z_x= \frac\partial{\partial x} \qquad \mbox{and} \qquad Z_a=\exp(x\,N)^b{}_a
\, \frac\partial{\partial z^b} \ . $$

The twisted torus is defined as the quotient $\cT_\sfG=\sfG / \sfG(\IZ)$ by the
equivalence relation $g\sim h\,g$, for all $g\in \sfG$ and $h\in\sfG(\IZ),$ where $\sfG(\IZ)$ is the discrete cocompact subgroup of $\sfG$ whose elements take the form
\be 
h=
\bigg( \begin{matrix}
\exp(\alpha\,N)^a{}_b & \beta^a \\
0 & 1
\end{matrix} \bigg)
\ , \nonumber
\ee
with $\alpha, \beta^a \in \mathbb{Z}.$ Therefore the global structure
of $\cT_\sfG$ is given by the simultaneous identifications 
$$ x \sim x+\alpha \qquad \mbox{and} \qquad z^a \sim
\exp(\alpha\,N)^a{}_b\, z^b + \beta^a \ .$$ From these identifications
it follows that the twisted torus is a torus bundle over a circle, with
local fiber coordinates $(z^a)\in T^{d-1}$ and base coordinate $x\in S^1$,
 whose monodromy is specified by the matrix $\cM=\exp(N)\in {\sf
   GL}(d-1,\IZ).$ The map $x\mapsto\exp(x\,N)$ appearing above is a local section of a
 ${\sf GL}(d-1,\IZ)$-bundle over $S^1$, and the torus bundle may be
 thought of as parameterizing a family of string theories over a circle: For each
 $x\in S^1$, there is a conformal field theory with target space the
 torus $T^{d-1}$.
 Since the equivalence relation defining the quotient is given by the
 left action of the subgroup $\sfG(\IZ),$ the left-invariant 1-forms
 and vector fields on $\sfG$ are globally defined on the compact
 manifold $\cT_\sfG= \sfG/ \sfG(\IZ).$ 

A natural way to construct the double of the twisted torus would be to
start with the Drinfel'd double $\sfD_\sfG=T^*\sfG$ of the Lie group $\sfG$,
as defined in Section~\ref{specialcase}, and then take the quotient
$M_\sfG=T^*\sfG / \sfD_\sfG(\IZ)$ generated by the equivalence relation given from the left action of a discrete cocompact subgroup $\sfD_\sfG(\IZ)$ on $T^*\sfG.$
However, we cannot follow the prescriptions discussed in Section~\ref{specialcase}  to write down the explicit form of the elements of
$T^*\sfG,$ since $\sfG$ is not semisimple in the present case: In the
construction of Section~\ref{specialcase}, the Lie algebra of
$T^*\sfG$ is represented using the inverse of the Cartan-Killing form of $\sfG,$
which is degenerate here. In other words, unless $\sfG$ carries an
invariant metric $\kappa$, we do not
have a general way to represent the Lie algebra of $T^*\sfG$ using $2d
\times 2d$ matrices whose blocks are given by the $d \times d$
matrices $t_x$ and $ t_a$ that realize the Lie algebra of $\sfG.$ On
the other hand, the discrete subgroup $\sfD_\sfG(\IZ)$ and the identifications defining the global structure of the quotient manifold can be explicitly written once a specific form of the monodromy matrix $\cM$ is given.

Nevertheless, we can still make some general remarks concerning the pertinent global
features of the doubled twisted torus: The left-invariant 1-forms and
their dual vector fields on $T^*\sfG$ are still globally defined on the quotient
$M_\sfG=T^*\sfG/ \sfD_\sfG(\IZ),$ since the equivalence relations giving the global
structure of the quotient arise from the left action of the subgroup
$\sfD_\sfG(\IZ).$ Thus the Lie algebra of the left-invariant vector fields
$Z_n$ and $\tilde{Z}^n$ on $M_\sfG$ is the Lie algebra of
$T^*\sfG=\sfG\ltimes\IR^d,$ whose generators $T_n$ and $\tilde T^n$ have the non-vanishing Lie brackets
$$[T_x,T_a]= {N^b}_a \, T_b \ , \quad [ T_x,\tilde{T}^a]= -{N^a}_b\,
\tilde{T}^b \qquad \mbox{and} \qquad [ T_a,\tilde{T}^b]= -{N^b}_a\,
\tilde{T}^x \ . $$
Since the structure constants here are rational, up to isomorphism
there exists a unique discrete cocompact subgroup $\sfD_\sfG(\IZ)$ of
$T^*\sfG$ by Malcev's Theorem. 
By parameterizing a generic group element $\gamma\in T^*\sfG$ as
\be\nonumber
\gamma = \exp\big(\tilde x\,\tilde T^x\big)\,\exp\big(\tilde
z_a\,\tilde T^a\big)\,\exp\big(x\,T_x\big)\,\exp\big(z^a\,T_a\big) \ ,
\ee
this can then be used to describe the doubled twisted torus as a
doubled torus bundle over a pair of circles~\cite{Dabholkar2005}, with
local fibre
coordinates $(z^a,\tilde z^a)\in T^{d-1}\times T^{d-1}$ and base
coordinates $(x,\tilde x)\in S^1\times S^1$.

Following~\cite{DallAgata2007}, from this parametrization of $T^*\sfG$ we obtain the left-invariant
1-forms
$\Theta=\gamma^{-1}\,\de\gamma=\Theta^m\,T_m+\tilde\Theta_m\,\tilde
T^m$ with components
\be\nonumber
\Theta^x = \de x \qquad \mbox{and} \qquad \Theta^a =
\exp(-x\,N)_b{}^a\, \de z^b \ ,
\ee
and
\be\nonumber
\tilde\Theta_x = \de\tilde x - N^b{}_a\, z^a\,\de \tilde z_b \qquad
\mbox{and} \qquad \tilde\Theta_a = \exp(x\,N)^b{}_a\,\de\tilde z_b \ ,
\ee
with dual left-invariant vector fields
\be\nonumber
Z_x= \frac\partial{\partial x} \qquad \mbox{and} \qquad Z_a=\exp(x\,N)^b{}_a
\, \frac\partial{\partial z^b} \ ,
\ee
and
\be\nonumber
\tilde Z^x = \frac\partial{\partial\tilde x} \qquad \mbox{and} \qquad
\tilde Z^a = \exp(-x\,N)^a{}_b\,\Big(\frac\partial{\partial\tilde z_b}
+ N^b{}_c\,z^c\,\frac\partial{\partial\tilde x}\Big) \ .
\ee
The left-invariant para-Hermitian structure from
Section~\ref{specialcase} is given by
$K=Z_n\otimes\Theta^n-\tilde Z^n\otimes \tilde\Theta_n$ and
$\eta=\Theta^n\otimes\tilde\Theta_n + \tilde\Theta_n\otimes\Theta^n$,
whose integrable distributions $L_+$ and $L_-$ are respectively
spanned by $Z_n$ and $\tilde Z^n$, with foliations having leaves
$\sfG$ and $\IR^d$. The fundamental 2-form $\omega =
\tilde\Theta_n\wedge\Theta^n$ yields the geometric $f$-flux
\be\nonumber
\de \omega = -3\,N^b{}_a\, \de x\wedge\de z^a\wedge\de\tilde z_b \ ,
\ee
and the generalized metric $\cH$ from Section~\ref{manindefor} shows
that the $B$-field is zero in this background.

From our previous analysis of Section \ref{manindefor}, we thus
obtain a globally defined para-Hermitian structure on $M_\sfG=T^*\sfG/
\sfD_\sfG(\IZ)$, and both the ${\nabla}^{\tt LC}$-bracket and the D-bracket
are determined by the structure constants of the Lie
algebra. The monodromy matrix therefore explicitly determines the algebra of both brackets, giving
$$\llbracket Z_x, Z_a \rrbracket^{{\nabla}^{\tt LC}}=
\frac{1}{2}\,{N^b}_a \, Z_b \ , \quad \llbracket Z_a, Z_b
\rrbracket^{{\nabla}^{\tt LC}}=0 \qquad \mbox{and} \qquad \llbracket
\tilde{Z}^m, \tilde{Z}^n \rrbracket^{{\nabla}^{\tt LC}}=0 \ ,$$ and
$$\llbracket Z_x, Z_a \rrbracket^{\tt D}={N^b}_a\, Z_b \ , \quad
\llbracket Z_a, Z_b \rrbracket^{\tt D}=0 \qquad \mbox{and} \qquad
\llbracket \tilde{Z}^m, \tilde{Z}^n \rrbracket^{\tt D}=0 \ .$$
In this case, we see how the monodromy matrix gives a geometric
$f$-flux, because we considered a Manin triple polarization for the
Lie algebra of $T^*\sfG,$ thereby leading to Frobenius and weak
integrability of the corresponding eigendistributions. The leaves of their
foliations are respectively given by the twisted torus $\cT_\sfG= \sfG/\sfG(\IZ)$
and the $d$-torus $T^d=\IR^d/\IZ^d$. 

In order to obtain other fluxes, a change of
polarization $\vartheta\in{\sf O}(d,d)(M_\sfG)$ is needed, which also
acts on the monodromy matrix $\cM$ as
\be\nonumber
\cM\longmapsto \cM_\vartheta = \vartheta^{-1}\, \cM\, \vartheta \ .
\ee
When
the transformed monodromy matrix $\cM_\vartheta$ lies in a geometric
subgroup ${\mathsf \Delta}(\IZ)$ of the
T-duality group ${\sf O}(d-1,d-1;\IZ)$, as in the case of the twisted
torus where ${\mathsf \Delta}(\IZ)={\sf GL}(d-1,\IZ)$ is the mapping
class group of the torus fibers $T^{d-1}$, the choice of polarization describes a geometric background,
while if $\cM_\vartheta\in {\sf
  O}(d-1,d-1;\IZ)$ is a monodromy in the Kalb-Ramond field the polarization selects a string
background with NS--NS $H$-flux. If $\cM_\vartheta$ involves a
T-duality, the polarization picks out a T-fold, and from this
perspective globally non-geometric $Q$-flux backgrounds
correspond to submanifolds of the doubled twisted torus $M_\sfG$, because
their monodromy is valued in the subgroup ${\sf O}(d-1,d-1;\IZ)$ of the mapping class group $
{\sf GL}(2(d-1),\IZ)$ of the doubled torus
fibers. On the other hand, locally non-geometric $R$-flux backgrounds are characterized by
monodromies in the full T-duality group ${\sf
  O}\big(d,d;\sfD_\sfG(\IZ)\big)={\sf O}(d,d)\cap {\sf Aut}_{\sfD_\sfG}\big(\sfD_\sfG(\IZ)\big)$ of the
doubled twisted torus, and may be thought of as $T^{d-1}$-bundles over
the dual $S^1$ with coordinate $\tilde x$. Our goal in the following is to understand
these features more intrinsically in the language of para-Hermitian geometry. This will be discussed through the
concrete example of the Heisenberg nilmanifold, which is obtained by
applying the above construction to the $3$-dimensional Heisenberg
group and wherein all of the considerations above can be made explicit. 

\subsection{The Heisenberg Nilmanifold and Its Double}

Our main example will be the $6$-dimensional doubled twisted torus
\cite{hulled}, with its different polarizations, in which T-duality
and fluxes are naturally described in terms of (almost) para-Hermitian
structures. In order to define the global structure of this manifold,
we need to recall the construction of the Heisenberg nilmanifold from
the 3-dimensional Heisenberg group following the general formalism of Section~\ref{torusbundle}.

The 3-dimensional Heisenberg group $\sfH$ has a non-compact group
manifold, with generators $ t_x $, $t_y$ and $t_z$ closing the Lie algebra
\be
[t_x, t_z]= m\,t_y \ , \quad [t_y, t_z]=0 \qquad \mbox{and} \qquad
[t_x,t_y]=0 \ , \label{heial}
\ee
which is not semisimple. It has a $3$-dimensional matrix representation given by  
\be \nonumber
t_x=
\begin{pmatrix}
0 & m & 0 \\
0 & 0 & 0 \\
0 & 0 & 0
\end{pmatrix}
\ , \quad t_y=
\begin{pmatrix}
0 & 0 & 1 \\
0 & 0 & 0 \\
0 & 0 & 0
\end{pmatrix}
\qquad \mbox{and} \qquad t_z=
\begin{pmatrix}
0 & 0 & 0 \\
0 & 0 & 1 \\
0 & 0 & 0
\end{pmatrix}
\ .
\ee
This is an example of the construction from Section \ref{torusbundle}:
The monodromy matrix is given by
\be
\cM = \exp(N) = \bigg( \begin{matrix}
1 & m \\ 0 & 1
\end{matrix} \bigg) \qquad \mbox{with} \quad
N =
\bigg( \begin{matrix}
0 & m \\
0 & 0
\end{matrix} \bigg)
\ , \nonumber
\ee
where $m\in \mathbb{Z}$ and $\cM$ lives in a parabolic conjugacy
class of ${\sf SL}(2,\IZ)$.
The exponential map gives, in local coordinates $(x,y,z)$ on the group
manifold $\sfH$, the general expression for an element $h=\exp(x\,t_x+y\,t_y+z\,t_z)\in \sfH$
given by
\be \nonumber
h=
\begin{pmatrix}
1 & m\,x & y \\
0 & 1 & z \\
0 & 0 & 1
\end{pmatrix}
\ .
\ee
The inverse group element is 
\be \nonumber
h^{-1}=
\begin{pmatrix}
1 & -m\,x & m\,x\,z-y \\
0 & 1 & -z \\
0 & 0 & 1
\end{pmatrix}
\ ,
\ee
thus the left-invariant 1-form $\Theta= h^{-1}\, \de h= \Theta^n \, t_n$ is given by 
\be 
\Theta^x = \de x \ , \quad \Theta^y=\de y - m\,x\, \de z \qquad
\mbox{and} \qquad \Theta^z= \de z \ .
\label{leftinvh}
\ee
From \eqn{leftinvh} we obtain, by duality, the left-invariant vector fields
\be
Z_x= \frac\partial{\partial x} \ , \quad Z_y= \frac\partial{\partial y} \qquad \mbox{and} \qquad
Z_z= \frac\partial{\partial z} + m\,x\, \frac\partial{\partial y} \ ,
\label{lefvec}
\ee
which clearly satisfy the Lie algebra \eqn{heial}. 

Similarly the right-invariant 1-forms $\Xi=\de h\, h^{-1}$ are given by 
$$ \Xi^x = \de x \ , \quad \Xi^y= \de y - m\,z \, \de x \qquad
\mbox{and} \qquad \Xi^z= \de z \ .$$
The dual right-invariant vector fields are then
$$Y_x=\frac\partial{\partial x} + m\,z\, \frac\partial{\partial y} \ , \quad Y_y= \frac\partial{\partial y}= Z_y
\qquad \mbox{and} \qquad Y_z= \frac\partial{\partial z} \ .$$ 

A natural metric on $\sfH$ is defined by using the left-invariant
1-forms to write
\be 
g= \delta_{np} \, \Theta^n \otimes \Theta^p= \de x \otimes \de x +
(\de y - m\,x\, \de z)\otimes  (\de y - m\,x\, \de z) + \de z \otimes
\de z \ , \label{metrlef}
\ee
which can be written in the matrix form
\be \label{eq:metricnilmanifold}
g=
\begin{pmatrix}
1 & 0 & 0 \\
0 & 1 & -m\,x \\
0 & -m\,x & 1+ (m \, x)^2
\end{pmatrix}
 \ ,
\ee
in the basis $\de x, \de y, \de z.$
Note that $Z_y$, $Z_z$, $Y_y$ and $Y_z$ are all Killing vector fields
of the metric $g.$ A similar metric can be introduced by using the right-invariant 1-forms.

The \emph{Heisenberg nilmanifold} $\cT_\sfH$ is the compact
$3$-manifold obtained as the quotient of $\sfH$ with respect to the
cocompact discrete subgroup $\sfH(\IZ)\subset \sfH$ whose elements are
of the general form
\be \nonumber
k =
\begin{pmatrix}
1 & m\,\alpha & \beta \\
0 & 1 & \delta \\
0 & 0 & 1
\end{pmatrix}
\ ,
\ee
with $\alpha, \beta, \delta \in \mathbb{Z}.$ The equivalence relation
is given by the left action of $\sfH(\IZ)$, i.e. $h \sim k\,h,$ and
leads to the simultaneous identifications 
$$ x\sim x + \alpha \ , \quad y \sim y + m\,\alpha\,z + \beta \qquad
\mbox{and} \qquad z \sim z+ \delta \ . $$
The left-invariant 1-forms and vector fields are invariant under the left
action of $\sfH(\IZ)$, hence they descend to the quotient
$\cT_\sfH= \sfH/ \sfH(\IZ).$ The right-invariant 1-forms and vector
fields are, instead, not invariant under the left action of $\sfH(\IZ)$: They transform as
$$ \Xi^x \longmapsto \Xi^x \ , \quad \Xi^y \longmapsto \Xi^y + m\,\alpha\, \Xi^z
- m\,\delta\, \Xi^x \qquad \mbox{and} \qquad \Xi^z \longmapsto \Xi^z \ ,$$
and
$$ Y_x \longmapsto Y_x + m\, \delta\, Y_y \ , \quad Y_y \longmapsto
Y_y \qquad \mbox{and} \qquad Y_z \longmapsto Y_z - m\,\alpha\, Y_y \ ,$$
so that only $\Xi^x$, $\Xi^z$ and $Y_y$ are globally defined.
The metric $g$ from \eqn{metrlef} is also globally defined on
$\cT_\sfH$ and so are the Killing vector fields $Z_y = Y_y$ and $Z_z$,
while $Y_z$ only gives a local solution of the Killing
equations. These vector fields are particularly relevant for the
description of T-dualities on the Heisenberg nilmanifold. 

We have thus
constructed the Heisenberg nilmanifold $\cT_\sfH$ as a compact
$3$-dimensional manifold 
with background metric $g$ given by \eqn{metrlef}, and vanishing
$B$-field inherited from the Heisenberg group. Since $\cT_\sfH$ possesses
globally defined isometries of the metric $g,$ it is possible to apply
the B\"uscher rules to obtain different T-dual backgrounds (see e.g.~\cite{hulled}).
In order to describe the different backgrounds arising from T-duality
transformations, we consider the corresponding doubled twisted torus in different
polarizations, following~\cite{hulled} to develop its para-Hermitian geometry. 
The doubled twisted torus is obtained from the quotient of the
Drinfel'd double $\sfD_\sfH=T^*\sfH$ of the Heisenberg group $\sfH$ with respect
to a discrete cocompact subgroup $\sfD_\sfH(\IZ).$ The Lie algebra of $T^*\sfH=\sfH\ltimes\IR^3$ has
non-vanishing brackets
\be
[T_x, T_z]= m\, T_y \ , \quad [T_x, \tilde{T}^y]= m\, \tilde T^z \qquad
\mbox{and} \qquad [T_z, \tilde{T}^y]= -m\, \tilde{T}^x \ ,  \label{lieth}
\ee
where here the Heisenberg algebra $\mathfrak{h}$ and the abelian
algebra $\mathbb{R}^3$ together with
$\mathfrak{d}_{\mathfrak{h}}=\mathfrak{h}\ltimes \mathbb{R}^3$ form a
Manin triple. It admits a matrix representation in terms of the
matrices $t_n$ from \eqn{heial} given by 
\begin{align}
 T_x &=\bigg(\begin{matrix} t_x & 0 \\ 0 & t_x \end{matrix} \bigg) \ ,
                                           \quad
                                           T_y=\bigg(\begin{matrix}
                                             t_y & 0 \\ 0 &
                                             t_y \end{matrix} \bigg)
                                                            \qquad
                                                            \mbox{and}
                                                            \qquad
                                                            T_z=\bigg(\begin{matrix}
                                                              t_z & 0
                                                              \\ 0 &
                                                              t_z \end{matrix}
                                                                     \bigg)
                                                                     \
                                                                     ,
                                                                     \nonumber\\[4pt]
 \tilde T^x&=\bigg(\begin{matrix} 0 & 0 \\ t_y & 0 \end{matrix} \bigg) \ ,
                                          \quad \tilde T^y
                                          =\bigg(\begin{matrix} 0 &
                                            -t_z  \\  -t_x &
                                            0 \end{matrix}\bigg)
                                                             \qquad
                                                             \mbox{and}
                                                             \qquad
                                                             \tilde T^z
                                                             =\bigg(\begin{matrix}
                                                               0 & t_y
                                                               \\ 0 &
                                                               0 \end{matrix}
                                                                      \bigg)
                                                                      \
                                                                      .
                                                                   \nonumber
\end{align}
Given the specific form of the monodromy matrix here, we are able to
represent the Lie algebra of the Drinfel'd double $T^*\sfH$ despite
the fact that $\sfH$ is not semisimple. Hence we can write down the
identifications defining the global structure of the doubled twisted
torus.  

In local coordinates, any element $\gamma\in T^*\sfH$  may be written as
\be
\gamma=
\begin{pmatrix}
1 & m\,x & y & 0 & 0 & \tilde{z} \\
0 & 1 & z & 0 & 0 & -\tilde{y} \\
0 & 0 & 1 & 0 & 0 & 0 \\
0 & -m\,\tilde{y} & \tilde{x}-m\,z\,\tilde{y} & 1 & m\,x & y+\frac{1}{2}\,m\,\tilde{y}^2 \\
0 & 0 & 0 & 0 & 1 & z \\
0& 0 & 0 & 0 & 0 & 1
\end{pmatrix}
\ee
 and therefore the left-invariant 1-forms are given by the Lie algebra
 components of $\Theta=\gamma^{-1}\,
 \de\gamma=\Theta^n\,T_n+\tilde\Theta_n\,\tilde T^n$ as
\begin{align}
\Theta^x &= \de x \ , \quad \Theta^y=\de y- m\,x\, \de z \qquad
           \mbox{and} \qquad \Theta^z=\de z \ , \nonumber \\[4pt]
\tilde{\Theta}_x &= \de \tilde{x}-m\,z \, \de \tilde{y} \ , \quad
                   \tilde{\Theta}_y = \de \tilde{y} \qquad \mbox{and}
                   \qquad \tilde{\Theta}_z= \de \tilde{z}+ m\,x\, \de \tilde{y} \label{thdeltaleft}
\end{align}
with dual left-invariant vector fields
\begin{align}
 Z_x&=\frac\partial{\partial x} \ , \quad Z_y= \frac\partial{\partial y} \qquad \mbox{and} \qquad
      Z_z= \frac\partial{\partial z} + m\,x\,\frac\partial{\partial y} \ , \label{zdist} \\[4pt]
 \tilde{Z}^x&= \frac\partial{\partial\tilde x} \ , \quad \tilde{Z}^y=
              \frac\partial{\partial\tilde y} + m\,z \,
              \frac\partial{\partial\tilde x} -m\,x\,\frac\partial{\partial\tilde z} \qquad
              \mbox{and} \qquad \tilde{Z}^z= \frac\partial{\partial\tilde z} \ . \label{tilddist}
\end{align}
It follows from \eqn{zdist} that $Z_n$ spans an involutive
distribution $L_+,$ thus it defines a foliation whose leaves are given
by the Heisenberg group $\sfH.$ Similarly \eqn{tilddist} tells us that
$\tilde{Z}^n$ spans an involutive distribution $L_-$ whose foliation
has leaves given by $\mathbb{R}^3$, the fiber of the cotangent bundle $\pi:T^*\sfH\to\sfH.$
Since $T^*\sfH$ is a Drinfel'd double, it is naturally endowed with a
left-invariant para-Hermitian structure with para-complex
structure $K=Z_n\otimes\Theta^n-\tilde Z^n\otimes\tilde\Theta_n$ for which $L_+$ is its $+1$-eigenbundle and $L_-$ is its
$-1$-eigenbundle. The Lorentzian metric is given by $$\eta= \Theta^n
\otimes \tilde{\Theta}_n + \tilde{\Theta}_n \otimes \Theta^n \ , $$
and the fundamental 2-form is $\omega=\tilde\Theta_n\wedge\Theta^n$
with $$\de \omega= m\, \de x \wedge \de z \wedge \de \tilde{y}$$ as discussed in Section \ref{paradrinfeld}.

The identifications giving the global structure of the doubled twisted
torus are obtained via the left action of a discrete cocompact subgroup $\sfD_\sfH(\IZ)$ of $\sfD_\sfH=T^*\sfH.$ Hence the left-invariant para-Hermitian structure of $T^*\sfH$ remains well-defined on the doubled twisted torus $M_\sfH= T^*\sfH / \sfD_\sfH(\IZ).$ 
A generic element $\xi \in \sfD_\sfH(\IZ)$ is given by
\be
\xi=
\begin{pmatrix}
1 & m\,\alpha & \beta & 0 & 0 & \tilde{\delta} \\
0 & 1 & \delta & 0 & 0 & -\tilde{\beta} \\
0 & 0 & 1 & 0 & 0 & 0 \\
0 & -m\,\tilde{\beta} & \tilde{\alpha}-m\,\delta\,\tilde{\beta} & 1 & m\,\alpha & \beta+\frac{1}{2}\,m\,\tilde{\beta}^2 \\
0 & 0 & 0 & 0 & 1 & \delta \\
0& 0 & 0 & 0 & 0 & 1
\end{pmatrix}
 \ , \nonumber
\ee
where $\alpha, \beta, \delta, \tilde{\alpha}, \tilde{\beta},
\tilde{\delta}\in \mathbb{Z}$. The group
action on coordinates induced by the equivalence relation $\gamma\sim
\xi\, \gamma$, which defines the quotient $M_\sfH=T^*\sfH / \sfD_\sfH(\IZ)$, leads to the simultaneous identifications
\begin{align}
x & \sim x+ \alpha \ , \quad y \sim y+m\,\alpha\, z + \beta \qquad
    \mbox{and} \qquad z \sim z+ \delta \ , \nonumber \\[4pt]
\tilde{x}&\sim \tilde{x} + m\,\delta\, \tilde{y} + \tilde{\alpha} \ , \quad
           \tilde{y} \sim \tilde{y}+ \tilde{\beta} \qquad \mbox{and}
           \qquad \tilde{z} \sim \tilde{z} -m\,\alpha\, \tilde{y}
           +\tilde{\delta} \ , \label{iden}
\end{align}
that evidently identify $M_\sfH$ as a $T^2\times T^2$-bundle over
$S^1\times S^1$.
As in the case of the Heisenberg nilmanifold, the left-invariant
1-forms \eqn{thdeltaleft}, together with the left-invariant vector
fields \eqn{zdist} and \eqn{tilddist}, are invariant under the
identifications \eqn{iden}, hence they globally descend to the quotient
$M_\sfH=T^*\sfH/\sfD_\sfH(\IZ).$ This also means that the para-Hermitian
structure $(K, \eta)$ descends to $M_\sfH$, hence the
corresponding eigendistributions $L_+^{\IZ}$ and $L_-^{\IZ}$ of $K$
are both integrable, since their local generators satisfy the
Lie bracket relations \eqn{lieth}; their integral foliations are
characterized, respectively, by the Heisenberg nilmanifold $\cT_\sfH$ and the $3$-torus
$T^3=\mathbb{R}^3/\IZ^3$ as leaves. This is called the \emph{nilmanifold
  polarization} in~\cite{hulled}, where it is shown how to recover the
Heisenberg nilmanifold background from this polarization. We stress
that the Drinfel'd double structure here, in the polarization given by a Manin triple, induces a para-Hermitian structure $(K, \eta)$ on $M_\sfH$. 

\subsection{Polarizations and T-Duality}

We shall now apply our general description of changes of polarization from
Section~\ref{sec:polarization} to the example of the doubled twisted
torus $M_\sfH$. We will use the construction discussed in
Section~\ref{manindefor} to understand how the different polarizations
of the doubled twisted torus fits into the framework of para-Hermitian
geometry. We shall see that the structures arising from each of the
three transformations discussed in Section~\ref{manindefor} may not be
globally defined under the identifications of the coordinates in each
polarization of the doubled twisted torus. As in
\cite{hulled}, this means that the quotient needed to recover the conventional
spacetime background may either be only locally defined or not defined
at all. For this, the left-invariant Born geometry on the Drinfel'd
double $\sfD_\sfH=T^*\sfH$, introduced in Section~\ref{manindefor}, plays an important
role.

Generally, the construction of Section \ref{manindefor} of a
left-invariant para-Hermitian structure is not a change of
polarization with respect to the trivial para-K\"ahler structure on an
even-dimensional manifold. A notable exception will be the T-fold
polarization of the doubled twisted torus $M_\sfH$, which is a local ${\sf
  O}(3,3)$-transformation of the trivial para-K\"ahler structure on $M_\sfH$.
We will now describe the different polarizations of the doubled
twisted torus following \cite{hulled}, giving them a concrete
interpretation in terms of para-Hermitian geometry; another point of
view with some similarities, in the setting of generalized geometry,
can be found in~\cite{Plauschinn}. Our goal is to obtain a description
of the standard T-duality chain~\cite{Shelton2005}
$$
H_{ijk}\stackrel{\vartheta_k}{\longleftrightarrow}
f_{ij}{}^k\stackrel{\vartheta_j}{\longleftrightarrow}
Q^{jk}{}_i\stackrel{\vartheta_i}{\longleftrightarrow} R^{ijk}
$$
relating the different backgrounds which result from performing a (local)
factorized
T-duality transformation $\vartheta_i\in{\sf O}(3,3;\IZ)$ along the $i$-th direction of a given
background related to the Heisenberg nilmanifold $\cT_\sfH$ with geometric
$f$-flux by a change of polarization on the doubled twisted torus
$M_\sfH$, regarded as a para-Hermitian manifold.

\paragraph{Nilmanifold.} 
The nilmanifold polarization is the polarization specified by the Lie
algebra \eqn{lieth} with globally defined vector fields \eqn{zdist} and \eqn{tilddist} spanning, respectively, the two complementary distributions on $TM_\sfH.$ 
Because of the integrability of \eqn{lieth}, no generalized $H$-flux
arises in this polarization: According to the ${\nabla}^{\tt
  LC}$-bracket written in Section~\ref{torusbundle}, the globally
defined para-Hermitian structure induced by the Drinfel'd double
construction and the trivial para-K\"ahler structure are
compatible. This can be seen as a condition implying the presence of
only $f$-flux in this polarization, as we showed in Section~\ref{manindefor}.
As in~\cite{hulled}, the background on the spacetime submanifold $\cT_\sfH$ is
obtained by simply writing down the $\mathbb{R}^3$-invariant metric
from \eqref{metrlef}, or locally \eqref{eq:metricnilmanifold} in the
$x$-coordinates. There is no $B$-field contribution from the
generalized metric $\cH$ in this polarization. 

\paragraph{NS--NS $\boldsymbol H$-Flux.} 
In order to obtain a background with an NS--NS $H$-flux we need a polarization which has a non-involutive distribution. Hence we choose the Lie algebra of the generators to be 
\be
[Z'_x, Z'_z]= m\, \tilde{Z}^{\prime\,y} \ , \quad [Z'_x, Z'_y]= -m\, \tilde{Z}^{\prime\,z} \qquad \mbox{and} \qquad [Z'_z, Z'_y]= m\, \tilde{Z}^{\prime\,x} \ , \label{t3lie}
\ee
with all other brackets vanishing (here we are shuffling around the
generators of the group $\sfD_\sfH=T^*\sfH$). We may regard this choice as an identity transformation of the holonomic basis of $TM_\sfH$ followed by a $B$-transformation. It is important to stress this because we can view this procedure as fixing the holonomic basis on the spacetime by the first transformation, in this case the identity, and then acting on it with other transformations. Thus the background on the spacetime can be recovered only with respect to the holonomic basis obtained after the first transformation. 

Here we shall describe this polarization in a different way than the
description of~\cite{hulled}, which is more naturally in the spirit of flux deformations of para-Hermitian structures. 
For this, we introduce the map $$B=\frac{1}{2}\,\bigg(\Big(-m\,z \,
\frac\partial{\partial\tilde y} +m\,y\,\frac\partial{\partial\tilde
  z}\Big)\otimes \de x +\Big(m\,z\,\frac\partial{\partial\tilde
  x}-m\,x\,\frac\partial{\partial\tilde z}\Big)\otimes \de
y+\Big(-m\,y\,\frac\partial{\partial\tilde
  x}+m\,x\,\frac\partial{\partial\tilde y}\Big)\otimes \de z\bigg) \
. $$
As discussed in Section \ref{anotear}, we then obtain the almost para-Hermitian structure describing this polarization as
$$K'=e^{-B}\, K_0\, e^B \ , \quad \eta=\eta_0 \qquad \mbox{and} \qquad
\omega'=\omega_0+2\,b \ ,$$ where $b=\eta_0\, B,$ which is of the form
\eqn{paranonint}. The new eigenbundles are thus spanned by
\begin{align*}
Z'_x &=\frac\partial{\partial x}
       +\frac{1}{2}\,\Big(m\,y\,\frac\partial{\partial\tilde
       z}-m\,z\,\frac\partial{\partial\tilde y}\Big) \ , \\[4pt]
Z'_y &=\frac\partial{\partial y}
       +\frac{1}{2}\,\Big(m\,z\,\frac\partial{\partial\tilde
       x}-m\,x\,\frac\partial{\partial\tilde z}\Big) \ , \\[4pt] 
Z'_z &=\frac\partial{\partial z}
       +\frac{1}{2}\,\Big(m\,x\,\frac\partial{\partial\tilde
       y}-m\,y\,\frac\partial{\partial\tilde x}\Big) \ , 
\end{align*}
and
\begin{align*}
\tilde{Z}^{\prime\,x} = \frac\partial{\partial\tilde x} \ , \quad
  \tilde{Z}^{\prime\,y} = \frac\partial{\partial\tilde y} \qquad
  \mbox{and} \qquad \tilde{Z}^{\prime\,z}
  =\frac\partial{\partial\tilde z} \ .
\end{align*}
Therefore the vector fields $Z'_n$ span a distribution which is not
involutive, while $\tilde{Z}^{\prime\,n}$ span an involutive
distribution whose integral foliation has leaves given by 
$\mathbb{R}^3$; alternatively, this polarization is another possible
polarization arising from the splitting induced by the projection map
of the cotangent bundle $T^*\sfH$ with the vector fields $\tilde{Z}^{\prime\,n}$
spanning the vertical distribution. In the description of~\cite{hulled}, new identifications are made on the coordinates in the
NS--NS $H$-flux polarization such that the left and right actions of
the abelian subgroup $\mathbb{R}^3$ are globally defined, since they
are generated by the vector fields $\frac\partial{\partial\tilde x_n}$
which are left- and right-invariant. Hence the quotient $M_\sfH /
\mathbb{R}^3$ is well-defined, since the left action of $\mathbb{R}^3$
is globally defined on $M_\sfH$, and gives the spacetime $T^3.$ This
also happens in the present case, with the difference that now the vector fields $Z'_n$ are no longer globally defined.

The ${\sf O}(3,3;\IZ)$-transformation connecting this splitting with the
nilmanifold polarization is given by $$\vartheta= Z_i \otimes
\Theta^{\prime\,i} + \tilde{Z}^i \otimes \tilde{\Theta}'_i \ ,$$ where
the vector fields $Z_i$ and $ \tilde{Z}^i$ are given by \eqn{zdist}
and \eqn{tilddist}. The
dual 1-forms are given explicitly by
\begin{align*}
\Theta^{\prime\,x} =\de x \ , \quad \Theta^{\prime\,y} =\de y \qquad
  \mbox{and} \qquad \Theta^{\prime\,z} =\de z \ ,
\end{align*}
and
\begin{align*}
\tilde{\Theta}'_x &= \de \tilde{x} +\frac{1}{2}\,\big(m\,y\, \de z
                    -m\,x\, \de y\big) \ , \\[4pt]
\tilde{\Theta}'_y &=\de \tilde{y} +\frac{1}{2}\,\big(m\,z\, \de x -m\,x\,
                    \de z\big) \ , \\[4pt]
\tilde{\Theta}'_z &=\de \tilde{z} +\frac{1}{2}\,\big(m\,x\, \de y -
                    m\,y\, \de x\big) \ .
\end{align*}
The 1-forms $\Theta^{\prime\,n}$ are globally defined on the doubled
twisted torus $M_\sfH$, while the 1-forms $\tilde{\Theta}'_n$ are not
as a consequence of our choice for the distributions. 
Thus the almost para-Hermitian structure is given by 
$$K'= \vartheta^{-1}\, K\, \vartheta= Z'_i \otimes \Theta^{\prime\,i}-
\tilde{Z}'_i \otimes \tilde{\Theta}^{\prime\,i} \ ,$$ where $K$ is the
para-complex structure of the nilmanifold polarization,
and  $$\omega'= \vartheta^{\rm t} \, \omega \, \vartheta=
\tilde{\Theta}'_i \wedge \Theta^{\prime\,i} \ .$$ 
As demonstrated in Section~\ref{anotear}, in this polarization there are non-vanishing ${\nabla}^{\tt LC}$-brackets 
$$\llbracket Z'_x, Z'_y \rrbracket^{{\nabla}^{\tt LC}}= -\frac{3}{2}\,
m\, \tilde{Z}^{\prime\,z} \ , \quad \llbracket Z'_x, Z'_z
\rrbracket^{{\nabla}^{\tt LC}}= \frac{3}{2}\, m\, \tilde{Z}^{\prime\,y} \qquad
\mbox{and} \qquad \llbracket Z'_z, Z'_y \rrbracket^{{\nabla}^{\tt LC}}=
\frac{3}{2}\, m\, \tilde{Z}^{\prime\,x} \ ,$$
showing that the trivial para-K\"ahler structure on $M_\sfH$ and the
almost para-Hermitian structure $(K', \eta)$ are not compatible; the violation of weak integrability gives the $H$-flux.

Finally, in order to recover the physical background fields on $T^3,$
we write down the generalized metric
$\mathcal{H}=\delta_{np}\,\Theta^{\prime\,n} \otimes \Theta^{\prime\,p} +
\delta^{np}\,\tilde{\Theta}'_n \otimes \tilde{\Theta}'_p.$ In the
coordinates $(x, \tilde{x})$,
it takes the expected form
\begin{equation}
\mathcal{H}=
\bigg( \begin{matrix}
g-b\,g^{-1}\,b & b\,g^{-1}\\
 -g^{-1}\,b & g^{-1}
\end{matrix} \bigg)
\ , \nonumber
\end{equation}
where the background $(g,b)$ depends only on the coordinates $x$
and is given by
\be
g= \mathds{1} \qquad \mbox{and} \qquad b=\eta_0 \, B \ . \nonumber
\ee

\paragraph{T-Fold.} 
The T-fold polarization is given by the choice of the Manin triple on $T^*\sfH$ with
\be
[Z''_x, \tilde{Z}^{\prime\prime\,z}]= m\, Z''_y \ , \quad [Z''_x, \tilde{Z}^{\prime\prime\,y}]= -m\, {Z}^{\prime\prime}_z \qquad \mbox{and} \qquad [\tilde{Z}^{\prime\prime\,z}, \tilde{Z}^{\prime\prime\,y}]= m\, \tilde{Z}^{\prime\prime\,x} \ , \label{tfoldlie}
\ee
which is well defined on $M_\sfH.$ This polarization can be obtained from the previous ones by applying the procedure described in Section~\ref{manindefor}.
In this polarization both distributions are integrable, hence there is no generalized $H$-flux arising from the ${\nabla}^{\tt LC}$-bracket, which replicates the Lie bracket written in \eqref{tfoldlie} as discussed in Section~\ref{torusbundle}.

In order to recover the spacetime background, let us describe the transformations needed to obtain this polarization as a deformation of the trivial para-K\"ahler structure on $M_\sfH,$ along the lines explained in Section \ref{manindefor}.
We first deform the eigendistributions of the trivial para-K\"ahler structure into the two integrable distributions spanned by the vector fields
\begin{align}
X_x&= \frac\partial{\partial x} \ , \quad X_y = \frac\partial{\partial y} \qquad \mbox{and} \qquad X_z =\frac\partial{\partial z} \ ,\nonumber \\[4pt]
\tilde{X}^x&= \frac\partial{\partial\tilde x} \ , \quad \tilde{X}^y =\frac\partial{\partial\tilde y}+m\,\tilde{z}\,\frac\partial{\partial\tilde x} \qquad \mbox{and} \qquad \tilde{X}^z =\frac\partial{\partial\tilde z} \ , \nonumber
\end{align}
via the transformations 
\begin{equation}
\lambda^{-1} =
\begin{pmatrix}
1 & 0 & 0 \\
0 & 1 & 0 \\
0 & 0 & 1 
\end{pmatrix}
\qquad \mbox{and} \qquad \tilde{\lambda}^{-1} =
\begin{pmatrix}
1 & 0 & 0 \\
m\,\tilde{z} & 1 & 0 \\
0 & 0 & 1
\end{pmatrix}
\ , \nonumber
\end{equation}
acting respectively on $\frac\partial{\partial x^n} $ and $\frac\partial {\partial\tilde{x}_n}.$ The Lie algebra of these vector fields is given by the single non-vanishing Lie bracket $[X^z, X^y]=m\, X^x.$ The dual 1-forms are
\begin{align}
W_x & = \de x \ , \quad W_y = \de y \qquad \mbox{and} \qquad W_z = \de z \ , \nonumber \\[4pt]
\tilde{W}^x&= \de \tilde{x}-m\,\tilde{z}\,\de \tilde{y} \ , \quad \tilde{W}_y = \de \tilde{y} \qquad \mbox{and} \qquad \tilde{W}_z =\de \tilde{z} \ . \nonumber
\end{align}
This fixes the bases for the tangent bundle on our spacetime and its ``dual''. The generalized metric will be expressed in these bases in order to recover the spacetime background. This transformation does not change the almost para-complex structure but gives another off-diagonal expression of the Lorentzian metric: The metric $\eta_0$ becomes $\eta= W^n\otimes \tilde{W}_n + \tilde{W}_n \otimes W^n$, and the two distributions spanned by $X_n$ and $\tilde X^n$ are maximally isotropic with respect to this metric. Similarly, the fundamental 2-form becomes $\omega_{\lambda,\tilde\lambda} =\tilde{W}_n \wedge W^n$ which is no longer closed. In this case, only the 1-forms $\tilde{W}_n$ (and their dual vector fields) are globally defined under the identifications of the coordinates in the T-fold polarization described in \cite{hulled}.

We can then obtain the T-fold polarization as a $\beta$-transformation twisting the distribution spanned by $\tilde{X}^n$ with 
$$\beta= -m\,x\, \frac\partial{\partial z} \otimes \tilde{W}^y + m\,x\, \frac\partial{\partial y} \otimes \tilde{W}^z \ .$$
This $\beta$-transformation does not satisfy the Maurer-Cartan equation \eqn{maucar}, hence the fundamental 2-form becomes $\omega=\omega_{\lambda,\tilde\lambda}+2 \tilde{b},$ where $\tilde{b}=\eta\, \beta.$ It is however an ${\sf O}(3,3)(M_\sfH)$-tranformation, hence the metric $\eta$ is preserved.
The globally defined integrable distributions are finally spanned by the vector fields
\begin{align}
Z''_x&= \frac\partial{\partial x} \ , \quad Z''_y = \frac\partial{\partial y} \qquad \mbox{and} \qquad Z''_z =\frac\partial{\partial z} \ ,\nonumber \\[4pt]
\tilde{Z}^{\prime\prime\,x}&= \frac\partial{\partial\tilde x} \ , \quad \tilde{Z}^{\prime\prime\,y} =\frac\partial{\partial\tilde y}+m\,\tilde{z}\,\frac\partial{\partial\tilde x}-m\,x\,\frac\partial{\partial z} \qquad \mbox{and} \qquad \tilde{Z}^{\prime\prime\,z} =\frac\partial{\partial\tilde z}+m\,x\,\frac\partial{\partial y} \ , \nonumber
\end{align}
which close the Lie algebra \eqn{tfoldlie} and have dual 1-forms
\begin{align}
\Theta''_x & = \de x \ , \quad \Theta''_y = \de y-m\,x\, \de \tilde{z} \qquad \mbox{and} \qquad \Theta''_z = \de z+m\,x\, \de \tilde{y} \ , \nonumber \\[4pt]
\tilde{\Theta}^{\prime\prime\,x}&= \de \tilde{x}-m\,\tilde{z}\,\de \tilde{y} \ , \quad \tilde{\Theta}''_y = \de \tilde{y} \qquad \mbox{and} \qquad \tilde{\Theta}''_z =\de \tilde{z} \ . \nonumber
\end{align}
The left-invariant para-Hermitian structure has the form given in Section \ref{paraherdri}. As shown in Section~\ref{manindefor}, the ${\nabla}^{\tt LC}$-bracket is non-vanishing and gives
$$\llbracket \tilde{Z}^{\prime\prime\,z}, \tilde{Z}^{\prime\prime\,y} \rrbracket^{{\nabla}^{\tt LC}}= \frac{3}{2}\,m\, \tilde{Z}^{\prime\prime\,x} \ ,$$ which demonstrates the presence of a $Q$-flux in this polarization.

We can finally write down the generalized metric $$\mathcal{H}=\delta_{np}\,\Theta^{\prime\prime\,n} \otimes \Theta^{\prime\prime\, p} + \delta^{np}\,\tilde{\Theta}''_n \otimes \tilde{\Theta}''_p \ ,$$ and express it in the basis $ W_n,$ $\tilde{W}^n$ where it takes the form
\begin{equation}
\mathcal{H} =
\begin{pmatrix}
1 & 0 & 0 & 0 & 0 & 0 \\
0 & 1 & 0 & 0 & 0 & -m\,x \\
0 & 0 & 1 & 0 & m\,x & 0 \\
0 & 0 & 0 & 1 & 0 & 0 \\
0 & 0 & m\,x & 0 & 1+(m\,x)^2 & 0 \\
0 & -m\,x & 0 & 0 & 0 & 1+(m\,x)^2
\end{pmatrix}
 \ . \nonumber
\end{equation}
We therefore read off the background
\begin{equation}
g =\frac{1}{1+(m\,x)^2}
\begin{pmatrix}
1+(m\,x)^2 & 0 & 0 \\
0 & 1 & 0 \\
0 & 0 & 1 \\
\end{pmatrix} \qquad \mbox{and} \qquad
b =\frac{m\,x}{1+(m\,x)^2}
\begin{pmatrix}
0 & 0 & 0 \\
0 & 0 & -1 \\
0 & 1 & 0
\end{pmatrix}
\nonumber
\end{equation}
as in \cite{hulled}. We stress that there is a $Q$-flux on the doubled space in this polarization and the $\beta$-twist of the two distributions spanned by $X^n$ and $\tilde{X}_n$ gives a non-vanishing $B$-field on the spacetime submanifold; the global non-geometry is entirely manifested by the feature that the background $(g,b)$ is only locally well-defined in this polarization.

\paragraph{Locally Non-Geometric $\boldsymbol R$-Flux.} 
The $R$-flux polarization can be obtained from the $H$-flux polarization
by exchanging the roles of the Lie algebras $\frh$ and $\IR^3$ in the
Manin triple associated to the Drinfel'd double $T^*\sfH$.
Despite the lack of even a local geometry for the $R$-flux polarization \cite{hulled,Plauschinn}, in the present framework we can follow the discussion of Section~\ref{anotear} to choose a polarization such that
$$[\tilde{Z}^x, \tilde{Z}^z]= m \, Z_y \ , \quad [\tilde{Z}^x, \tilde{Z}^y]= -m \, Z_z \qquad \mbox{and} \qquad[\tilde{Z}^z, \tilde{Z}^y]= m \, Z_x \ .$$ 
The almost para-Hermitian structure with eigendistributions closing this Lie algebra can be obtained as a $B$-transformation of the trivial para-K\"ahler structure on $M_\sfH,$ and takes the form \eqn{paranonint}.
Thus as shown in Section \ref{anotear}, this induces a non-vanishing generalized $R$-flux from the ${\nabla}^{\tt LC}$-bracket
$$\llbracket \tilde{Z}^x, \tilde{Z}^y \rrbracket^{{\nabla}^{\tt LC}}=
-\frac{3}{2}\,m \, Z_z \ , \quad \llbracket \tilde{Z}^x, \tilde{Z}^z
\rrbracket^{{\nabla}^{\tt LC}}= \frac{3}{2}\,m \, Z_y \qquad
\mbox{and} \qquad \llbracket \tilde{Z}^z, \tilde{Z}^y
\rrbracket^{{\nabla}^{\tt LC}}= \frac{3}{2}\,m \, Z_x \ .$$
As discussed in \cite{hulled,ReidEdwards2009}, the generalized metric
$\mathcal{H}$ in this polarization depends on the coordinates
$(\tilde{x},\tilde y,\tilde z)$, hence it is not possible to recover
the conventional spacetime background with any quotient.

\section*{Acknowledgments}

We thank Laurent Freidel, Felix Rudolph and David Svoboda for helpful discussions.
This work was supported in part by the Action MP1405 ``Quantum Structure of Spacetime'' funded by the European
Cooperation in Science and Technology (COST). The work of V.E.M. was
funded by the Doctoral Training Grant ST/R504774/1 from the UK Science and Technology
Facilities Council (STFC). The work of R.J.S. was supported by the STFC
Consolidated Grant ST/P000363/1.

\end{document}